\PassOptionsToPackage{final}{graphicx}
\PassOptionsToPackage{colorlinks,linkcolor={blue},citecolor={blue},urlcolor={red},breaklinks=true,final}{hyperref}

\documentclass[a4paper,UKenglish,cleveref,autoref,final]{lipics-v2021}

\usepackage{etex}

\usepackage{ifdraft}                            %% For access to the draft/final option of the documentclass
\ifdraft
{
  \usepackage[notcite,notref]{showkeys}
  
}{}

\usepackage{crossreftools}
\usepackage{microtype} 

\usepackage{yhmath}                             %% For \wideparen

\usepackage{proof}                              %% Drawing logical rules
\usepackage{amsmath,amsthm}
\usepackage{amssymb}                            %% Extra math symbols and macro
\usepackage{stmaryrd}                  			    %% More symbols
\usepackage{slashed}                            %% Сrossed symbols
\usepackage{gensymb}
\usepackage{accents} 

\usepackage{bbold}

\usepackage{scalerel} 

\usepackage{mathtools}

\usepackage{graphicx}
\usepackage{mathrsfs}
\usepackage{subcaption}

\usepackage{tikz}
\usepackage{tkz-euclide}
\usepackage{tikz-cd}

\usetikzlibrary{arrows}

\tikzset{
  cong/.style={draw=none,edge node={node [sloped, allow upside down, auto=false]{$\cong$}}}, 
  iso/.style={draw=none,every to/.append style={edge node={node [sloped, allow upside down, auto=false]{$\cong$}}}}
}

\usetikzlibrary{shapes,arrows}

\tikzstyle{block} = [rectangle, text width=12em, text centered, rounded corners=.75mm, minimum height=2em]
\tikzstyle{line} = [draw, -latex']

\usepackage{relsize}

\usetikzlibrary{
  fit,%
  calc,%
  arrows,%
  arrows.meta,%
  intersections,%
  shapes.misc,% wg. rounded rectangle
  shapes.arrows,%
  patterns,%
  automata,%
  chains,%
  matrix,%
  positioning,% wg. " of "
  scopes,%
  decorations.markings,%
  decorations.pathmorphing,%
  external
}

\tikzset{
commutative diagrams/.cd,
arrow style=tikz,
diagrams={>=stealth},
row sep=large, 
column sep = huge
}

\tikzset{shiftarr/.style={
        rounded corners,%
        to path={--([#1]\tikztostart.center)
                     -- ([#1]\tikztotarget.center) \tikztonodes
                     -- (\tikztotarget)},
}}

\tikzset{cong/.style={draw=none,edge node={node [sloped, allow upside down, auto=false]{$\cong$}}},
         Isom/.style={draw=none,every to/.append style={edge node={node [sloped, allow upside down, auto=false]{$\cong$}}}}}

\usepackage{optparams}                          %% Multiple optional paramenters
\usepackage{xspace}

\usepackage{dsfont}

\ifdraft{
 \usepackage{todos}
 \usepackage[layout=footnote,draft]{fixme}
}{
 \newcommand{\todo}[1]{}
 \usepackage[layout=footnote,final]{fixme}
}

\usepackage{needspace}

\renewcommand{\labelitemi}{$\vcenter{\hbox{\rule{1.2ex}{1pt}}}$}

  {\gdef\scalefactor{#1}\begin{center}\proofSkipAmount \leavevmode}%
  {\scalebox{\scalefactor}{\DisplayProof}\proofSkipAmount \end{center} }

\FXRegisterAuthor{sg}{asg}{SG} %Sergey

\providecommand{\catname}{\mathbf} 
\providecommand{\clsname}{\mathcal}
\providecommand{\oname}[1]{\operatorname{\mathsf{#1}}}

\def\defcatname#1{\expandafter\def\csname B#1\endcsname{\catname{#1}}}
\def\defcatnames#1{\ifx#1\defcatnames\else\defcatname#1\expandafter\defcatnames\fi}
\defcatnames ABCDEFGHIJKLMNOPQRSTUVWXYZ\defcatnames

\def\defclsname#1{\expandafter\def\csname C#1\endcsname{\clsname{#1}}}
\def\defclsnames#1{\ifx#1\defclsnames\else\defclsname#1\expandafter\defclsnames\fi}
\defclsnames ABCDEFGHIJKLMNOPQRSTUVWXYZ\defclsnames

\def\defbbname#1{\expandafter\def\csname BB#1\endcsname{\mathbb{#1}}}
\def\defbbnames#1{\ifx#1\defbbnames\else\defbbname#1\expandafter\defbbnames\fi}
\defbbnames ABCDEFGHIJKLMNOPQRSTUVWXYZ\defbbnames

\def\Set{\catname{Set}}

\def\dCpo{\catname{dCpo}}
\def\Top{\catname{Top}}

\providecommand{\argument}{\operatorname{-\!-}}

\providecommand{\wave}[1]{\widetilde{#1}}		               %

\DeclareOldFontCommand{\bf}{\normalfont\bfseries}{\mathbf}
\providecommand{\Id}{\operatorname{Id}}

\providecommand{\Hom}{\mathsf{Hom}}
\providecommand{\id}{\mathsf{id}}

\providecommand{\iso}{\mathbin{\cong}}

\providecommand{\bang}{\operatorname!}				             %

\providecommand{\ito}{\hookrightarrow}				             %
\providecommand{\mto}{\mapsto}
\providecommand{\xto}[1]{\,\xrightarrow{#1}\,}

\providecommand{\dar}{\kern-1.2pt\operatorname{\downarrow}}	
\providecommand{\uar}{\kern-1.2pt\operatorname{\uparrow}}

\providecommand{\bigor}{\bigvee}

\providecommand{\impl}{\Rightarrow}

\providecommand{\appr}{\sqsubseteq}

\providecommand{\fst}{\oname{fst}}
\providecommand{\snd}{\oname{snd}}

\providecommand{\brks}[1]{\langle #1\rangle}

\providecommand{\inl}{\oname{inl}}
\providecommand{\inr}{\oname{inr}}

\DeclareSymbolFont{Symbols}{OMS}{cmsy}{m}{n}
\DeclareMathSymbol{\iobj}{\mathord}{Symbols}{"3B}
\providecommand{\curry}{\oname{curry}}
\providecommand{\uncurry}{\oname{uncurry}}
\providecommand{\ev}{\oname{ev}}

\usepackage{stmaryrd}

\providecommand{\comma}{,\operatorname{}\linebreak[1]}		 %
\providecommand{\dash}{\nobreakdash-\hspace{0pt}}		       %

\providecommand{\by}[1]{\text{/\!\!/~#1}}			             %
\providecommand{\pacman}[1]{}					                     %

\newcommand{\undefine}[1]{\let #1\relax}					                       %

\providecommand{\mone}{{\text{\kern.5pt\rmfamily-}\mathsf{\kern-.5pt1}}}

\makeatletter
\@ifpackageloaded{enumitem}{}{
  \usepackage[loadonly]{enumitem}                          %
}                                                          %
\makeatother

\newlist{citemize}{itemize}{1}
\setlist[citemize]{label=\labelitemi,wide} 

\newlist{cenumerate}{enumerate}{1}
\setlist[cenumerate,1]{label=\arabic*.~,ref={\arabic*},wide} 
\makeatletter
\def\mfix#1{\oname{#1}\@ifnextchar\bgroup\@mfix{}}	       %
\def\@mfix#1{#1\@ifnextchar\bgroup\mfix{}}			           %
\makeatother

\providecommand{\case}[3]{\mfix{case}{\mathbin{}#1}{of}{#2}{\kern-1pt;}{\mathbin{}#3}}

\usepackage{tikz,pgf,pgfmath}
\usepackage{relsize}

\usetikzlibrary{shapes,arrows}
\usetikzlibrary{calc}
\usetikzlibrary{fit,decorations.pathreplacing,arrows.meta,positioning}
\usetikzlibrary{decorations.pathreplacing}

\usetikzlibrary{decorations.pathreplacing}
\tikzset{%
  show curve controls/.style={
    postaction={
      decoration={
        show path construction,
        curveto code={
          \draw [blue]
            (\tikzinputsegmentfirst) -- (\tikzinputsegmentsupporta)
            (\tikzinputsegmentlast) -- (\tikzinputsegmentsupportb);
          \fill [red, opacity=0.5]
            (\tikzinputsegmentsupporta) circle [radius=.5ex]
            (\tikzinputsegmentsupportb) circle [radius=.5ex];
        }
      },
      decorate
}}}

\pgfkeys{/pgf/.cd,
  east ports/.initial={.5},
  west ports/.initial={.5},
}

\makeatletter

\pgfdeclareshape{gprof}{

  \inheritsavedanchors[from=rectangle]
  \inheritanchorborder[from=rectangle]

  \foreach \p in {north, north west, north east, center,
      west, east, mid, mid west, mid east, base, base west,
      base east, south, south west, south east}
  {\inheritanchor[from=rectangle]{\p}}

  \backgroundpath{

    \northeast \pgf@xa=\pgf@x \pgf@ya=\pgf@y
    \southwest \pgf@xb=\pgf@x \pgf@yb=\pgf@y

    \ifgbar@left
      \pgfpathmoveto{\pgfpoint{\pgf@xa}{\pgf@ya}}
      \pgfpathlineto{\pgfpoint{\pgf@xa}{\pgf@yb}}
      \pgfpathlineto{\pgfpoint{0.5*\pgf@xa+0.5*\pgf@xb}{\pgf@yb}}
      \pgfpathlineto{\pgfpoint{0.5*\pgf@xa+0.5*\pgf@xb}{\pgf@ya}}
    \else
      \pgfpathmoveto{\pgfpoint{0.5*\pgf@xa+0.5*\pgf@xb}{\pgf@ya}}
      \pgfpathlineto{\pgfpoint{0.5*\pgf@xa+0.5*\pgf@xb}{\pgf@yb}}
      \pgfpathlineto{\pgfpoint{\pgf@xb}{\pgf@yb}}
      \pgfpathlineto{\pgfpoint{\pgf@xb}{\pgf@ya}}
    \fi
    \pgfpathclose
    \pgfusepath{clip}
    \pgfsetcornersarced{
      \pgfpoint{.25*min(\pgf@xa-\pgf@xb,\pgf@ya-\pgf@yb)}{
                .25*min(\pgf@xa-\pgf@xb,\pgf@ya-\pgf@yb)}
    }
    \pgfpathrectanglecorners{\southwest}{\northeast}
    \pgfusepath{fill}
  }
}

\pgfdeclareshape{gbox}{

  \inheritsavedanchors[from=rectangle]
  \inheritanchorborder[from=rectangle]

  \foreach \p in {north, north west, north east, center,
    west, east, mid, mid west, mid east, base, base west,
    base east, south, south west, south east}
  {\inheritanchor[from=rectangle]{\p}}

  \inheritbackgroundpath[from=rectangle]
  \inheritbeforebackgroundpath[from=rectangle]
  \inheritbehindforegroundpath[from=rectangle]
  \inheritforegroundpath[from=rectangle]
  \inheritbeforeforegroundpath[from=rectangle]%

  \savedmacro\getdboxparameters{%
    \pgfmathsetlengthmacro\pgf@right@top@offset{\pgfkeysvalueof{/tikz/right top offset}}%
    \addtosavedmacro\pgf@right@top@offset%
    \pgfmathsetlengthmacro\pgf@right@mid@offset{\pgfkeysvalueof{/tikz/right mid offset}}%
    \addtosavedmacro\pgf@right@mid@offset%
    \pgfmathsetlengthmacro\pgf@left@bot@offset{\pgfkeysvalueof{/tikz/left bot offset}}%
    \addtosavedmacro\pgf@left@bot@offset%
    \pgfmathsetlengthmacro\pgf@left@mid@offset{\pgfkeysvalueof{/tikz/left mid offset}}%
    \addtosavedmacro\pgf@left@mid@offset%
  }%

  \anchor{geast}{
    \getdboxparameters%
    \pgf@sh@lib@gbox@geastanchor{\pgf@right@top@offset}{\pgf@right@mid@offset}
  }

  \anchor{gwest}{
    \getdboxparameters%
    \pgf@sh@lib@gbox@gwestanchor{\pgf@left@bot@offset}{\pgf@left@mid@offset}
  }

  \anchor{uwest}{
    \getdboxparameters%
    \pgf@sh@lib@gbox@geastanchor{-\pgf@left@bot@offset}{-\pgf@left@mid@offset}
    \pgf@ya=\pgf@y
    \southwest
    \pgf@y=\pgf@ya
  }

  \anchor{ueast}{
    \getdboxparameters%
    \pgf@sh@lib@gbox@gwestanchor{-\pgf@right@top@offset}{-\pgf@right@mid@offset}
    \pgf@ya=\pgf@y
    \northeast
    \pgf@y=\pgf@ya
  }

  \backgroundpath{

    \def\pgf@left@bot@offset{\pgfkeysvalueof{/tikz/left bot offset}}%
    \def\pgf@right@top@offset{\pgfkeysvalueof{/tikz/right top offset}}%
    \def\pgf@left@mid@offset{\pgfkeysvalueof{/tikz/left mid offset}}%
    \def\pgf@right@mid@offset{\pgfkeysvalueof{/tikz/right mid offset}}%
    \def\pgf@bar@thickness{\pgfkeysvalueof{/tikz/bar thickness}}%

    \ifx\tikz@fillcolor\pgfutil@empty\relax\else
    \pgfpathrectanglecorners{\southwest}{\northeast}%
    \pgfusepath{fill}
    \fi

    \southwest \pgf@xa=\pgf@x \pgf@ya=\pgf@y
    \northeast \pgf@xb=\pgf@x \pgf@yb=\pgf@y

    \ifx\tikz@strokecolor\pgfutil@empty\pgfsetfillcolor{black}\else
    \pgfsetfillcolor{\tikz@strokecolor}
    \fi

    \ifhide@bars\relax\else

    \pgfmathparse{int(0.5*\pgf@yb-\pgf@right@top@offset - 0.5*\pgf@ya - \pgf@right@mid@offset)}
    \ifnum\pgfmathresult<0\relax\else

    \fi

    \southwest \pgf@xa=\pgf@x \pgf@ya=\pgf@y
    \northeast \pgf@xb=\pgf@x \pgf@yb=\pgf@y

    \pgfmathparse{int(0.5*\pgf@yb-\pgf@left@mid@offset - 0.5*\pgf@ya - \pgf@left@bot@offset)}
    \ifnum\pgfmathresult<0\relax\else

      \fi
    \fi

   \pgfpathrectanglecorners{\southwest}{\northeast}%
   \pgfusepath{stroke}

  }%
}%

\def\pgf@sh@lib@gbox@gwestanchor#1#2{%
  \southwest%
  \pgf@ya=\pgf@y
  \northeast
  \pgf@yb=.5\pgf@y
  \advance\pgf@yb by .5\pgf@ya
  \advance\pgf@ya by #1
  \advance\pgf@yb by #2
  \southwest%
  \pgf@y=.5\pgf@ya
  \advance\pgf@y by .5\pgf@yb
}

\def\pgf@sh@lib@gbox@geastanchor#1#2{
  \northeast
  \pgf@ya=\pgf@y
  \southwest
  \pgf@yb=.5\pgf@y
  \advance\pgf@yb by .5\pgf@ya
  \advance\pgf@ya by #1
  \advance\pgf@yb by  #2
  \northeast%
  \pgf@y=.5\pgf@ya
  \advance\pgf@y by .5\pgf@yb
}

\pgfdeclareshape{box}{

  \inheritsavedanchors[from=rectangle]
  \inheritanchorborder[from=rectangle]

  \foreach \p in {north, north west, north east, center,
    west, east, mid, mid west, mid east, base, base west,
    base east, south, south west, south east}
  {\inheritanchor[from=rectangle]{\p}}

  \inheritbackgroundpath[from=rectangle]
  \inheritbeforebackgroundpath[from=rectangle]
  \inheritbehindforegroundpath[from=rectangle]
  \inheritforegroundpath[from=rectangle]
  \inheritbeforeforegroundpath[from=rectangle]%
  \inheritbehindbackgroundpath[from=rectangle]

  \savedmacro\box@eastports{%
    \def\box@eastports{\pgfkeysvalueof{/pgf/east ports},0}
  }

  \savedmacro\box@westports{%
    \def\box@westports{\pgfkeysvalueof{/pgf/west ports},0}
  }

  \backgroundpath{
    \pgfpathrectanglecorners{\southwest}{\northeast}%
  }%

  \pgfutil@g@addto@macro\pgf@sh@s@box{%
    \pgfmathparse{dim(\box@eastports)}\pgfmathresult
    \let\portnum=\pgfmathresult

    \pgfmathloop%
    \ifnum\pgfmathcounter=\portnum\relax%
    \else%
      \pgfutil@ifundefined{pgf@anchor@box@o\pgfmathcounter}{%
        \expandafter\xdef\csname pgf@anchor@box@o\pgfmathcounter\endcsname{%
          \noexpand\pgf@sh@lib@box@eastanchor{\pgfmathcounter}%
        }%
      }{}%
    \repeatpgfmathloop%

    \pgfmathparse{dim({\box@westports})}\pgfmathresult
    \let\portnum=\pgfmathresult

    \pgfmathloop%
    \ifnum\pgfmathcounter=\portnum\relax%
    \else%
      \pgfutil@ifundefined{pgf@anchor@box@i\pgfmathcounter}{%
        \expandafter\xdef\csname pgf@anchor@box@i\pgfmathcounter\endcsname{%
          \noexpand\pgf@sh@lib@box@westanchor{\pgfmathcounter}%
        }%
      }{}%
    \repeatpgfmathloop%
  }
}

\def\pgf@sh@lib@box@eastanchor#1{%
  \pgfmathparse{dim(\box@eastports)}\pgfmathresult
  \let\portnum=\pgfmathresult
  \ifnum#1=\portnum\relax%
    \pgfpointorigin%
  \else
    \pgfmathparse{{\box@eastports}[#1-1]}%
    \let\z\pgfmathresult
    \southwest
    \pgf@ya=\pgf@y%
    \advance\pgf@ya by -\z\pgf@y
    \northeast
    \advance\pgf@ya by \z\pgf@y
    \pgf@y=\pgf@ya%
  \fi
}

\def\pgf@sh@lib@box@westanchor#1{%
  \pgfmathparse{dim(\box@westports)}\pgfmathresult
  \let\portnum=\pgfmathresult
  \ifnum#1>\portnum\relax%
    \pgfpointorigin%
  \else
    \pgfmathparse{{\box@westports}[#1-1]}%
    \let\z\pgfmathresult
    \northeast
    \pgf@ya=\z\pgf@y%
    \southwest
    \advance\pgf@y by -\z\pgf@y
    \advance\pgf@y by \pgf@ya
  \fi
}

\tikzset{
  left bot offset/.initial=.05cm,
  left mid offset/.initial=.0cm,
  right mid offset/.initial=-.05cm,
  right top offset/.initial=-.0cm,
  bar thickness/.initial=.15cm,
}

\tikzset{
  guarded box/.style={
    draw,
    thick,
    text width=2em,
    minimum height=1.5em,
    align=center,
    font=\itshape,
    gbox,
    left bot offset=.05cm,
    left mid offset=-.05cm,
    right mid offset=-.05cm,
    right top offset=.05cm,
    bar thickness=1cm,
    append after command={
      let \p1=(\tikzlastnode.gwest), 
          \p2=(\tikzlastnode.geast),
          \p3=(\tikzlastnode.east),
          \p4=(\tikzlastnode.north) in
      node[gbar, 
        left, 
        text width=0em,
        inner sep=0pt,
        minimum width=\pgfkeysvalueof{/tikz/bar thickness},
        minimum height=\y4-\y3-\pgfkeysvalueof{/tikz/left bot offset}-\pgfkeysvalueof{/tikz/left mid offset},
        at=(\p1)] {} 
      node[gbar, 
        right,
        text width=0em, 
        inner sep=0pt, 
        minimum width=\pgfkeysvalueof{/tikz/bar thickness},
        minimum height=\y4-\y3+\pgfkeysvalueof{/tikz/right top offset}+\pgfkeysvalueof{/tikz/right mid offset}, 
        at=(\p2)] {} 
    }
  },
  hide bars/.code=\hide@barstrue,
  guarded box/.default=0
}

\newif\ifhide@bars
\newif\ifgbar@left

\tikzset{
  gbar/.style={
    draw,
    gprof,
    minimum size=.35cm,
    left/.code=\gbar@lefttrue,
    right/.code=\gbar@leftfalse,
  },
  gprof/.default=0
}

\makeatother

\tikzset{
  unguarded box/.style={
    draw,
    thick,
    text width=2em,
    minimum height=1.5em,
    align=center,
    box,
  },
}

\tikzset{
  annotation/.style={
    font=\footnotesize,
    draw,
    solid,
    circle,
    very thin,
    inner sep=.3pt
  }
}

\tikzset{
  sigma line/.style={draw, smooth} %, out=0,in=180},
}

\makeatletter
\newcommand\currentcoordinate{\the\tikz@lastxsaved,\the\tikz@lastysaved}
\makeatother

\pgfkeys{
  /tikz/left pin offset/.initial=.5cm,
  /tikz/right pin offset/.initial=.5cm,
}

\tikzset{
  left pins/.code args={#1}{
       \foreach \i in {#1}
          \draw let \p1 = \i in \i to ([xshift=-\pgfkeysvalueof{/tikz/left pin offset}]\x1,\y1);
  },
  right pins/.code args={#1}{
       \foreach \i in {#1}
          \draw let \p1 = \i in \i to ([xshift=\pgfkeysvalueof{/tikz/right pin offset}]\x1,\y1);
  }
}

\tikzset{
  tension/.initial=1,
  xwobble/.initial=.5,
  sigma line/.style={
    xshift/.initial=20,
    to path={
      .. controls ($ (\tikztostart -| \tikztotarget) !\pgfkeysvalueof{/tikz/tension}*\pgfkeysvalueof{/tikz/xwobble}! (\tikztostart) $)
      and ($ (\tikztotarget -| \tikztostart)  !\pgfkeysvalueof{/tikz/tension} * (1-\pgfkeysvalueof{/tikz/xwobble})! (\tikztotarget)  $) 
      .. (\tikztotarget) \tikztonodes
    }
  },
  ystretch/.initial=20,
  trace line/.style={
    xshift/.initial=20,
    to path={
      let \p1 = (\tikztostart),
          \p2 = (\tikztotarget),
          \n1 = {max(\y1,\y2)},
          \n2 = {\pgfkeysvalueof{/tikz/ystretch}},
          \n3 = {(\n1-\y1)},
          \n4 = {(\n1-\y2)} 
      in
         .. controls ++(5:\n2) and ++(-5:1.5*\n2) 
         .. ($(\x1,\n1) + (-.5*\n2,\n2)$) 
         .. controls ++(175:1*\n2) and ++(5:{.5*(\n2+\y1-\y2)}) 
         .. ($(\x2,\n1) + ({.5*(\n2+\y1-\y2)}, \n2)$) 
         .. controls ++(5:{-1.5*\n2-.7*(\y1-\y2)}) and ++(-185:{\n2+.2*(\y1-\y2)}) 
         .. (\tikztotarget) \tikztonodes
    }
  },
  sigma line/.default=.5cm
}

\tikzset{
  cprop/.style={scale=0.8, inner ysep=0pt,every node/.style={scale=0.8,line width=.6pt}},
  guarded box/.append style={draw, line width=.6pt},
  unguarded box/.append style={draw, line width=.6pt}
}

\def\defbbname#1{\expandafter\def\csname BB#1\endcsname{{\bm{\mathsf{#1}}}}}
\def\defbbnames#1{\ifx#1\defbbnames\else\defbbname#1\expandafter\defbbnames\fi}
\defbbnames ABCDEFGHIJKLMNOPQRSTUVWXYZ\defbbnames

\usepackage[utf8]{inputenc}

\usepackage{savesym}

\savesymbol{degree}
\savesymbol{leftmoon}
\savesymbol{rightmoon}
\savesymbol{fullmoon}
\savesymbol{newmoon}
\savesymbol{diameter}
\savesymbol{emptyset}
\savesymbol{bigtimes}
\savesymbol{langle}
\savesymbol{rangle}

\savesymbol{triangleright}
\savesymbol{triangleleft}

\savesymbol{not}

\usepackage[matha,mathx]{mathabx}

\restoresymbol{other}{langle}
\restoresymbol{other}{rangle}

\restoresymbol{other}{triangleright}
\restoresymbol{other}{triangleleft}

\restoresymbol{other}{not}

\usepackage{tikz}
\usetikzlibrary{%
  calc,%
  arrows,%
  shapes.misc,% wg. rounded rectangle
  shapes.arrows,%
  chains,%
  matrix,%
  positioning,% wg. " of "
  scopes,%
  decorations.pathmorphing
}
\usepackage{tikz-cd}

\tikzset{
  commutative diagrams/.cd,
  arrow style=tikz,
  diagrams={>=stealth},
  row sep=large, 
  column sep = huge
}

\theoremstyle{definition}

\newcommand{\klstar}{\star}  			%% Kleisli star
\newcommand{\klcomp}{\mathbin{\diamond}}  	%% Kleisly composition
\newcommand{\pistar}{\sharp} %\dagger}  			%% Iteration
\newcommand{\istar}{\dagger} %\dagger}  			%% Iteration
\newcommand{\iistar}{{\hat\pistar}}  		%% Strong iteration
\newcommand{\bistar}{{\langle\kern-.8pt\pistar}}
\newcommand{\bbistar}{{\langle\kern-.8pt\istar}}

\newcommand{\zero}{\oname{o}}
\newcommand{\suc}{\oname{s}}

\newcommand{\swap}{\oname{swap}}

\renewcommand{\uncurry}{\curry^\mone}

\let\o\relax
\newcommand{\o}{\hspace{.8ex}}
\renewcommand{\c}{\colon}
\renewcommand{\wave}[1]{{\accentset{\scalebox{.9}[1]{\texttildelow}}{#1}}}

\newcommand{\dom}{\oname{dom}}
\newcommand{\rest}{\downharpoonright}

\usepackage{bm}

\newcommand{\nat}{\mathbb{N}}
\newcommand{\enat}{\bar{\mathbb{N}}}

\newcommand{\out}{\oname{out}}
\newcommand{\tuo}{\oname{out}^{\text{\kern.5pt\rmfamily-}\kern-.5pt1}\kern-1pt}
\newcommand{\dist}{\oname{dstr}}
\newcommand{\ldist}{\oname{dstl}}
\newcommand{\assoc}{\oname{assoc}}
\newcommand{\coit}{\oname{coit}}

\newcommand{\init}{\oname{init}}
\newcommand{\primr}{\oname{p-rec}}

\newcommand{\now}{\oname{now}}

\DeclareMathOperator{\lat}{\triangleright} % {\oname{later}}
\DeclareMathOperator{\ear}{\triangleleft}  %{\oname{earlier}}

\newcommand{\IE}{\BBK}
\newcommand{\IA}{K}

\usepackage{pict2e}

\DeclareRobustCommand{\pigpenA}{%
  \begingroup\setlength{\unitlength}{1em}%
  \linethickness{.075em}%
  \begin{picture}(1,.8)
  \roundcap\roundjoin
  \polyline(.2,.2)(.8,.2)(.8,.8)
  \end{picture}%
  \endgroup
}

\newcommand{\pbk}{\arrow[dr, phantom, "\text{\tiny\pigpenA}", pos=0.05]}
\usepackage{todos}

\def\KTO{{%
    \setbox0\hbox{$\xrightarrow{\kern9pt}$}%
    \rlap{\hbox to \wd0{$\hss\klcomp\hss$}}\box0
}}

\renewcommand{\ij}{\mathop{\hstretch{.7}{\bigsqcup}}}

\tikzset{shiftarr/.style={
        rounded corners,%
        to path={--([#1]\tikztostart.center)
                     -- ([#1]\tikztotarget.center) \tikztonodes
                     -- (\tikztotarget)},
}}

\newcommand{\bind}[2]{\mfix{do}{\mathbin{}#1}{\kern-1pt;}{\mathbin{}#2}}

\newcommand{\mdo}[2]{\mfix{mdo}{\mathbin{}#1}{\kern-1pt;}{\mathbin{}#2}}

\newcommand{\axname}[1]{{\upshape\bfseries\textsf{#1}}\xspace}

\newcommand{\FIX}{\axname{Fixpoint}}
\newcommand{\NAT}{\axname{Naturality}}
\newcommand{\UNI}{\axname{Uniformity}}
\newcommand{\COM}{\axname{Compositionality}}
\newcommand{\COD}{\axname{Codiagonal}}
\newcommand{\STR}{\axname{Strength}}
\newcommand{\FOL}{\axname{Folding}}

\makeatletter
\newcommand{\dotdiv}{\mathbin{\text{\@dotminus}}}

\newcommand{\@dotminus}{%
  \ooalign{\hidewidth\raise1ex\hbox{.}\hidewidth\cr$\m@th-$\cr}%
}
\makeatother

\renewcommand{\paragraph}[1]{\medskip\noindent{\bfseries\sffamily #1.}}

\EventEditors{Nikhil Bansal, Emanuela Merelli, and James Worrell}
\EventNoEds{3}
\EventLongTitle{48th International Colloquium on Automata, Languages, and Programming (ICALP 2021)}
\EventShortTitle{ICALP 2021}
\EventAcronym{ICALP}
\EventYear{2021}
\EventDate{July 12--16, 2021}
\EventLocation{Glasgow, Scotland (Virtual Conference)}
\EventLogo{}
\SeriesVolume{198}
\ArticleNo{131}

\title{Uniform Elgot Iteration in Foundations}
\author{Sergey Goncharov}{FAU Erlangen-N\"urnberg}{Sergey.Goncharov@fau.de}{https://orcid.org/0000-0001-6924-8766}{
Support by Deutsche Forschungsgemeinschaft (DFG) under project GO~2161/1\dash 2 is gratefully acknowledged
}
\Copyright{Sergey Goncharov}
\authorrunning{S. Goncharov}

\hideLIPIcs
\nolinenumbers

\begin{document}\allowdisplaybreaks
\maketitle

\begin{abstract}
Category theory is famous for its innovative way of thinking of concepts by their 
descriptions, in particular by establishing \emph{universal properties}. Concepts
that can be characterized in a universal way receive a certain quality seal, which 
makes them easily transferable across application domains. The notion of partiality
is however notoriously difficult to characterize in this way, although the importance
of it is certain, especially for computer science where entire research areas, such 
as \emph{synthetic} and \emph{axiomatic domain theory} revolve around it. More 
recently, this issue resurfaced in the context of (constructive) 
\emph{intensional type theory}. Here, we provide a generic categorical iteration-based 
notion of partiality, which is arguably the most basic one. We show that the emerging
free structures, which we dub \emph{uniform-iteration algebras} enjoy various 
desirable properties, in particular, yield an \emph{equational lifting monad}. We then
study the impact of classicality assumptions and choice principles on this monad,
in particular, we establish a suitable categorial formulation of the \emph{axiom 
of countable} choice entailing that the monad is an \emph{Elgot monad}.
\end{abstract}

\section{Introduction}
\emph{Natural numbers} form a prototypical domain for programming and reasoning.
Both in \emph{category theory} and in \emph{type theory} they are characterized by a universal 
property, which consists of two parts: a definitional principle -- \emph{(structural)
primitive recursion} and a reasoning principle -- \emph{induction}. Dualization 
yields respectively \emph{co-natural numbers}, \emph{co-recursion} and \emph{co-induction}. 
Amid these two structuralist extremes, here, we analyse the challenging case of non-structural recursion in the form 
of iteration, which arises as follows.
A map 
\begin{align*}
h\c S\to X + S
\end{align*}
presents the simplest possible model of a computation process: with $S$ regarded
as a state space,~$h$ sends any state either to a successor \emph{state} or to
a terminal \emph{value} in $X$. We wish to be able to form an object $\IA X$
of \emph{denotations} potentially reachable via such processes. Besides the
values of~$X$ reachable in a finite number of steps, $\IA X$ must also contain
a designated value for divergence, generated by the right injection $h=\inr$.
We then ask: what would be the generic universal characterization of $\IA X$ and 
what properties it would imply? Somewhat surprisingly, this question has not been addressed yet on a level
of generality, sufficiently close to the settings where the question can be
posed, although many similar closely related questions have been addressed,
mostly couched in type-theoretic terms. 

The question trivializes whenever one of the two following perspectives is adopted.
\begin{itemize}
  \item \emph{intensional perspective:} the domain $\IA X$ keeps track not only 
of results, but also of the number of steps needed to reach them. This 
leads to the identification of $\IA X$ as the final coalgebra $D X = \nu\gamma.\ X+\gamma$,
known as \emph{Capretta's monad} or the \emph{delay monad}~\cite{Capretta05}.
  \item \emph{non-constructive perspective:} assuming non-constructive principles,
  such as the \emph{law of excluded middle}, leads to the identification of $\IA X$ as 
  the \emph{maybe-monad} $X+1$. 
\end{itemize}
Here, we generally keep aloof from these interpretations of $\IA X$ and work both 
extensionally and generically, using the language of the category theory to analyse 
the issue in the abstract, and keeping the potential class of models possibly large.

We introduce $\IA X$ as a certain free structure, equipped with an iteration
operator, which sends any $f\c S\to\IA X+S$ to $f^\pistar\c S\to\IA X$, and satisfies 
the following two basic and uncontroversial principles:
\begin{itemize}
  \item \emph{fixpoint:} $f^\pistar$ is in an obvious sense a fixpoint of $f$;
  \item \emph{uniformity:} the structure of the state space $S$ is ineffective
  (i.e.\ merging or adding new states done coherently does not influence the result).
\end{itemize}
We dub such structures \emph{uniform-iteration algebras}
and show that on a high level of generality (in any extensive category with finite 
limits and a stable natural number object) if~$\IA X$ exists then it satisfies
a number of other properties:~$\IA $ extends to a monad~$\IE$, which is an equational lifting
monad~\cite{BucaloFuhrmannEtAl03}, the Kleisli category of~$\IA $ is enriched over 
partial orders and monotone maps, and the iteration operator is a least fixpoint operator w.r.t.\ this order;
moreover, the iteration operator satisfies an additional principle, previously
dubbed \emph{compositionality}~\cite{AdamekMiliusEtAl06}. 

In some environments, such as \emph{homotopy type theory (HoTT)}, $\IE$ can 
be constructed directly, by using \emph{higher inductive types}. One can then define
a universal map from the delay monad $\BBD$ to $\BBK$ and regard it as a form of
\emph{extensional collapse}. However, proving $\BBK$ to be a quotient of $\BBD$ seems to
be impossible without using (weak) choice principles~\cite{ChapmanUustaluEtAl15,
AltenkirchDanielssonEtAl17,EscardoKnapp17}. We interpret this categorically, first by introducing a categorical~\emph{limited principle of 
omniscience (LPO)} under which $\BBK$ turns out to be isomorphic to the \emph{maybe-monad}
$(\argument+1)$ and also turns out to be an \emph{Elgot monad}. This generalizes
slightly previous results~\cite{GoncharovRauchEtAl15} obtained for \emph{hyper-extensive 
categories}~\cite{AdamekMiliusEtAl08}. Second, we identify other cases of~$\BBK$ 
being a quotient of $\BBD$ and additionally being an (initial) Elgot monad, by
introducing certain coequalizer preservation conditions, abstractly capturing the 
corresponding instances of the axiom of countable~choice.

From the type-theoretic perspective, in our work we revisit the familiar waymarks 
of using/avoiding principles of classical/constructive mathematics in view of the 
tradoffs in expressive power of the corresponding constructions. Our present approach of uniform-iteration algebras as a fundamental primitive  
is entirely new, though. Moreover, we would like to emphasize that our results,
being generic, apply to a wide range of categories, whose objects need not be like sets,
or types in any conventional sense. This has a massive impact on the underlying proof methods. In topos
theory, calculations are facilitated by existence of the subobject classifier~$\Omega$,
which is used as a global parent space for propositions. Predicative theories, such 
as HoTT make do without~$\Omega$, but it is still possible to form predicative
types of propositions per universe, implying that the style of proofs
can to a significant extent be maintained, with~$\Omega$ intuitively regarded 
as ``scattered'' over the cumulative universe hierarchy.  Contrastingly, here we 
do not assume any kind of general reference spaces for propositions, which results 
in completely different proof methods. Nevertheless, we conjecture that our results 
can be implemented in HoTT. This is clear for the universe of sets, which in HoTT
form a \emph{pretopos}~\cite{RijkeSpitters15}, and hence directly satisfy our 
assumptions. For types of higher homotopy levels this should be possible by using 
existing recipes of formalizing \emph{precategories of types}~\cite{UnivalentFoundationsProgram13}.

\paragraph{Previous related work}
We relate to the work on iteration theories, starting from a seminal paper of 
Elgot~\cite{Elgot75}, who identified iteration as a fundamental unifying
notion. Equational properties of Elgot iteration were extensively explored by Bloom 
and \'Esik~\cite{BloomEsik93} with the initial iteration structure playing a prominent
role, however, since the whole setup therein is inherently classical, most of our present
agenda is essentially moot there. The uniformity property occurred under the name 
\emph{functorial dagger implication} in Bloom and \'Esik's monograph, and is an established 
and powerful principle, thus notably recognized in Simpson and 
Plotkin's work~\cite{SimpsonPlotkin00}, in the context of generic recursion (as 
opposed to the present dual case of generic iteration). Ad\'amek et al~\cite{AdamekMiliusEtAl06} 
introduced axioms of~\emph{(guarded) Elgot algebras}, and it follows from their results 
that these axioms are complete w.r.t.\ the algebras of the delay monad. Uniform-iteration
algebras are generally a proper weakening of Elgot algebras, but we show that 
$\IA X$ as a free uniform-iteration algebra over $X$ is in fact also an Elgot algebra.

Another line of research we relate to is concerned with notions of partiality, via
\emph{dominances}, in particular the \emph{Rosolini dominance} in \emph{synthetic domain} 
theory~\cite{Rosolini86}, via \emph{equational lifting monads}~\cite{BucaloFuhrmannEtAl03}, 
and via \emph{restriction categories}~\cite{CockettLack02}.
We remark that these approaches are rather concerned with specifying a notion of 
partiality than with defining it. This distinction is particularly significant
in the context of constructive type theories, such as HoTT, which revitalized the interest to defining a notion of partiality 
both predicatively and constructively and to understanding the impact of (restricted) choice principles.
Chapman et al~\cite{ChapmanUustaluEtAl15} provided a construction of a partiality monad
as a quotient of the delay monad assuming countable choice. Also, Uustalu and 
Veltri~\cite{UustaluVeltri17} explored universal properties of the obtained 
quotient as an initial \emph{$\omega$-complete pointed classifying monad}. 
Altenkirch et al~\cite{AltenkirchDanielssonEtAl17} directly based on $\omega$-complete 
partial orders to obtain a partiality monad in HoTT as a certain \emph{quotient inductive-inductive type}
without using any choice whatsoever, but established an equivalence with the delay 
monad quotient under countable choice. Chapman et al~\cite{ChapmanUustaluEtAl19} subsequently
used more basic \emph{quotient inductive types} for the same purpose.

Recently, Escard\'o and Knapp~\cite{EscardoKnapp17} 
reinforced the issue of discrepancy between the quotient of the delay monad and
partiality monads, by showing that the quotient precisely captures extensions of
Turing computable values, whereas in the absence of any choice, the reasonable 
partiality monads seem to yield properly larger carriers. The latter view
is particularly fine grained, and involves a monad, which is essentially
our monad~$\IE$. According to them, showing the desired connection between $\IE$ 
and the delay monad still amounts to (very weak) choice principles (albeit still 
not natively available in HoTT), while equivalence to more expressive monads would 
again require countable choice. Further relevant details of type-theoretic 
analysis of partiality can be found in recent theses~\cite{Veltri17,Knapp18}.
A comparison of various lifting monads in type theory using a unifying notion 
of~\emph{container} was recently provided by Uustalu and Veltri~\cite{UustaluVeltri17a}.

%Due to page limits, most proofs are collected in the appendix.

\section{Categories and Monads}
We assume familiarity with standard categorical concepts~\cite{Mac-Lane71,Awodey10}. 
In what follows, we generally work in an ambient extensive category $\BC$ with finite products, a stable 
natural number object~$\nat$ and exponentials $X^{\nat}$. By $|\BC|$ we refer to
the objects of $\BC$. We often drop indices of natural transformations to avoid clutter. 
For the same purpose, we juxtaposition of morphisms as composition.
Let us clarify this and fix some conventions.

\paragraph{Extensive categories and pointful reasoning} %\sgnote{lextensive?}
Extensiveness means existence of disjoint finite coproducts and stability 
of them under pullbacks (which must exist).
Every extensive category is \emph{distributive}, that is, every morphism 
$[\id\times\inl,\id\times\inr]\c X\times Y + X\times Z\to X\times (Y+Z)$ is an isomorphism
whose inverse we denote $\dist\c X \times (Y+Z)\to X\times Y + X\times Z$. Let 
$\ldist\c(X+Y)\times Z\to X\times Z + Y\times Z$ be the obvious dual to $\dist$.

% 
%The object $2=1+1$ is an internal Boolean algebra, e.g.\ we have morphisms $\truem,\falsem:1\to 2$, 
%$\andm,\orm\c 2 \times 2\to 2$
%and $\notm\c 2 \to 2$ defined in an obvious way using distributivity. Additionally,
%we define $\ifm\c 2 \times X\times X\to X$ and $\eqm\c\nat \times \nat\to 2$. 
In order to simplify reasoning, we occasionally use a rudimentary pointful notation for stating equalities
in $\BC$, most notably we use the case distinction operator $\oname{case}$, e.g.\ we write 
\begin{displaymath}
  f(x) = \case{g(x)}{\inl y\mto h(y)}{\inr z\mto u(z)}
\end{displaymath}
meaning $f = [h, u]\o g$ where $f\c X\to W$, $g\c X\to Y+Z$, $h\c Y\to W$ and
$u\c Z\to W$. 
%More generally, both~$h$ and $u$ may use the variable $x$, which
%semantically amounts to calling distributivity morphisms.

\paragraph{Natural numbers and primitive recursion}
%
%\begin{definition}[Parameterized Natural Number Object]
A \emph{stable natural number object} (NNO) in a Cartesian category $\BC$, 
is an object $\nat$ equipped with two morphisms $\zero\c 1\to \nat$ (\emph{zero}) and $\suc\c\nat \to \nat$ (\emph{successor}) 
such that for any $X,Y\in|\BC|$ and any $f\c X\to Y$ and $g\c Y \to Y$ there is unique 
$\init[f,g]\c X\times \nat\to Y$ such that
\begin{equation*}
\begin{tikzcd}[column sep = 8ex,row sep = 4ex]
X\rar["\brks{\id,\, \zero\bang}"]\drar["f"'] & X\times \nat\dar[dotted, "{\init[f,g]}"]\rar["\id\times\suc"] & X\times \nat \dar[dotted, "{\init[f,g]}"]\\
&Y\rar["g"] & Y
\end{tikzcd}
\end{equation*}
commutes. This combines two separate properties: there exists an initial $(1+\argument)$-algebra 
$(\nat,[\zero,\suc ]\c 1+\nat\to \nat)$, and $(X\times\nat,[\brks{\id,\zero\bang},\id\times\suc]\c X+X\times\nat\to X\times\nat)$
is an initial $(X+\argument)$-algebra. The latter property follows from the 
former in Cartesian closed categories.%It 
%is well-known that $\mu\gamma.\,X+\gamma$ is a monad. The assumption that it is 
%isomorphic to $X\times \nat$ additionally implies that it is strong with the strength
%being the associativity isomorphism $X\times (Y\times \nat)\cong (X\times Y)\times \nat$.

%We can thus structure the 
%assumptions defining $\nat$ in a different way. Assuming that an initial algebra
%$\mu\gamma.\,(X+\gamma)$ exists for every $X\in|\BC|$, it well known that 
%the functor $NX=\mu\gamma.\,(X+\gamma)$ extends to a monad.
%
%\end{definition}

More generally,
we need the derivable Lawvere's internalization of \emph{primitive recursion}~\cite{Lawvere64}:
Given $f\c X\to Y$ and $g\c Y \times X\times \nat\to Y$ there
is unique $\primr(f,g)\c X \times \nat\to Y$ such that
\begin{align*}
\primr(f,g)(x,\zero) = f(x), && \primr(f,g)(x,\suc n) = g(\primr(f,g)(x,n),x,n).
\end{align*}
We thus say that $\primr(f,g)$ is \emph{defined} by 
(primitive) recursion, whereas \emph{induction} is a \emph{proof principle}, 
stating that $\primr(f,g) = w$ for any $w\c X \times\nat\to Y$ satisfying the 
same equations.

Exponentials $X^\nat$ are adjoint to products $X\times \nat$, meaning that there 
is an isomorphism $\curry\c \BC(X\times \nat, Y)\to \BC(X,Y^\nat)$
%
%
%\begin{equation*}
%\begin{tikzcd}
%\BC(X\times \nat, Y) 
%	\ar[r,shift left=.75ex,"{\curry}"]
%    \ar[r,leftarrow,shift right=.75ex,swap,"{\uncurry}"]
%&[3em]
%\BC(X,Y^\nat)
%\end{tikzcd}
%\end{equation*}
natural in $X$. This induces an evaluation morphism $\ev=\uncurry\id\c X^\nat\times \nat\to X$ 
with the standard properties.

\paragraph{Strong functors and monads}
A functor $T$ is strong if it is equipped with a natural transformation \emph{strength}
$\tau_{X,Y}\c X\times TY\to T(X\times Y)$, satisfying standard coherence conditions
w.r.t.\ the monoidal structure $(1,\times)$ of $\BC$~\cite{Kock72}. This induces
the obvious dual $\hat\tau_{X,Y}\c TX\times Y\to T(X\times Y)$.
%%
%
%where $\assoc$ is the obvious associativity isomorphism. 
A natural transformation 
$\alpha\c F\to G$ between two strong functors is itself \emph{strong} if it preserves 
strength in the obvious sense, i.e.\ $\alpha\o\tau = \tau\o(\id\times\alpha)$. 

A monad $\BBT$ (in the form of a Kleisli triple) consists of an endomap $T\c |\BC|\to |\BC|$,
a family of morphisms $(\eta_X\in\BC(X,TX))_{X\in |\BC|}$ and a lifting operation
$(\argument)^\klstar\c\BC(X,TY)\to\BC(TX,TY)$, satisfying standard laws~\cite{Moggi91a}. It then follows
that $T$ is an endofunctor with $Tf=(\eta\o f)^\klstar$,~$\eta$ extends to a natural
transformation, and the \emph{multiplication} transformation ${\mu\c TT\to T}$ is definable
as $\id^\klstar$.
For every monad $\BBT$, whose underlying functor~$T$ is strong, $\eta$ and $\mu$ are 
strong (with $\id$ being a strength of $\Id$ and $(T\tau)\o\tau$ being a strength 
of~$\mu$).
% where $\id\c X\times Y\to X\times Y$ is the strength of $\Id$ and 
%$(T\tau)\o\tau\c X\times TTY\to TT(X\times Y)$ is the strength of $TT$. 
Such monad~$\BBT$ is then called \emph{strong} if both $\eta$ and $\mu$ are strong. 
A strong monad is \emph{commutative} if $\tau^\klstar\o\hat\tau = \hat\tau^\klstar\o\tau$.

We adopt Moggi's perspective~\cite{Moggi91a} to strong monads as carriers of 
computational effects, and thus say that a morphism $f\c X\to TY$ \emph{computes a 
value in $Y$}. Since, the only effect we deal with here is divergence, $f$ can 
either produce a value or diverge (modulo the inherent linguistic inaccuracy of the 
excluded middle law baked into the natural language).

\paragraph{Functor algebras and monad algebras} For an endofunctor $T$, we distinguish 
$T$-algebras, which are pairs $(A,a\c TA\to A)$, from $\BBT$-algebras, which can only be 
formed for monads~$\BBT$ on~$T$: a $\BBT$-algebra is a $T$-algebra $(A,a)$, which additionally
satisfies $a\o\eta=\id$ and $a\o\mu = a\o Ta$. Both $T$- and $\BBT$-algebras 
form categories under the standard structure preserving morphisms, the latter 
fully embeds into the former.

%\begin{example}
With our assumptions on $\BC$, we mean to cover the following (classes of) categories.
\begin{enumerate}
  \item Zermello-Fraenkel set theory with choice (ZFC) and further variants of 
  set theory: ETCS, ZF, CZF, etc. 
  \item Toposes satisfying countable choice, e.g.\ %, e.g.\ realizability toposes, 
  the \emph{topological topos}~\cite{Johnstone79}.
  \item Toposes not satisfying countable choice, e.g.\ \emph{nominal sets}. 
  \item Pretoposes, e.g.\ \emph{$\Pi$W}-pretoposes, compact Hausdoff spaces.
  \item The category of topological spaces $\Top$, and its subcategories, such 
  as the category of directed complete sets $\dCpo$.\sgnote{locales?}
%  \item \todo{locales?}
\end{enumerate}
%\end{example}

\section{Basic Properties of the Delay Monad}
%
%
%\begin{align*}
%\hat f = \untab\curry\o \bigl((\id +\bang)\o\iter\bigl([\inr,[\inl, f]]\c Y + (X+Y)\to X+Y\bigr)\bigr)\\
%\end{align*} 
%
%
The final coalgebras $DX=\nu\gamma.\, X + \gamma$ jointly yield a monad~$\BBD$,
called %\emph{partiality}, or 
the \emph{delay monad}~\cite{Capretta05}. %(we stick to the latter name):
%
%the unit is $\eta = \tuo\o\inl$ and Kleisli lifting of a morphism $f\c X\to DY$ is 
%characterized as the unique morphism $f^\klstar\c DX\to DY$ satisfying equation
%%
%$\out\o f^\klstar = [\out, \inr\o f^\klstar]\o\out$~\cite{Uustalu03}. 
Capretta~\cite{Capretta05}
showed that $\BBD$ is strong, 
%Strength $\tau\c X \times DY\to D(X\times Y)$
%is the final coalgebra morphsim from $(X\times DY\comma\dist\o(\id\times\out): X\times DY\to X\times Y+ X\times DY)$ to 
%$(D(X\times Y),\out)$. 
which remains valid in our setting.  
%Moreover, by general 
%considerations,~$\BBD$ is a \emph{guarded iterative monad}. 
%
By Lambek's lemma, the final coalgebra structure $\out\c DX\to X+DX$ is an isomorphism.
Its inverse $\tuo = [\now,\oname{later}]\c X+DX\to DX$ is composed of the morphisms, conventionally 
called $\now$ and $\oname{later}$, of which the first one is the monad unit, and the
effect of the second one is intuitively to postpone the argument computation by 
one time unit. In what follows, we will write $\lat$ instead of $\oname{later}$ for the sake of 
succinctness.
%
%Moreover,~$\BBD$ fits into the generic framework of \emph{guarded iterative 
%monads}~\cite{GoncharovSchroderEtAl17}.
%
%
As a final coalgebra, $DX$ comes together with a \emph{coiteration operator}: for any 
$f\c Y\to X + Y$, $\coit f\c Y\to DX$ is the unique morphism, such that 
$\out\o(\coit f)= (\id+\coit f) \o f$. 

We denote $D1$ by $\enat$, and think of it as an object of co-natural or possibly
infinite natural numbers. Note that the initial algebra structure 
$[\brks{\id,\zero\bang},\id\times\suc]\c X+X\times\nat\to X\times\nat$, is an 
isomorphism, and thus yields a $D$-coalgebra structure on $X\times\nat$. This 
induces a unique $D$-coalgebra morphism $\iota_X\c X\times\nat\to DX$. Alternatively,
we can regard $\iota_X$ as defined by primitive recursion:
\begin{align*}
  \iota_X(x,\zero)  =&\; \now(x)&
  \iota_X(x,\suc(n)) =&\; \lat(\iota_X(x, n))&
\end{align*}
Let $\hat\iota\c\nat\to\enat$ be $\iota_1$ modulo the obvious isomorphism.

In our setting, $DX$ need not be postulated, for it is in fact definable as 
a retract of the object $(X+1)^\nat$ of \emph{infinite streams}, which is elaborated 
in detail by Chapman et al~\cite{ChapmanUustaluEtAl15}. This also entails that 
$\iota$ is a componentwise monic. Intuitively, $DX$ consists of precisely those streams, which contain at most one element of 
the form $\inl x$. This intuition becomes precise in (possibly non-classical) 
set theory, where 
\begin{align*}
\now x =&\; (\inl x,\inr\star,\inr\star,\ldots)&
\lat\; (e_1,e_2,\ldots) =&\; (\inr\star,e_1,e_2,\ldots)
\end{align*}
This explains why classically, more precisely, under the \emph{law of excluded}, 
$DX$ is isomorphic to $X\times\nat + 1$. We provide a stronger result to this 
effect further below. Let us record some general facts about $\BBD$ first.

\begin{proposition}\label{prop:D-props}
The monad $\BBD$ admits the following characterization:
%is \emph{strongly guarded}, meaning that for any $X,Y,Z\in |\BC|$,
%the subspaces of $\BC(X,D(Y+Z))$ consisting of \emph{guarded morphisms} satisfy 
%the axioms in~\cite{GoncharovRauchEtAl18} where a morphism $f\c X\to D(Y+Z)$ is guarded 
%if $\out\o f\c X\to (Y+X)+D(Y+X)$ factors through $\inl+\id\c Y+D(Y+X)\to (Y+Z)+D(Y+Z)$.
%Moreover,
%
\begin{enumerate}\setlength\itemsep{-.1em}
  \item unit $\now\c X\to DX$ of\/ $\BBD$ satisfies $\out\o\now = \inl$;
  \item Kleisli lifting of $f\c X\to DY$ is the unique morphism $f^\klstar\c DX\to DY$,
  for which the diagram 
  \begin{equation*}
  \begin{tikzcd}[column sep=normal, row sep=normal]
  DX
    \dar["\out"']
    \rar["f^\klstar"]&[4em] DY\dar["\out"]\\
  X + DX\rar["{[\out f,~\inr f^\klstar]}"] & Y + DY            
  \end{tikzcd}
  \end{equation*}
  commutes;
%  satisfying equation $\out\o f^\klstar = [\out\o f,\inr\o f^\klstar]\o\out$;
  \item strength $\tau\c X\times DY\to D(X\times Y)$ is a unique such morphism that
  the diagram
  \begin{equation*}
  \begin{tikzcd}[column sep=normal, row sep=normal]
  X\times DY
    \dar["\id\times\out"']
    \ar[rr,"\tau"]&[2em]&[2em] 
  D(X\times Y)
  \dar["\out"]\\
  X \times (Y+DY)
    \rar["{\dist}"] & X\times Y + X\times DY\rar["\id+\tau"] & Y + DY            
  \end{tikzcd}
  \end{equation*}
  commutes.
%$\out\o\tau={(\id+\tau)}\o\dist\o(\id\times\out)$.
%  \item for every guarded $f\c X\to D(Y+X)$ there is unique $f^\pistar\c X\to DY$ 
%such that $f^\pistar=[\eta,f^\pistar]^\klstar\o f$.
\end{enumerate}
\end{proposition}
\begin{proof}
(1) and (2) follow from a more general characterization by Uustalu~\cite{Uustalu03};
(3) is established in~\cite{GoncharovSchroderEtAl18}.
% and (4) follows from the fact that $\BBD$ is a completely
%iterative monad~\cite{Milius05}, as for~$\BBD$ 
%guarded iterativity coincides with complete iterativity~\cite{GoncharovSchroderEtAl17}. 
\end{proof}
\begin{proposition}\label{lem:d-comm}
$\BBD$ is commutative.
\end{proposition}
Let us proceed with a characterization of the situations when $DX\iso X\times\nat + 1$.
Recall that a monic $\sigma$ is called \emph{complemented} if there exists 
$\sigma'\c X'\ito Y$, such that $Y$ is a coproduct of~$X$ and $X'$ 
with $\sigma$ and $\sigma'$ as coproduct injections. The \emph{law of excluded middle} 
states that any monic is complemented. We involve a rather more specific property.
\begin{proposition}\label{prop:D-class}
The monic $\hat\iota\c\nat\ito\enat$ is complemented iff $DX\iso X\times\nat + 1$.
\end{proposition}
\begin{proof}[Proof (Sketch).]
The necessity is obvious. Let us proceed with the sufficiency. Using extensiveness 
of $\BC$ one can obtain the following pullback:
\begin{equation*}
\begin{tikzcd}[column sep = 15ex,row sep = 3ex]
X\times\nat\pbk
	\rar["\snd"]
	\dar[hookrightarrow, "\iota"'] &  
\nat 
	\dar[hookrightarrow, "\hat\iota"]\\
DX
	\rar["D\bang"] &  
\enat
\end{tikzcd}
\end{equation*}
By assumption,
$\hat\iota$ is complemented, and since $\BC$ is extensive, so is $\iota$.
We obtain that $DX\iso\nat\times X + R$ for some $R$, and then it follows from finality of 
$DX$ that $R\iso 1$.
\end{proof}
The property of $\hat\iota\c\nat\ito\enat$ to be complemented is a categorical 
formulation of the \emph{limited principle of omniscience (LPO)}, which is rejected 
in constructive mathematics. Informally, LPO states that every infinite bit-stream
either contains $1$ at some position or contains only~$0$ everywhere (the constraint 
that the stream contains \emph{at most one} $1$, does not make a difference). We say that~$\BC$ 
is an LPO category if $\hat\iota\c\nat\ito\enat$ is complemented.
\begin{corollary}\label{cor:D-class}
Suppose that (i) $\BC$ has countable products and (ii) given a family $({\sigma_i\c A_i\to A})_{i\in\omega}$ 
of complemented pairwise disjoint monos, the induced universal morphism 
$\coprod_i A_i\to A$ is complemented. Then $\BC$ satisfies LPO and hence $DX\iso X\times\nat + 1$.
\end{corollary}
\begin{proof}
It is folklore that in categories with countable products $\nat$
is isomorphic to the sum of~$\omega$ copies of $1$. Thus $\hat\iota\c\nat\to\enat$ 
is the induced universal map, which is complemented~by~(ii).
\end{proof}
\begin{example}\label{exa:D-neg}
As expected, \cref{prop:D-class} does not apply to models, designed with constructivist 
principles in mind, such as intensional type theories, or realizability toposes,
although, it is technically possible to design a realizability topos, 
satisfying LPO~\cite{Bauer15}, in which thus $DX\iso X\times\nat+1$. Another class of examples to which \cref{prop:D-class}
does not apply stems from topology. In $\Top$, $\enat$ is a subspace of the \emph{Cantor
space} $2^\nat$ whose topology is generated by the base of opens of the form 
$\{sr\mid r\in\{0,1\}^\omega\}$ with $s\in 2^\star$. Then $\enat$ is isomorphic to 
a \emph{one-point compactification} of~$\nat$, i.e.\ it is the set~$\nat\cup\{\infty\}$, 
whose opens are all subsets of $\nat$ and additionally all complements of finite 
subsets of $\nat$ in $\nat\cup\{\infty\}$. Clearly, $\enat\ncong\nat+1$. This kind of arguments 
is inherited by higher order topology-based models, such as Johnstone's \emph{topological
topos}~\cite{Johnstone79}, which is a Grothendieck topos not satisfying LPO.
\end{example}

\begin{example}\label{exa:D-pos}
\cref{prop:D-class} and \cref{cor:D-class} cover quite a few models constructed 
in the scope of classical mathematics. Every set theory satisfying the law of
excluded middle satisfies LPO. Every presheaf topos (w.r.t.\ a classical set theory)
inherits countable coproducts from $\Set$ and those satisfy~(ii) of \cref{cor:D-class}. 
As we indicated in \cref{exa:D-neg}, a Grothendieck topos generally need not 
satisfy LPO, but, e.g.\ \emph{Schanuel topos} (aka the topos of nominal sets) 
does satisfy it, because this topos is Boolean.  
As we indicated in \cref{exa:D-neg}, $\Top$ does not satisfy LPO, but curiously
the full subcategory of directed complete partial orders $\dCpo$ (under Scott topology)
does. Both $\Top$ and $\dCpo$ have countable coproducts, but $\Top$ fails to 
satisfy condition~(ii), of~\cref{cor:D-class}, while $\dCpo$ does satisfy it.
This can be read as a manifestation of (undesirable) effects, which motivated synthetic 
domain theory~\cite{Hyland92}. 
 
Conditions (i) and (ii) in \cref{cor:D-class} are essentially the axioms of
\emph{hyper-extensive categories} by Ad\'amek et al~\cite{AdamekMiliusEtAl08} 
(modulo our background extensiveness assumption). An example of an LPO category
that fails (i) is Lawvere's ETCS. Another example of a Grothendieck topos that 
fails~(ii) can be rendered as a certain category of \emph{J\'onsson-Tarski 
algebras}~\cite{AdamekMiliusEtAl08}.
\end{example}
The above examples indicate that in models developed w.r.t.\ constructive foundations
LPO fails by design, while in models developed w.r.t.\ classical foundations,
depending on the purposes, constructively questioned principles may leak in from the 
metalogic level inside of the category, possibly in a weakened form, resulting 
in an explicit expression for $DX$.

%
%
%\begin{proposition}\label{prop:D-cart}
%$\BBD$ is \emph{Cartesian}, i.e.\ the following naturality squares for unit and multiplication
%are pullbacks:  
%%
%\begin{equation*}
%\begin{tikzcd}[column sep = 9ex,row sep = 5ex]
%X\pbk
%	\rar["f"]
%	\dar["\eta_X"'] &  
%Y 
%	\dar["\eta_Y"]\\
%DX
%	\rar["Df"] &  
%DY
%\end{tikzcd}
%%
%\qquad\qquad
%%
%\begin{tikzcd}[column sep = 9ex,row sep = 5ex]
%D^2X\pbk
%	\rar["D^2f"]
%	\dar["\mu_X"'] &  
%D^2Y 
%	\dar["\mu_Y"]\\
%DX
%	\rar["Df"] &  
%DY
%\end{tikzcd}
%\end{equation*}
%\end{proposition}

%\todo*{How to show that $\tau$ is epic?}

\section{Unguarded Elgot Algebras}
Recall the following notion from~\cite{AdamekMiliusEtAl06} where the term
\emph{complete Elgot algebra over $H$} is used. 
\begin{definition}[Guarded Elgot Algebras]\label{def:ea}
Given an endofunctor $H$, an \emph{($H$-)guarded Elgot algebra} is a tuple $(A,a\c HA\to A, (\argument)^\pistar)$
where the \emph{iteration} $f^\pistar\c X\to A$ for every given $f\c X\to A+HX$,
satisfies the following axioms:
\begin{itemize}
  \item {\upshape(\FIX)} for every $f\c X\to A+HX$, $f^\pistar = [\id, a\o Hf^\pistar]\o f$;
  \item {\upshape(\UNI)} for every $f\c X\to A+HX$ every $g\c Y \to A+HY$ and every $h\c X\to Y$, 
$(\id+Hh)\o f = g\o h$ implies $f^\pistar  = g^\pistar\o h$;
  \item {\upshape(\COM)} for every $h\c Y \to X+HY$ and $f\c X\to A+HX$, $((f^\pistar+\id)\o h)^\pistar 
= \bigl([(\id+H\inl)\o f\comma\inr\o(H\inr)]\o{[\inl\comma h]}\c X+Y\to A+H(X+Y)\bigr)^\pistar\o \inr$.
\end{itemize}  
$H$-guarded Elgot algebras form a category together with \emph{iteration preserving} 
morphisms defined as follows: a morphism $h$ from $(A,a, (\argument)^\pistar)$ to 
$(B,b, (\argument)^\pistar)$ is a morphism $h\c A\to B$ between carriers, such 
that $h\o f^\pistar = ((h+\id)\o f)^\pistar$
for every $f\c X\to A+HX$ 
(this entails $h\o a = b\o (Hh)$~\cite[Lemma 5.2]{AdamekMiliusEtAl06}). 
\end{definition}
The \COM axiom is the most sophisticated one. It intuitively states that 
running~$h$ in a loop over $Y$ as the state space, and subsequently running 
$f$ in a loop over $X$ as the state space, equivalently corresponds to running 
a certain term constructed from $f$ and $g$ in a single loop over the combined state~$X+Y$.

The axioms of guarded Elgot algebras are complete in the following sense.
\begin{theorem}\textup{\cite[Theorem 5.4,Corollary 5.7,Theorem 5.8]{AdamekMiliusEtAl06}}\label{thm:D}
For every $X$, a final coalgebra $\nu\gamma.\,X+H\gamma$ is a free $H$-guarded algebra
over $X$, in particular, existence of final coalgebras is equivalent to existence 
of free $H$-guarded Elgot algebras. 
The categories of $H$-guarded Elgot algebras and 
algebras of the monad $\nu\gamma.\,X+H\gamma$ are isomorphic.
\end{theorem}
%
%\begin{example}
%Let us illustrate importance of the~\COM axiom by showing that it 
%cannot be lifted. In a theory of classical sets $\BC=\Set$, consider $A=\{\top,\bot\}$, 
%$H=(\argument)^2$, and $a\c A\times A\to A$ being the logical disjunction. 
%For $f\c X\to A+X\times X$, and $\sigma\in\{0,1\}$ say that $b\in A$ is 
%$\sigma$-reachable from $x\in X$ if $\sigma=\eps$ and $f(x)=\inl b$, or 
%$\sigma=\sigma'0$, $f(x) = \inr\brks{x',y}$ and $b$ is $\sigma'$-reachable from $x'$,
%or $\sigma=\sigma'1$, $f(x) = \inr\brks{y,x'}$ and $b$ is $\sigma'$-reachable from $x'$.
%Let $f^\pistar(x) = \top$ if $\top$ is $\sigma$-reachable from $x$ for some $\sigma$, 
%and the number of those $\sigma$ for which $\bot$ is $\sigma$-reachable from $x$
%is finite. It is easy to see that thus defined iteration operator satisfies 
%\FIX and \UNI, but it fails \COM. Indeed, 
%\end{example}
%
By Theorem~\ref{thm:D}, free algebras of the delay monad are thus precisely the free $\Id$-guarded
Elgot algebras. We then introduce \emph{un-guarded Elgot algebras} as a certain subcategory 
of $\Id$-guarded~ones.
\begin{definition}[Unguarded Elgot Algebras]\label{def:uea}
We call\/ $\Id$-guarded Elgot algebras of the form $(A,\id\c A\to A, (\argument)^\pistar)$ 
\emph{unguarded Elgot algebras}, or simply \emph{Elgot algebras} if no confusion arises. 
Given two Elgot algebras $A$ and $B$, we call $f\c X\times A\to B$ right iteration 
preserving~if 
\begin{align*}
f (\id\times h^\pistar) = \bigl(X\times Z\xto{\id\times h} X\times (A+Z)\xto{\dist} X\times A+X\times Z\xto{f+\id} B + X\times Z\bigr)^\pistar
\end{align*}
for any $h\c Z\to A+Z$. This generalizes Elgot algebra morphisms under~$X=1$. 
\end{definition}
We write simply `iteration preserving' instead of `right iteration preserving'
in the sequel if the decomposition of $X\times A$ into the Elgot algebra part $A$ 
and the parameter part $X$ is clear from the context. Parametrization will be 
needed later for characterizing stability of free algebras~(Lemma~\ref{lem:stable}).

The unguarded Elgot algebras thus differ from the $\Id$-guarded ones 
in that the $\Id$-algebra structures $a\c A\to A$ in the former case are forced to be
trivial. This has an impact on forming the corresponding free structures: in the
guarded case, the $\Id$-algebra structures must be maximally unrestricted, which is the 
reason why we obtain a free $\Id$-guarded Elgot algebra~$DX$ with the $\Id$-algebra 
structure playing the role of delays. Intuitively, a free unguarded Elgot algebra must be a quotient of a free 
guarded one under removing delays, which is indeed what happens for LPO categories,
as we show later. Otherwise, the situation is much more subtle, and it is one of 
our goals to demonstrate that free unguarded Elgot algebras are exactly the 
semantic carriers generated by unguarded iteration.

In the unguarded case \COM can be replaced by a simpler looking  new law that we dub~\FOL:
\begin{proposition}\label{prop:fold}
Given $A\in |\BC|$, $(A, (\argument)^\pistar)$ is an Elgot algebra
iff\/ $(\argument)^\pistar$ satisfies 
\begin{itemize}
  \item {\upshape(\FIX)} for every $f\c X\to A+X$, $f^\pistar = [\id, f^\pistar]\o f$;
  \item {\upshape(\UNI)} for every $f\c X\to A+X$ every $g\c Y \to A+Y$ and every $h\c X\to Y$, 
$(\id+h)\o f = g\o h$ implies $f^\pistar  = g^\pistar\o h$;
  \item {\upshape(\FOL)} for every $h\c Y \to X+Y$ and $f\c X\to A+X$, 
  \mbox{$(f^\pistar+ h)^\pistar=[(\id + \inl)\o f\comma\inr\o h]^\pistar$.} 
\end{itemize}
\end{proposition}
The laws of Elgot algebras are summarised in Fig.~\ref{fig:elg-eq} in the style
of string diagrams, akin to those, which are used for axiomatizing traced symmetric 
monoidal categories~\cite{JoyalStreetEtAl96}. In contrast to the latter, here we 
essentially can only form traces of morphisms of the form $X+Y\to A+Y$ where $A$
is an Elgot algebra. Merging wires is to be interpreted as calling codiagonal 
morphisms $\nabla\c X+X\to X$. 

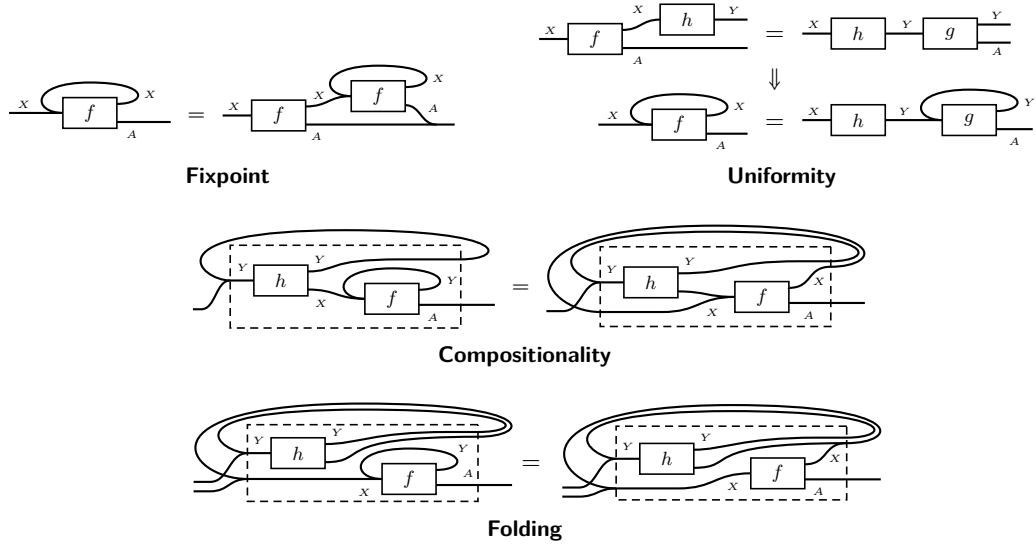
\begin{figure}[t]
\captionsetup[subfigure]{labelformat=empty}
\captionsetup[subfigure]{justification=centering}
\small\centering
\subfloat[][\FIX]{
$
\begin{aligned}
 \begin{tikzpicture} [cprop]

    \node[unguarded box] (g) {$f$};

    \draw [sigma line, thick] ([yshift=-.15cm]g.east)
      to node [yshift=-.2cm,xshift=-.2cm] {\tiny $A$} ++(6.5ex,0);

    \draw [sigma line, thick] (g.west) 
      to node [yshift=.15cm,xshift=-.2cm] {\tiny $X$} ++(-6.5ex,0);

    \begin{pgfinterruptboundingbox}
      \draw [thick] ([yshift=.15cm]g.east) to node [yshift=.15cm,xshift=.55cm] {\tiny $X$} ++(0,0)
        .. controls +(right:.5cm) and +(right:1cm) .. ([yshift=.25cm]g.north)
        .. controls +(left:1cm) and +(left:.5cm) .. (g.west);
    \end{pgfinterruptboundingbox}

    \node [minimum size=0,inner sep=0] at ([yshift=.25cm]g.north) {};

\end{tikzpicture}
\end{aligned}
~=~
\begin{aligned}
 \begin{tikzpicture} [cprop]

    \node[unguarded box] (f) {$f$};
    \node[unguarded box, above right=-.15cm and .6cm of f] (g) {$f$};

    \draw [sigma line, thick] ([yshift=.15cm]f.east)
      to node [yshift=.1cm,xshift=-0.15cm] {\tiny $X$} (g.west);

    \draw [sigma line, thick] ([yshift=-.15cm]f.east)
      to node [yshift=-.2cm,xshift=-1.05cm] {\tiny $A$} ++(18ex,0);

    \draw [sigma line, thick] (f.west) to node [yshift=.15cm,xshift=-.05cm] {\tiny $X$} ++(-3.5ex,0);
    
    \draw [sigma line, thick] ([yshift=-.15cm]g.east)
      to node [yshift=.1cm,xshift=.2cm] {\tiny $A$} ([xshift=2.15cm,yshift=-.15cm]f.east);
    
    \begin{pgfinterruptboundingbox}
      \draw [thick] ([yshift=.15cm]g.east) to node [yshift=.15cm,xshift=.55cm] {\tiny $X$} ++(0,0)
        .. controls +(right:.5cm) and +(right:1cm) .. ([yshift=.25cm]g.north)
        .. controls +(left:1cm) and +(left:.5cm) .. (g.west);
    \end{pgfinterruptboundingbox}

    \node [minimum size=0,inner sep=0] at ([yshift=.25cm]f.north) {};

\end{tikzpicture}
\end{aligned}
$
%\label{fig:m-vanish-i}
}%
\qquad
\subfloat[][\UNI]{
\begin{tabular}{r@{~~}c@{~~}l}
$
\begin{aligned}
 \begin{tikzpicture} [cprop,baseline=0pt]

    \node[unguarded box] (f) {$f$};
    \node[unguarded box, above right=-.15cm and .5cm of f] (g) {$h$};

    \draw [sigma line, thick] ([yshift=-.15cm]f.east)
      to node [yshift=-.2cm,xshift=-.5cm] {\tiny $A$} ([yshift=-.15cm]f.east-|g.east) to ++(.55cm,0);

    \draw [sigma line, thick] ([yshift=.15cm]f.east)
      to node [yshift=.25cm,xshift=-0.05cm] {\tiny $X$} (g.west);

    \draw [sigma line, thick] (f.west) 
      to node [yshift=.15cm,xshift=-.05cm] {\tiny $X$} ++(-3.5ex,0);
    
    \draw [sigma line, thick] (g.east)
      to node [yshift=.15cm,xshift=.05cm] {\tiny $Y$} ([xshift=.55cm]g.east);
    
\end{tikzpicture}
\end{aligned}
$
&=&
$
\begin{aligned}
 \begin{tikzpicture} [cprop,baseline=0pt]

    \node[unguarded box] (f) {$h$};
    \node[unguarded box, right=.5cm of f] (g) {$g$};

    \draw [sigma line, thick] (f.east)
      to node [yshift=.15cm,xshift=0.05cm] {\tiny $Y$} (g.west);

    \draw [sigma line, thick] (f.west) 
      to node [yshift=.15cm,xshift=-.05cm] {\tiny $X$} ++(-3.5ex,0);
    
    \draw [sigma line, thick] ([yshift=.15cm]g.east)
      to node [yshift=.15cm,xshift=.05cm] {\tiny $Y$} ++(.55cm,0);
    
    \draw [sigma line, thick] ([yshift=-.15cm]g.east)
      to node [yshift=-.15cm,xshift=.05cm] {\tiny $A$} ++(.55cm,0);

\end{tikzpicture}
\end{aligned}
$\\
&$\Downarrow$&\\
$
\begin{aligned}
 \begin{tikzpicture} [cprop,baseline=0pt]

    \node[unguarded box] (g) {$f$};

    \draw [sigma line, thick] ([yshift=-.15cm]g.east)
      to node [yshift=-.2cm,xshift=-.1cm] {\tiny $A$} ++(5ex,0);
      
    \draw [sigma line, thick] (g.west) 
      to node [yshift=.15cm,xshift=-.2cm] {\tiny $X$} ++(-6.5ex,0);

    \begin{pgfinterruptboundingbox}
      \draw [thick] ([yshift=.15cm]g.east) to node [yshift=.15cm,xshift=.55cm] {\tiny $X$} ++(0,0)
        .. controls +(right:.5cm) and +(right:1cm) .. ([yshift=.25cm]g.north)
        .. controls +(left:1cm) and +(left:.5cm) .. (g.west);
    \end{pgfinterruptboundingbox}

    \node [minimum size=0,inner sep=0] at ([yshift=.25cm]g.north) {};

\end{tikzpicture}
\end{aligned}
$
&=&
$
\begin{aligned}
 \begin{tikzpicture} [cprop,baseline=0pt]

    \node[unguarded box] (f) {$h$};
    \node[unguarded box, right=.75cm of f] (g) {$g$};

    \draw [sigma line, thick] (f.east)
      to node [yshift=.15cm,xshift=-0.15cm] {\tiny $Y$} (g.west);

    \draw [sigma line, thick] (f.west) to node [yshift=.15cm,xshift=-.05cm] {\tiny $X$} ++(-3.5ex,0);
    
    \draw [sigma line, thick] ([yshift=-.15cm]g.east)
      to node [yshift=-.2cm,xshift=-.0cm] {\tiny $A$} ([xshift=.6cm,yshift=-.15cm]g.east);
    
    \begin{pgfinterruptboundingbox}
      \draw [thick] ([yshift=.15cm]g.east) to node [yshift=.15cm,xshift=.55cm] {\tiny $Y$} ++(0,0)
        .. controls +(right:.5cm) and +(right:1cm) .. ([yshift=.25cm]g.north)
        .. controls +(left:1cm) and +(left:.5cm) .. (g.west);
    \end{pgfinterruptboundingbox}

\end{tikzpicture}
\end{aligned}
$
\end{tabular}
%\label{fig:m-vanish-i}
}%

\vspace{5ex}
\subfloat[][\COM]{
$
\begin{aligned}
 \begin{tikzpicture} [cprop,baseline=0pt]

    \node[unguarded box] (f) {$h$};
    \node[unguarded box, above right=-.65cm and .75cm of f] (g) {$f$};
    
    \node[unguarded box,
      densely dashed, fit={(f) (g)},
      minimum height=2.42em,
      minimum width=3.5em,
      inner xsep=.8cm,
      inner ysep=.35cm,
      xshift=.15cm,
      yshift=.05cm] (z) {};

    \draw [sigma line, thick] ([yshift=-.15cm]f.east)
      to node [yshift=-.15cm,xshift=-.25cm] {\tiny $X$} (g.west);

    \draw [sigma line, thick] (f.west-|z.west) to node [yshift=.45cm,xshift=.45cm] {\tiny $Y$} ++(-3.5ex,-3.5ex) to ++(-1ex,0);
    
    \draw [sigma line, thick] ([yshift=-.1cm]g.east)
      to node [yshift=-.2cm,xshift=-.4cm] {\tiny $A$} ([xshift=1.25cm,yshift=-.1cm]g.east);
          
    \begin{pgfinterruptboundingbox}
      \draw [thick] ([yshift=.15cm]g.east) to node [yshift=.15cm,xshift=.55cm] {\tiny $Y$} ++(0,0)
        .. controls +(right:.5cm) and +(right:1cm) .. ([yshift=.25cm]g.north)
        .. controls +(left:1cm) and +(left:.5cm) .. (g.west);

    \draw [sigma line, thick] ([yshift=.15cm]f.east)
        to node [yshift=.2cm,xshift=.2cm] {\tiny $Y$} ++(0,0) 
        .. controls +(right:.5cm) and +(left:1.25cm) .. ([yshift=.35cm,xshift=-1cm]f.east-|z.east) to ++(1cm,0)
        .. controls +(right:.75cm) and +(right:2.75cm) .. ([yshift=.25cm]z.north)
        .. controls +(left:3.25cm) and +(left:.5cm) .. (f.west-|z.west) to (f.west);
    \end{pgfinterruptboundingbox}

\end{tikzpicture}
\end{aligned}
~=~
\begin{aligned}
 \begin{tikzpicture} [cprop,baseline=0pt]

    \node[unguarded box] (f) {$h$};
    \node[unguarded box, above right=-.6cm and .75cm of f] (g) {$f$};
    
    \node[unguarded box,
      densely dashed, fit={(f) (g)},
      minimum height=2.42em,
      minimum width=3.5em,
      inner xsep=.8cm,
      inner ysep=.35cm,
      xshift=.15cm,
      yshift=.05cm] (z) {};

    \draw [sigma line, thick] ([yshift=-.15cm]f.east)
      to node [yshift=-.3cm,xshift=0.1cm] {\tiny $X$} (g.west);

    \draw [sigma line, thick] (f.west-|z.west) to node [yshift=.45cm,xshift=.5cm] {\tiny $Y$} ++(-4.5ex,-3.5ex) to ++(-2ex,0);
    
    \draw [sigma line, thick] ([yshift=-.1cm]g.east)
      to node [yshift=-.2cm,xshift=-.3cm] {\tiny $A$} ([xshift=1.25cm,yshift=-.1cm]g.east);
          
    \begin{pgfinterruptboundingbox}

      \draw [sigma line, thick] ([yshift=.15cm]f.east)
          to node [yshift=.2cm,xshift=0.2cm] {\tiny $Y$} ++(0,0) 
          .. controls +(right:1cm) and +(left:1.25cm) .. ([yshift=.35cm,xshift=-.5cm]f.east-|z.east) to ++(.5cm,0)
          .. controls +(right:.75cm) and +(right:2.75cm) .. ([yshift=.25cm]z.north)
          .. controls +(left:3.35cm) and +(left:.4cm) .. (f.west-|z.west) to (f.west);

      \draw [sigma line, thick] ([yshift=.15cm]g.east)
          to node [yshift=.1cm,xshift=0.5cm] {\tiny $X$} ++(0,0) 
          .. controls +(right:.4cm) and +(left:.4cm) .. ([yshift=.5cm]g.east-|z.east)
          .. controls +(right:.85cm) and +(right:2.95cm) .. ([yshift=.35cm]z.north)
          .. controls +(left:3.75cm) and +(left:1cm) .. ([yshift=-.25cm]g.west-|z.west) to ++(7.5ex,0) to (g.west);

    \end{pgfinterruptboundingbox}
\end{tikzpicture}
\end{aligned}
$
%\label{fig:m-vanish-i}
}%
\vspace{5ex}
\subfloat[][\FOL]{
$
\begin{aligned}
 \begin{tikzpicture} [cprop,baseline=0pt]

    \node[unguarded box] (f) {$h$};
    \node[unguarded box, above right=-.75cm and .75cm of f] (g) {$f$};
    
    \node[unguarded box,
      densely dashed, fit={(f) (g)},
      minimum height=2.42em,
      minimum width=3.5em,
      inner xsep=.8cm,
      inner ysep=.25cm,
      xshift=.15cm,
      yshift=.05cm] (z) {};

    \draw [sigma line, thick] ([yshift=-.15cm]f.east)
      to node [yshift=-.5cm,xshift=.65cm] {\tiny $X$} ++(0,0) 
          .. controls +(right:.75cm) and +(left:2.5cm) .. ([yshift=.7cm]g.east-|z.east);

    \draw [sigma line, thick] (f.west-|z.west) to node [yshift=.45cm,xshift=.5cm] {\tiny $Y$} ++(-4.5ex,-3.5ex) to ++(-2ex,0);
    
    \draw [sigma line, thick] ([yshift=-.1cm]g.east)
      to node [yshift=.2cm,xshift=-.1cm] {\tiny $A$} ([xshift=1.25cm,yshift=-.1cm]g.east);
          
    \draw [sigma line, thick] ([yshift=0cm]g.west-|z.west) to ++(-4.5ex,-1.5ex) to ++(-2ex,0) ++(-2ex,0);
          
    \begin{pgfinterruptboundingbox}

      \draw [sigma line, thick] ([yshift=.15cm]f.east)
          to node [yshift=.2cm,xshift=0.2cm] {\tiny $Y$} ++(0,0) 
          .. controls +(right:.55cm) and +(left:1.25cm) .. ([yshift=.35cm,xshift=-.5cm]f.east-|z.east) to ++(.5cm,0)
          .. controls +(right:.75cm) and +(right:2.75cm) .. ([yshift=.25cm]z.north)
          .. controls +(left:3.35cm) and +(left:.4cm) .. (f.west-|z.west) to (f.west);

      \draw [sigma line, thick] ([yshift=.7cm]g.east-|z.east)
          .. controls +(right:.85cm) and +(right:2.95cm) .. ([yshift=.35cm]z.north)
          .. controls +(left:3.75cm) and +(left:1cm) .. ([yshift=-.0cm]g.west-|z.west) to ++(7.5ex,0) to (g.west);
          
      \draw [thick] ([yshift=.15cm]g.east) to node [yshift=.35cm,xshift=.45cm] {\tiny $Y$} ++(0,0)
         .. controls +(right:.5cm) and +(right:1cm) .. ([yshift=.25cm]g.north)
         .. controls +(left:1cm) and +(left:.5cm) .. (g.west);

    \end{pgfinterruptboundingbox}
\end{tikzpicture}
\end{aligned}
$
=
$
\begin{aligned}
 \begin{tikzpicture} [cprop,baseline=0pt]

    \node[unguarded box] (f) {$h$};
    \node[unguarded box, above right=-.6cm and .75cm of f] (g) {$f$};
    
    \node[unguarded box,
      densely dashed, fit={(f) (g)},
      minimum height=2.42em,
      minimum width=3.5em,
      inner xsep=.8cm,
      inner ysep=.3cm,
      xshift=.15cm,
      yshift=.05cm] (z) {};

    \draw [sigma line, thick] ([yshift=-.15cm]f.east)
      to node [yshift=-.3cm,xshift=.65cm] {\tiny $X$} ++(0,0) 
          .. controls +(right:.75cm) and +(left:2.5cm) .. ([yshift=.5cm]g.east-|z.east);

    \draw [sigma line, thick] (f.west-|z.west) to node [yshift=.45cm,xshift=.5cm] {\tiny $Y$} ++(-4.5ex,-3.5ex) to ++(-2ex,0);
    
    \draw [sigma line, thick] ([yshift=-.1cm]g.east)
      to node [yshift=-.2cm,xshift=-.4cm] {\tiny $A$} ([xshift=1.25cm,yshift=-.1cm]g.east);
          
    \draw [sigma line, thick] ([yshift=-.25cm]g.west-|z.west) to ++(-4.5ex,-1ex) to ++(-2ex,0) ++(-1.5ex,0);
          
    \begin{pgfinterruptboundingbox}

      \draw [sigma line, thick] ([yshift=.15cm]f.east)
          to node [yshift=.2cm,xshift=0.2cm] {\tiny $Y$} ++(0,0) 
          .. controls +(right:.5cm) and +(left:1.25cm) .. ([yshift=.35cm,xshift=-.5cm]f.east-|z.east) to ++(.5cm,0)
          .. controls +(right:.75cm) and +(right:2.75cm) .. ([yshift=.25cm]z.north)
          .. controls +(left:3.35cm) and +(left:.4cm) .. (f.west-|z.west) to (f.west);

      \draw [sigma line, thick] ([yshift=.15cm]g.east)
          to node [yshift=.1cm,xshift=.5cm] {\tiny $X$} ++(0,0) 
          .. controls +(right:.4cm) and +(left:.4cm) .. ([yshift=.5cm]g.east-|z.east)
          .. controls +(right:.85cm) and +(right:2.95cm) .. ([yshift=.35cm]z.north)
          .. controls +(left:3.75cm) and +(left:1cm) .. ([yshift=-.25cm]g.west-|z.west) to ++(7.5ex,0) to (g.west);

    \end{pgfinterruptboundingbox}
\end{tikzpicture}
\end{aligned}
$
%\label{fig:m-vanish-i}
}%
\caption{Laws of (unguarded) Elgot algebras.}
\label{fig:elg-eq}
% \vspace*{4pt}
\end{figure}
As expected, products and exponents of Elgot algebras can be formed in a canonical way.
\begin{lemma}\label{lem:prod_exp}
Given two Elgot algebras $(A, (\argument)^\pistar)$ and $(B, (\argument)^\pistar)$ 
and an object $X\in |\BC|$,
\begin{enumerate}
  \item $(A\times B, (\argument)^\iistar)$ is an Elgot algebra with
%  \begin{align*}
    $h^\iistar = 
      \brks{((\fst+\id)\o h)^\pistar,
            ((\snd+\id)\o h)^\pistar
      }$ for any $h\c Z\to A\times B+Z$.
%  \end{align*}
%  Moreover, the forgetful functor from Elgot monads over $\BC$ to $\BC$ creates 
%  finite products.
  \item If $A^X$ exists then $(A^X, (\argument)^\iistar)$ is an 
  Elgot algebra with
%  \begin{align*}
    $h^\iistar = \curry ((\ev+\id)\o\ldist\o(h\times\id))^\pistar$ for any $h\c Z\to A^X+Z$.
%  \end{align*}    
\end{enumerate}
\end{lemma}
Every Elgot algebra $(A, (\argument)^\pistar)$ comes together with a divergence 
constant $\bot\c 1\to A = (\inr\c 1\to A+1)^\pistar$. Note that $\bot$ is automatically 
preserved by Elgot algebra morphisms.
%
%We distinguish Elgot 
%$\varowedge$-algebras, which are those, that are equipped with an operator 
%$\varowedge\c A\times A\to A$, such that $\bot\varowedge x = \bot$. Analogously, 
%an Elgot $\varovee$-algebra is equipped with an operator $\varovee\c A\times A\to A$, 
%such that $\bot\varovee x = x$. A morphism $h\c A\to B$ of Elgot $\varowedge$-algebras
%is required to preserve $\varowedge$: $h(a\varowedge b) = h(a)\varowedge h(b)$, and a 
%morphism $h\c A\to B$ of Elgot $\varovee$-algebras is required to preserve $\varovee$: 
%$h(a\varovee b) = h(a)\varovee h(b)$.
%

By omitting the not entirely self-motivating \COM (or \FOL) law, we obtain what we dub 
\emph{uniform-iteration algebras}. As we see later, this law is automatic 
for free uniform-iteration algebras.
\begin{definition}[Uniform-Iteration Algebras]\label{def:ba}
A \emph{uniform-iteration algebra} is a tuple $(A, (\argument)^\pistar)$ as in
\cref{def:uea} but $(\argument)^\pistar$ is only required to satisfy 
\FIX and \UNI. Morphisms of uniform-iteration algebras are defined in the same way.
\end{definition}
\section{The Initial Pre-Elgot Monad}\label{sec:ini-pre}
The goal of this section is to show that free uniform-iteration algebras
coincide with free Elgot algebras (Theorem~\ref{thm:K-pre-Elgot}), and enjoy a number of other
characteristic properties. In particular, we characterize the functor sending 
any~$X$ to a free uniform-iteration algebra on~$X$ as an initial 
pre-Elgot monad. We define pre-Elgot monads as follows.
\begin{definition}[Pre-Elgot Monads]\label{def:pem}
We call a monad $\BBT$ \emph{pre-Elgot} if every $TX$ is equipped with an Elgot 
algebra structure, in such a way that $h^\klstar\o f^\pistar = ((h^\klstar+\id)\o f)^\pistar$
for any $f\c Z\to TX+Z$ and any $h\c X\to TY$. 
A pre-Elgot monad $\BBT$ is \emph{strong pre-Elgot} if $\BBT$ is strong as a monad  
and strength is iteration preserving.
\end{definition}
Pre-Elgot monads are to be compared with Elgot monads, which support a stronger 
type profile for the iteration operator, and satisfy more sophisticated axioms.
\begin{definition}[Elgot Monads~\cite{Elgot75,AdamekMiliusEtAl11}]\label{def:elgot-monad}
  A mo\-nad\/ $\BBT$ is an \emph{Elgot monad} if it is equipped with an iteration
  operator sending each $f\colon X\to T(Y+X)$ to $f^\istar\colon X\to TY$ and
  satisfying:
\begin{itemize}%\setlength\itemsep{-.1em}
  \item {\upshape (\FIX)}  $f^\istar = [\eta,f^\istar]^\klstar\o f$;
  \item {\upshape (\NAT)} $g^{\klstar}\o f^{\istar} = ([(T\inl) \o g\comma\eta\o\inr]^{\klstar}\o f)^{\istar}$ for $f\colon X\to T(Y+X)$, $g \colon Y \to TZ$;
  \item {\upshape (\COD)} $(T[\id,\inr] \o f)^{\istar} = f^{\istar\istar}$ for  $f \colon X \to T((Y + X) + X)$;
  \item {\upshape (\UNI)} $f \o h = T(\id+ h) \o g$ implies
        $f^{\istar} \o h = g^{\istar}$ for $f\colon X \to T(Y + X)$, $g\colon Z \to T(Y + Z)$ and
        $h\colon Z \to X$.
\end{itemize}
If\/ $\BBT$ is additionally strong then $\BBT$ is strong Elgot if moreover: 
\begin{itemize}
  \item {\upshape (\STR)} $\tau\o(\id\times f^\istar) = ((T\dist)\o\tau\o (\id\times f))^\istar$ for any $f\colon X\to T(Y+X)$.
\end{itemize}
\end{definition}
\begin{proposition}
(Strong) Elgot monads are (strong) pre-Elgot under $f^\pistar = ([T\inl,\eta\inr]\o f)^\istar$.
\end{proposition}
It has been argued~\cite{GoncharovRauchEtAl15,GoncharovSchroderEtAl17} 
that strong Elgot monads are minimal semantic structures for interpreting effectful
while-languages. In that sense, we acknowledge an expressivity 
gap between Elgot and pre-Elgot monads, which generally happen to be too weak.
We will consider approaches to close this gap, in particular by drawing on some 
versions of the axiom of countable choice. Even though, in general, the gap presumably 
cannot be closed, we regard the initial pre-Elgot monad to be an important notion,
which arises from first principles and carries a very clear operational intuition. 
The discrepancy between pre-Elgot monads and Elgot monads seems to represent a 
very basic form of discrepancy between operational and denotational semantics.
We thus find it important to conceptually delineate between Elgot monads and 
pre-Elgot monads, no matter how desirable it is to have them to be equivalent.  
\begin{lemma}\label{lem:b-mon}
If for every $X\in |\BC|$ a free uniform-iteration algebra $\IA X$ exists then 
$\IA$ extends to a monad $\IE$ whose algebras are precisely uniform-iteration algebras. 
%Moreover, the induced initial uniform-iteration algebra morphism 
%$\rho_X\c DX\to \IA X$ to an iteration preserving monad morphism $\rho\c\BBD\to\IE$.
\end{lemma}
As in the case of natural numbers, one cannot make much progress without stability.
\begin{definition}[Stable Free Uniform-Iteration Algebras]
A free uniform-iteration algebra~$\IA Y$ over $Y$ is \emph{stable} if 
for every $X\in |\BC|$, $\fst\c X\times \IA Y\to X$ is a free uniform-iteration 
algebra in the slice category $\BC/X$. 
\end{definition}
\begin{lemma}\label{lem:stable}
For $Y\in |\BC|$, $\IA Y$ is stable iff for every uniform-iteration  
$A$ and every $f\c X\times Y\to A$, there is unique iteration preserving 
$f^\hash\c X\times \IA Y\to A$ such that ${f = f^\hash\o (\id\times\eta)}$.
% 
%\begin{equation*}
%\begin{tikzcd}[column sep = 12ex,row sep = 4ex]
%X\times \IA Y\rar["f^{\hash}"] & A\\
%X\times Y\uar["\id\times\eta"]\urar["f"'] &             
%\end{tikzcd}
%\end{equation*}
\end{lemma}
Using \cref{lem:prod_exp}, it is easy to show that in Cartesian closed categories 
every $\IA X$ is stable. 
%
%\begin{proposition}
%If\/ $\BC$ is Cartesian closed then every $\IA X$ is stable.
%\end{proposition}
%%
%\begin{proof}
%By \cref{lem:prod_exp}, given an uniform-iteration algebra $A$, $A^X$ is also 
%an uniform-iteration algebra. Thus, given $f\c X\times Y\to A$, there is unique 
%uniform-iteration algebra morphism $g\c KY\to A^X$ such that
%\begin{equation*}
%\begin{tikzcd}[column sep = 12ex,row sep = 4ex]
%\IA Y\rar["g"] & A^X\\
%Y\uar["\eta"]\urar["f"'] &             
%\end{tikzcd}
%\end{equation*}
%This yields the requisite $f^\hash\c X\times\IA Y\to A$ as an obvious transpose 
%of $g$.
%\end{proof}
%
For the rest of the section, we assume that all free
 uniform-iteration algebras $\IA X$ exist and are stable. 
\begin{proposition}\label{prop:tau-prop}
The monad $\IE$ is 
strong, with the components 
of strength $\tau\c X\times \IA Y\to \IA(X\times Y)$ uniquely identified 
by the conditions:
\begin{flalign*}
&&\tau\o (\id\times\eta) =\;\eta,&&
&&\tau\o(\id\times h^\pistar) =\; ((\tau+\id)\o \dist\o (\id\times h))^\pistar && && (h\c Z\to \IA Y+Z)
\end{flalign*}
\end{proposition}
\begin{proof}
In the notation of \cref{lem:stable} we define strength of $\IE$ as 
$(\eta\c X\times Y\to \IA(X\times Y))^\hash$. The axioms of strength are easy to verify.
\end{proof}
%
%\begin{lemma}
%Given $h\c Z\to \IA X+Z$,
%\begin{align*}
%&\bigl(Z\times \IA Y\xto{\tau} \IA(Z\times Y)\xto{\IA(h^\pistar\times\id)} \IA(\IA X\times Y)\bigr)=\\
%  &\bigl(Z\times \IA Y\xto{h\times\id} (\IA X+Z)\times \IA Y\xto{\ldist} 
%\IA X\times \IA Y + Z\times \IA Y\xto{\tau+\id} \IA(\IA X\times Y)+Z\times \IA Y\bigr)^\pistar  
%\end{align*}
%\end{lemma}
%\begin{proof}
%  \todo{prove!}
%\end{proof}
%
As a next step, we show that $\IE$ is an equational lifting monad in the sense 
of Bucalo et al~\cite{BucaloFuhrmannEtAl03}. This means precisely that $\IE$ is 
commutative and satisfies the equational law:
\begin{align}\label{eq:L-eq}
  \tau\o\Delta =  \IA\brks{\eta,\id}.
\end{align} 
This law is rather restrictive, and roughly means that some form 
of non-termination is the only possible effect of the monad. Proving~\eqref{eq:L-eq}
is nontrivial. The key step is the following property, which allows 
for splitting a loop involving a product of algebras into two loops.
\begin{lemma}\label{lem:key}
Given uniform-iteration algebras $A,B$ and $C$, $f\c Z\to {A\times B}+Z$ and ${h\c A\times B\to C}$,
%\begin{align*}
$((h + \id)\o f)^\pistar = ((h+\id)\o\dist\o(\id\times(\snd+\id)\o f))^\pistar\o\brks{((\fst+\id)\o f)^\pistar,\id}$.
%\end{align*}
%
\end{lemma}
\begin{lemma}\label{lem:tau-Delta}
Given $X,Z\in |\BC|$, and $h\c Z\to\IA X+Z$, then
%
%\begin{align*}
$\tau\o\brks{h^\pistar,h^\pistar} = ((\tau\o\Delta+\id)\o h)^\pistar$.
%\end{align*}
\end{lemma}

\begin{proof}
It follows from \cref{lem:key} that
%
%\begin{align*}
$((\tau+\id)\o \dist\o (\id\times h))^\pistar\o\brks{h^\pistar,\id} = 
((\tau\o\Delta+\id)\o h)^\pistar$.
%\end{align*}
%
On the other hand, by \cref{prop:tau-prop},
%
%\begin{align*}
$((\tau+\id)\o \dist\o (\id\times h))^\pistar\o\brks{h^\pistar,\id} = \tau\o\brks{h^\pistar,h^\pistar}$.
%\end{align*}
By combining the last two identities, we obtain the goal.
\end{proof}

\begin{theorem}\label{thm:K-lift}
$\IE$ is an equational lifting monad.
\end{theorem}
\begin{proof}
Let us sketch the proof of~\eqref{eq:L-eq}. 
Since $\IA\brks{\eta,\id}=(\eta\o\brks{\eta,\id})^\klstar$, using the definition 
of Kleisli star for $\IA $, it suffices to show that~$\tau\o\Delta$ is a unique 
iteration preserving morphism for which $\eta\o\brks{\eta,\id} = \tau\o\Delta\o\eta$.
%the diagram 
%\begin{equation*}
%\begin{tikzcd}[column sep=large, row sep=normal]
%\IA X\rar["\tau\o\Delta"] & \IA(\IA X\times X)\\
%X\uar["\eta"]\urar["\eta\o\brks{\eta,\,\id}"'] &             
%\end{tikzcd}
%\end{equation*}
%commutes. 
Indeed, $\tau\o\Delta\o\eta = 
\tau\o(\id\times\eta)\o\brks{\eta,\id} = \eta\o\brks{\eta,\id}$, and
$\tau\o\Delta$ is iteration preserving by \cref{lem:tau-Delta}.
\end{proof}
%
%Given $f\c Z\to\IA X + Z$, let $\steps f\c Z\times\nat\to\IA\nat + Z\times\nat$ be
%as follows: 
%%
%\begin{align*}
%Z\times\nat
%  \xto{f\times\id} (\IA X + Z) \times\nat
%  \xto{\ldist} \IA X\times\nat + Z\times\nat 
%  \xto{\eta\o\snd+\id\times\suc} 
%\IA\nat + Z\times\nat.
%\end{align*}
%%
%Intuitively, $(\steps f)^\pistar\o\brks{\id,\zero\bang}\c Z\to\IA\nat$ 
%returns the eventual number of steps needed to reach the result.
%
%
%\begin{lemma}
%For any $f\c Z\to\IA X + Z$, 
%%
%\begin{align*}
%f^\pistar = (\fiter f)^\klstar\o\tau\o\brks{\id, (\steps f)^\pistar\o\brks{\id,\zero\bang}}.
%\end{align*}
%%
%\end{lemma}
%\begin{proof}
%%
%We have
%\begin{align*}
%f^\pistar = (f^\pistar\o\fst)^\klstar\o\tau\o\brks{\id, (\steps f)^\pistar\o\brks{\id,\zero\bang}}.
%\end{align*}
%
%
%
%Let us rewrite the right-hand side:
%%
%\begin{align*}
%(\fiter f)^\klstar\o\tau\o\brks{\id, (\steps f)^\pistar\o\brks{\id,\zero\bang}}
%=&\;(((\fiter f)^\klstar\o\tau+\id)\o\dist\o(\id\times \steps f))^\pistar\o\brks{\id, \brks{\id,\zero\bang}}
%\end{align*}
%%
%\end{proof}
%
The fact that $\IE$ is an equational lifting monad has a number of implications,
in particular, the Kleisli category of $\IE$ is a \emph{restriction category}~\cite{CockettLack02}.
That is, we can calculate the \emph{domain (of definiteness)}, represented by an idempotent
Kleisli morphism as follows: given $f\c X\to\IA Y$, 
\begin{align*}
  \dom f = (\IA\fst)\o\tau\o\brks{\id,f}\c X\to\IA X,
\end{align*}
We additionally use the notation $f\rest g = \fst^\klstar\o\tau\o\brks{f,g}$, meaning:
restrict $f$ to the domain of $g$. It is easy to see that $\dom f = \eta\rest f$
and $f\rest g = f^\klstar\o(\dom g)$. Let $f\appr g$ abbreviate 
$f = g\rest f$. Under this definition, every $\BC(X, \IA Y)$ is partially ordered,
which is a general fact about restriction categories. In our case, moreover, this 
partial order additionally has a bottom element $\bot = \inr^\pistar$;
$\dom(\eta\o f) = \eta$ for any $f\c X\to\IA Y$, and $\dom f\appr\eta$ for any~$f$.

\begin{proposition}\label{prop:K-enriched}
The Kleisli category of\/ $\IE$ is enriched over pointed partial orders
and strict monotone maps. Moreover, strength preserves $\bot$ and $\appr$ as follows:
\begin{align*}
\tau\o(\id\times\bot) =\bot && f\appr g\text{\qquad implies\qquad}\tau\o(\id\times f) \appr \tau\o(\id\times g)
\end{align*} 
\end{proposition}
\begin{corollary}
$\IA\iobj\iso 1$.
\end{corollary}
\begin{proof}
Since $\bang\bot = \id\c 1\to 1$ and $\bot\bang = \id\c\IA\iobj\to\IA\iobj$,
we obtain an isomorphism $\IA\iobj\iso 1$.
\end{proof}
\begin{proposition}\label{prop:copy}
The monad $\IE$ is \emph{copyable} and \emph{weakly discardable}~\cite{GoncharovSchroder13a}, 
i.e.:
%\begin{align*}
$\hat\tau^\klstar\o\tau\o\Delta = \IA\Delta$ and %&& 
$(\IA\fst)\o\hat\tau^\klstar\o\tau\o\brks{f,g}\appr f$
%\end{align*}
%
%\begin{equation*}
%\begin{tikzcd}[column sep=normal, row sep=normal]
%X\rar["\brks{\id,f}"]\dar["f"'] & X\times\IA Y\dar["\tau\o(f\times\id)"]\\
%\IA Y\rar["\IA\Delta"]& \IA(Y\times Y)               
%\end{tikzcd}
%\end{equation*}
for $f\c X\to\IA Y$ and $g\c X\to\IA Z$. 
\end{proposition}
\begin{definition}[Bounded Iteration]\label{def:biter}
Let $A$ be a \emph{pointed object}, i.e.\ an object with a canonical map 
$\bot\c 1\to A$. Then we define \emph{bounded iteration} $(\argument)^\bistar\c
\BC(X, A+X)\to\BC(X\times\nat, A)$ by primitive recursion as follows:
\begin{align*}
f^{\bistar}(x,\zero) = \bot && f^{\bistar}(x,\suc n) = \case{f(x)}{\inl a\mto a}{\inr y\mto f^{\bistar}(y,n)}.
\end{align*}
\end{definition}
Intuitively, $f^{\bistar}(x,n)$ behaves as 
$f^{\pistar}(x)$ except that at each iteration the counter $n$ is decreased, and
$\bot$ is returned once $n=\zero$. We next show that $f^{\pistar}(x)$ is in a suitable sense a limit of the $f^{\bistar}(x,n)$
as $n$ tends to infinity. This is, of course, a form of \emph{Kleene fixpoint theorem}.
\begin{theorem}[Kleene Fixpoint Theorem]\label{thm:klee}
Given $f\c X\to\IA Y + X$, and $g\c X\to\IA Y$, 
%\begin{enumerate}
  \item 
(i) $f^{\bistar}\appr f^\pistar\o\fst$, and
%  \item 
(ii) $f^{\bistar}\appr g\o\fst$ implies $f^\pistar\appr g$.
%\end{enumerate}
\end{theorem}

\begin{corollary}\label{cor:klee}
Given $f\c X\to\IA Y + X$, $f^\pistar\c X\to\IA Y$ is the least pre-fixpoint of the 
map $[\id,\argument]\o f\c\BC(X,\IA Y)\to\BC(X,\IA Y)$.
\end{corollary}
Finally, we obtain
\begin{theorem}\label{thm:K-pre-Elgot}
$\IE$ is an initial pre-Elgot monad and an initial strong pre-Elgot monad.
\end{theorem}
%

%\begin{proposition}
%Suppose that all free Elgot algebras $\IA X$ exist and are stable. Then $\IA X$ 
%extends to a free Elgot $\varowedge$-algebra on $X$ with 
%$f \varowedge g = g\rest f$.
%\end{proposition}
%\begin{proof}
%By definition, $\bot\varowedge g = g\rest\bot=\bot$. We are left to show that 
%for any $\varowedge$-algebra $(A, (\argument)^{\hat\iistar})$ and any $f\c X\to A$
%the induced universal morphism $h\c\IA X\to A$ preserves $\varowedge$.
%\end{proof}

\section{Quotienting the Delay Monad}\label{sec:quo}
By \cref{thm:D}, $\Id$-guarded Elgot algebras are precisely the $\BBD$-algebras. We 
proceed to characterize uniform-iteration and Elgot algebras as certain $D$-algebras,
which we dub \emph{search-algebras}. Intuitively, modulo identification of $DA$ 
with a set of streams from $(A+1)^\nat$, a search-algebra structure $a\c DA\to A$
is guaranteed to find the first element in the stream of the form $\inl a$ if it 
exists. We expect that this notion can formulated more generally, but we do not pursue it here.
\begin{definition}[Search-Algebra]
We call a $D$-algebra $(A, a\c DA\to A)$ a \emph{search-algebra} if it satisfies 
the conditions: %$a = [\id,a]\o\out$, or, equivalently, the pair of conditions 
%
%\begin{align*}
$a\o\now=\id$,
%&& 
$a\o\lat=a$.
%\end{align*}
%
Search-algebras form a full subcategory of the category of all $D$-algebras.
\end{definition}
Uniform-iteration algebras capture the structure of search-algebras 
independently of the assumption that $D$ exists. This and further connections 
between categories of $D$-algebras illustrated in Fig.~\ref{fig:d-alg} (arrows indicate full embeddings of categories)
are formalized as follows.

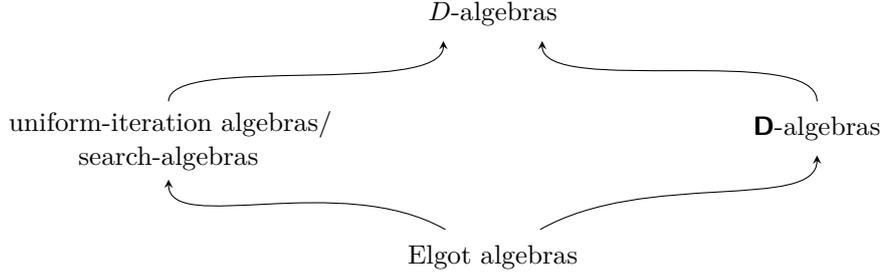
\begin{figure}[t]
\begin{center}%
  \begin{tikzpicture}[node distance = 2cm, auto]
    % nodes
    \node [block] (d-alg) {$D$-algebras};
    \node [block, below left  =.8cm and -.2cm of d-alg,  node distance=-5cm]  (uni) {uniform-iteration algebras/\\search-algebras};
    \node [block, below right =.8cm and -.2cm of d-alg, node distance=.5cm]  (malg) {$\BBD$-algebras};
    \node [block, below =2.5cm of d-alg.south] (elgot) {Elgot algebras};

    \draw [line,stealth-] ([xshift= -.65cm]d-alg.south) 
      .. controls +(down:.75cm) and +(up:1.15cm) .. (uni);
    \path [line,stealth-] ([xshift=  .65cm]d-alg.south) 
      .. controls +(down:.75cm) and +(up:1.15cm) .. (malg);
    \draw [line,stealth-] (uni.south) 
      .. controls +(down:.75cm) and +(up:1.15cm) .. ([xshift= -2cm]elgot);
    \path [line,stealth-] (malg.south) 
      .. controls +(down:.75cm) and +(up:1.15cm) .. ([xshift= 2cm]elgot);
  \end{tikzpicture}
\end{center}%
\caption{Connections between classes of $D$-algebras.}
\label{fig:d-alg}
\end{figure}

%
%
%\begin{lemma}\label{lem:lat-ear}
%The coequalizer~\eqref{eq:D-quote} exists iff a coequalizer of $D\fst\c 
%D(X\times\nat)\to DX$ and $\iota^\klstar\c D(X\times\nat)\to DX$
%exists and both are isomorphic.
%\end{lemma}
%\begin{proof}
%\todo{This seems to require that $\tau$ is epi.}
%\end{proof}

%
%Observe the following.
%%
%\begin{lemma}
%The canonical initial uniform-iteration algebra morphism 
%$\rho_X\c DX\to \IA X$ is an iteration preserving strong monad morphism $\rho\c\BBD\to\IE$.
%\end{lemma}
%\begin{proof}
%Since $\IA X$ is defined as a free object, $\rho$ preserves monad unit by definition.
%Given $f\c X\to \IA Y$, note that $\rho$, $f^\klstar$ and $(\rho\o f)^\klstar$ are 
%by definition iteration preserving, and hence so are both $\rho\o f^\klstar$ 
%and $(\rho\o f)^\klstar\o\rho$. Since $\rho\o f^\klstar\o\eta = \rho\o f 
%= (\rho\o f)^\klstar\o\eta=(\rho\o f)^\klstar\o\rho\o\eta$, this implies $\rho f^\klstar=
%(\rho f)^\klstar\o\rho$, by the universal property of $\IA X$, hence $\rho$ preserves 
%Kleisli star. 
%
%\todo{preservation of strength.}
%\end{proof}
%
\begin{proposition}\label{prop:uniform-iteration-eq}
\begin{enumerate}
  \item
The categories of uniform-iteration algebras and search-algebras are isomorphic
under:
\begin{align*}
(A,(\argument)^\pistar)&\;\mto (A\comma\out^\pistar\c DA\to A),\\
(A,a\c DA\to A)&\;\mto (A\comma a\o \coit(\argument)\c\BC(X,A+X)\to\BC(X,A)).
\end{align*}
\item Elgot algebras are precisely those $D$-algebras, which are
search-algebras and\/ $\BBD$-algebras.
\end{enumerate}
\end{proposition}
%
%
%Elgot algebras and uniform-iteration algebras can also be related as follows.
%

%Suppose that each $DX=\nu\gamma.\,X+\gamma$ exists. 

%
\begin{lemma}\label{lem:elgot-iota}
Every Elgot algebra $(DA,a\c DA\to A)$ satisfies $a\o\iota^\klstar = a\o (D\fst)$. 
\end{lemma}
%
%
%
%Let $\ear = [\now,\id]\o\out\c DX\to DX$.
%%
%\begin{lemma}\label{lem:lat-ear}
%\todo{no longer true.} The coequalizer~\eqref{eq:D-quote} exists iff a coequalizer of $\id$ and $\ear$
%exists and both are isomorphic.
%\end{lemma}
%\begin{proof}
%%
%It suffices to show that any $f\c DX\to Y$ coequalizes $\id,\lat$ iff it 
%coequalizes $\id,\ear$. 
%
%$\Rightarrow$\,: if $f = f\o\lat$ then  $f\o\ear = 
%[f\o\now,f]\o\out=[f\o\now,f\o\lat]\o\out = f$;
%
%$\Leftarrow$\,: if $f = f\o\ear$ then $f\o\lat = 
%f\o [\now,\id]\o\out\o\lat=f\o [\now,\id]\o\inr = f$.
%\end{proof}
%
We proceed to model the construction of quotienting $\BBD$
by weak bisimilarity~$\approx$, previously described in type-theoretic
terms~\cite{ChapmanUustaluEtAl15}. Modulo identification of~$DX$ with the object
of those streams $\sigma\c\nat\to X+1$ for which $\sigma(n)\neq\inr\star$ for at most
one $n$, $\approx$ can be described as follows: $\sigma\approx\sigma'$
if for every $a$, $\sigma(n)=a$ for some $n$~iff $\sigma'(n) = a$ for some $n$.

%This 
%can be equivalently formulated as follows: $\sigma\approx\sigma'$ if for every
%$n$ either $\sigma(k)=\sigma'(k)=\inr\star$ for all $k\leq n$ or there are $k,m$
%and $a$, such that $n=k+m$ and $\sigma(k) = \sigma(m)=\inl a$. We thus consider 
%the kernel pair $\eps_X\c\approx_X\;\ito DX\times DX$ of the morphism 
%$\hat\tau^\klstar\o\tau\c DX\to D(X\times X)$ as a categorical formalization of this 
%idea.

Recall the  embedding $\iota\c X\times\nat\ito DX$, and define the 
quotient of $DX$ by the  coequalizer
\begin{equation}\label{eq:D-quote}
\begin{tikzcd}
D(X\times\nat) 
  \ar[r,shift right=.35ex,"D\fst"']
  \ar[r,shift left=.75ex, "\iota^\klstar"]
&[3em]
DX\rar["\rho_X"] 
&[2em] 
\wave DX
\end{tikzcd}
\end{equation}
which we assume to exist and be preserved by products.
It is then straightforward that~$\wave D$
is a functor and $\rho_X$ is natural in~$X$. It also follows that $X\xto{~\now~} DX\xto{~\rho~} \wave DX$ is strong.
Following tradition, we denote $\wave D1$ as $\Sigma$. 
\begin{lemma}\label{lem:rho-eq-lat}
$\rho\o\lat = \rho$.
\end{lemma}
%\begin{proof}
%We have $\rho = \rho\o D\fst\o D\brks{\id,\suc\o\zero\o\bang} 
%= \rho\o \iota^\klstar\o D\brks{\id,\suc\o\zero\o\bang} = \rho\o (\iota\o\brks{\id,\suc\o\zero\o\bang})^\klstar = \rho\o\lat$.
%\end{proof}
%
Defining $\rho$ as a coequalizer of $\lat$ and $\id$ in the first place does not 
seem to be sufficient, though, in particular, for showing the following property.
We leave open the question of identifying conditions under which it is 
possible.  
\begin{proposition}\label{lem:bis}
The following is a coequalizer:
\begin{equation}\label{eq:D-quote-sym}
\begin{tikzcd}
D(X + (X\times\nat+X\times\nat)) 
  \ar[r,shift right=.35ex,"{[\eta,\,[\eta\o\fst,\,\iota\o(\id\times\suc)]]}^\klstar"']
  \ar[r,shift left=.75ex, "{[\eta,\,[\iota\o(\id\times\suc),\,\eta\o\fst]]}^\klstar"]
&[8em]
DX\rar["\rho_X"] 
&[1.5em] 
\wave DX
\end{tikzcd}
\end{equation}
\end{proposition}
The last proposition brings the definition of $\rho$ in accordance with the 
intuition; the coproduct $X + (X\times\nat+X\times\nat)$ covers three alternatives
for $\sigma\approx\sigma'$: either $\sigma=\sigma'$, or $\sigma$ terminates earlier
than $\sigma'$ by a specified number, or the other way around. It can be verified 
that the embedding $D(X + (X\times\nat+X\times\nat))\ito DX\times DX$ is an 
internal equivalence relation.

\begin{theorem}\label{thm:Drho}
The following conditions are equivalent:
\begin{enumerate}
%  \item\label{it:q2} every $(\wave DX, \rho_X\o\now\c X\to\wave DX)$ is a free 
%  uniform-iteration algebra on $X$ and $\rho_X\mu_X = ((\rho_X+\id)\o\out)^\pistar$;
%  \item\label{it:q0} the coequalizer~\eqref{eq:D-quote} with $X=1$ is preserved by $D$; 
  \item\label{it:q2} for every $X$, coequalizer~\eqref{eq:D-quote} is preserved by $D$; 
  \item\label{it:q1} every $\wave DX$ extends to a search-algebra, 
%  $(\wave DX\comma\alpha_X\c D\wave DX\to\wave DX)$, 
so that each~$\rho_X$ is a $D$-algebra morphism;  
  \item\label{it:q4} for every $X$, $(\wave DX,\rho\o\now\c X\to\wave DX)$ is 
  a stable free Elgot algebra on $X$, $\rho_X$ is a $D$-algebra morphism and 
  $\rho_X = ((\rho_X\o\now +\id)\o\out)^\pistar$;
	\item\label{it:q3} $\wave D$ extends to a strong monad, so that $\rho$ is a strong monad morphism.
\end{enumerate}
\end{theorem}
If the equivalent conditions of Theorem~\ref{thm:Drho} are satisfied, we obtain 
an explicit construction of the initial pre-Elgot monad $\IE$, which we explored 
previously. Let us consider concrete examples.
\begin{example}[Maybe-Monad]\label{exa:maybe}
Suppose that $\BC$ is an LPO category, and recall that $DX$ is isomorphic to 
$X\times\nat + 1$. It is then easy to check that~\eqref{eq:D-quote} 
exists, it is preserved by products, $\wave DX\iso X+1$ and $\rho = \fst+\id\c X+1\to X\times\nat + 1$. Since
$D$ is the composition of $(\argument\times\nat)$ and $(\argument+1)$, and both these
functors preserve coequalizers (first as a left adjoint, and second by extensiveness 
of~$\BC$),~$D$ preserves~\eqref{eq:D-quote}. We thus obtain that the maybe-monad 
is an initial pre-Elgot monad. This covers instances of LPO categories from~\cref{exa:D-pos}.
Moreover, the initial pre-Elgot monad is in fact an initial Elgot monad in this 
case: the profiles of the iteration operators $(\argument)^\pistar$ and~$(\argument)^\istar$ 
agree up to rearrangement of summands, and the axioms of
\cref{def:elgot-monad} become the axioms of \cref{def:pem}, except for \COD, 
which can be checked directly.
\end{example}
Note that Example~\ref{exa:maybe} entails that the maybe-monad is the initial 
Elgot monad in $\dCpo$. This is a result of our assumption that $\dCpo$ is developed
w.r.t.\ a classical set theory, which entails that $\dCpo$ is an LPO category. 
This would not be the case if we defined $\dCpo$ internally to a non-classical 
environment, which is indeed the core idea of synthetic domain theory.

Another direction for obtaining an Elgot monad 
from~\eqref{eq:D-quote} is by
using a suitable instance of the \emph{axiom of countable choice}. In our setting 
this takes the following form.
\needspace{2\baselineskip}
\begin{theorem}\label{thm:elg}
Suppose that the coequalizers~\eqref{eq:D-quote} are preserved by the 
exponentiation~$(\argument)^\nat$.
\begin{enumerate}
  \item\label{it:elg1} The equivalent conditions of~\cref{thm:Drho} hold, in particular, 
  $\wave\BBD$ is an initial (strong) pre-Elgot monad.
  \item\label{it:elg2} If every~\eqref{eq:D-quote-sym} is an effective quotient, 
  i.e.\ $D(X + (X\times\nat+X\times\nat))$ is a kernel pair of $\rho_X$, then $\wave\BBD$ 
  is a strong Elgot monad with $f^\istar$ being the least fixpoint of\/ 
  $[\eta,\argument]^\klstar\o f\c\BC(X,\wave DY)\to\BC(X,\wave DY)$ for any $f\c X\to \wave D(Y+X)$.
%  The Kleisli category of\/ $\wave\BBD$ is enriched over (internally) 
%  $\omega$-complete partial orders and continuous morphisms, and strength 
%  preserves joins in the relevant argument. 
%  \item\label{it:elg3} $\wave\BBD$ is a strong Elgot monad.
\end{enumerate}  
\end{theorem}
The effectiveness assumption in clause~{\bfseries\sffamily\ref{it:elg2}.} is satisfied in any exact 
category (e.g.\ in any pretopos) -- by definition, every internal equivalence
relation there is effective.
\begin{example}
Theorem~\ref{thm:elg} applies to $\Top$, yielding a concrete description for $\IE$.
Recall that in $\Top$, coequalizers are computed as in $\Set$ and are equipped
with the quotient topology. Note that $DX$ is the set $X\times\nat\cup\{\infty\}$ 
whose base opens are $\{(x,n)\mid x\in O\}$ and $\{(x,k)\mid x\in X,k\geq n\}\cup\{\infty\}$ 
with $n\in\nat$ and $O$ ranging over the opens of $X$. The collapse~$\wave DX$
computed with~\eqref{eq:D-quote} is thus the set $X\cup\{\infty\}$, whose opens 
are those of $X$ and additionally the entire space $X\cup\{\infty\}$, in particular, 
$\wave D1$ is the \emph{Sierpi\'nski space}. 

To obtain that~\eqref{eq:D-quote} 
is preserved by $(\argument)^\nat$, it suffices to show that the opens of 
$(X\cup\{\infty\})^\nat$ are precisely those, whose inverse images under 
$\rho^\nat$ are open. This is in fact true for any regular epi in $\Top$. The 
effectiveness condition in~{\bfseries\sffamily\ref{it:elg2}.} is not vacuous for 
$\Top$, which is not an exact category (and not even regular), but it can be checked 
manually.
%
% and we 
%conclude that $\wave\BBD$ is an Elgot monad.
\end{example}
In every pretopos, preservation of~\eqref{eq:D-quote} by~$(\argument)^\nat$
is a proper instance of the \emph{internal axiom of countable choice}, or \emph{internal 
projectivity of\/ $\nat$}, which means preservation of epis by 
$(\argument)^\nat$, roughly because every pretopos is \emph{exact} and our quotienting 
morphism~$\rho$ is associated with an internal equivalence relation by~\cref{lem:bis}.
%
%Indeed, in a topos, every epimorphism is regular. Also, every 
%topos is \emph{exact}, i.e.\ every equivalence relation, in particular, figuring 
%in~\eqref{eq:D-quote-sym}, is a kernel pair. Now, if~$\nat$ is internally projective,
%then $\rho^\nat$ is epi, hence a regular epi, hence a coequalizer of its kernel pair,
%which are preserved by $(\argument)^\nat$ as limits. As a result,~\eqref{eq:D-quote-sym}
%is preserved by~$(\argument)^\nat$, and the fact that the resulting coequalizer 
%is a coequalizer of $(\iota^\klstar)^\nat$ and $(D\fst)^\nat$ is shown analogously 
%to~\cref{lem:bis}. 
\cref{thm:elg} can thus be related to the existing result in 
synthetic domain theory, that Rosolini dominance, i.e.\ our $\Sigma$, is indeed 
a dominance~\cite{Rosolini86}, which applies to Hyland's \emph{effective topos}~\cite{Hyland82}, 
as it satisfies countable choice. Contrastingly, we cannot apply \cref{thm:elg} to nominal 
sets, which falsify countable choice, however, as a 
Boolean topos, nominal sets fall into the scope of~\cref{exa:maybe}.

We currently do not have a concrete example of $\IE$ being definable, but not
being an Elgot monad. Theorem~\ref{thm:elg} and Example~\ref{exa:maybe} suggests
that a non-artificial category to witness this must neither support excluded middle nor the axiom 
of countable choice.

\section{Conclusions and Further Work}\label{sec:conc}
Iteration and iteration theories emerged as unifying concepts for computer 
science semantics and reasoning. By interpreting iteration suitably, 
one obtains a basic extensible equational logic of programs, shown to be sound and complete
across various models~\cite{BloomEsik93}. Elgot monads implement
this inherently algebraic view in the general categorical realm 
of abstract data types and effects. The class of Elgot monads (over a fixed category) 
is stable under various categorical constructions (monad transformers), and thus one can build 
new Elgot monads from old, but the most simple Elgot monad, the initial one, does 
not arise in this way. 

Here, we proposed an approach to defining
an initial iteration structure from first principles, characterized it in 
various ways, analysed conditions, under which it can be concretely described, 
and to yield an Elgot monad. Unsurprisingly, these conditions 
generally cannot be lifted, as the previous research in type theory indicates. 
We consider broadening the scope in which results about notions of partiality apply, 
and unifying both classical and non-classical models, as an important part of our 
contribution. Universal properties play a central role in category theory,  
but many important concepts are not covered by them. One example is
Sierpi\'nski space, which is fundamental in topology, duality theory 
and domain theory. It follows from our results, that it is in fact a free 
uniform-iteration algebra on one generator. We believe that the structure of 
our results can be reused in more sophisticated setting, such as semantics 
of \emph{hybrid systems}, which require a notion of partiality, combined with continuous 
evolution, and rise semantic issues, structurally similar to those, we considered 
here~\cite{DiezelGoncharov20}. Another potential for taking further
the present work is to consider more general shapes of the basic functor (instead of 
the current $(X+\argument)$), prospectively leading to more sophisticated (non-)structural 
recursion scenarios (see e.g.~\cite{AdamekMiliusEtAl20}).

%\clearpage

\bibliography{monads}

%\end{document}

\clearpage
\appendix
\allowdisplaybreaks

\section{Appendix: Omitted Details and Proofs}

\subsection{Axioms of Strength}
\begin{equation*}
\begin{tikzcd}[column sep = 6ex,row sep = 4ex]
X\times TY\dar["\tau"']\rar["\snd"] &  TY\\
T(X\times Y)\urar["T\snd"'] &  
\end{tikzcd}
\hspace{4ex}
\begin{tikzcd}[column sep = 3ex,row sep = 4ex]
(X\times Y)\times TZ\dar["\assoc"']\ar[rr,"\tau"] &[3ex] &  T((X\times Y)\times Z)\dar["T\assoc"]\\
X\times (Y\times TY)\rar["\id\times\tau"] & X\times T(Y\times Z)\rar["\tau"] &  T(X\times (Y\times Z))
\end{tikzcd}
\end{equation*}
%Explicitly, the following two additional laws are required: 
%

%
\begin{equation*}
\begin{tikzcd}[column sep = 4ex,row sep = 4ex]
X\times Y\dar["\id\times\eta"']\rar["\eta"] &  T(X\times Y)\\
X\times TY\urar["\tau"'] &  
\end{tikzcd}
\hspace{6ex}
\begin{tikzcd}[column sep = 4ex,row sep = 4ex]
X\times TTY\dar["\tau"']\ar[rr,"\id\times\mu"] &[1ex] &  X\times TY\dar["\tau"]\\
T(X\times TY)\rar["T\tau"] & TT(X\times Y)\rar["\mu"] &  T(X\times Y)
\end{tikzcd}
\end{equation*}
%
%For commutative monads, we denote the combined symmetric transformation 
%$\psi = \tau^\klstar\o\hat\tau = \hat\tau^\klstar\o\tau\c TX\times TY\to T(X\times Y)$ 
%throughout.

\subsection{Proof of \cref{lem:d-comm}}
We will need the following 
\begin{lemma}\label{lem:later}
Let $f\c X\to DY$. Then 
\begin{enumerate}
  \item $\lat f^\klstar = (\lat f)^\klstar = f^\klstar\o\lat$;
  \item $\tau\o (\id\times\lat) = \lat\o\tau$.
\end{enumerate}
\end{lemma}
\begin{proof}
For the first clause, note that 
\begin{align*}
  \out [f, (\lat f)^\klstar]\o\out =&\; [\out\o f, [\out\o\lat f, \inr (\lat f)^\klstar]\o\out]\o\out\\*
 =&\; [\out\o f, \inr [f, (\lat f)^\klstar]\o\out]\o\out,
\end{align*}
which implies $[f, (\lat f)^\klstar]\o\out = f^\klstar$, for the resulting equation 
is uniquely satisfied by~$f^\klstar$. Now, $\out (\lat f)^\klstar = 
[\out \lat f,\inr (\lat f)^\klstar]\o\out = \inr [f, (\lat f)^\klstar]\o\out = \inr\o f^\klstar=\out(\lat f^\klstar)$,
which implies $\lat f^\klstar = (\lat f)^\klstar$, since $\out$ is an isomorphism.

Analogously, $\out\o f^\klstar\o\lat = [\out\o f, \inr\o f^\klstar]\o\out\o\lat = [\out\o f, \inr\o f^\klstar]\o\out\o\inr = \inr\o f^\klstar = \out\o\lat f^\klstar$,
and therefore $f^\klstar\o\lat = \lat f^\klstar$.

Let us proceed with the second clause. By definition, 
$\out\o\tau=(\id+\tau)\o\dist\o(\id\times\out)$, which implies 
$\tau\o(\id\times\tuo)=\tuo\o(\id+\tau)\o\dist$. Hence,
\begin{align*}
 \tau\o(\id\times\lat)
=&\;\tau\o(\id\times\tuo)\o(\id\times\inr)\\
=&\;\tuo\o(\id+\tau)\o\dist\o(\id\times\inr)\\
=&\;\tuo\o(\id+\tau)\o\inr\\
=&\;\tuo\o\inr\o\tau\\
=&\;\lat\tau,
\end{align*}
and we are done.
\end{proof}
Let us continue with the proof of \cref{lem:d-comm}.
It is easy to see that the dual $\hat\tau\c DX\times Y\to D(X\times Y)$ of $\tau$ is 
a final coalgebra morphism from $(DX\times Y\comma\ldist\o(\out\times\id): DX\times Y\to X\times Y+ DX\times Y)$ to 
$(D(X\times Y),\out)$ where $\ldist$ is the obvious dual of $\dist$. We need to 
to check that $\hat\tau^\klstar\o\tau = \tau^\klstar\o\hat\tau$. Using the elementary
properties of strength, distributivity transformations and \cref{lem:later},
we obtain: 
\begin{flalign*}
&& \out&\o\hat\tau^\klstar\o\tau\\*
&&&  \;=[\out \hat\tau,\inr \hat\tau^\klstar]\o\out\o\tau &\\
&&&\;=[\out \hat\tau,\inr \hat\tau^\klstar]\o (\id+\tau)\o \dist\o(\id\times\out) \\
&&&\;=[(\id+\hat\tau)\o\ldist\o(\out\times\id),\inr \hat\tau^\klstar\o\tau]\o \dist\o(\id\times\out) \\
&&&\;=[(\id+\hat\tau)\o\ldist,\inr \hat\tau^\klstar\o\tau\o (\tuo\times\id)]\o \dist\o(\out\times\out) \\
&&&\;=[(\id+\hat\tau)\o\ldist,\inr \hat\tau^\klstar D(\tuo\times\id) \tau ] \dist\o(\out\times\out) \\
&&&\;=[(\id+\hat\tau)\o\ldist,\inr\o (\hat\tau\o(\tuo\times\id))^\klstar \tau ] \dist\o(\out\times\out) \\
&&&\;=[(\id+\hat\tau)\o\ldist,\inr\o (\tuo\o (\id+\hat\tau)\o\ldist)^\klstar \tau ]\o \dist\o(\out\times\out) \\
&&&\;=[(\id+\hat\tau)\o\ldist,\inr\o (\tuo\o (\id+\hat\tau))^\klstar\o(D\ldist)\o\tau ]\o \dist\o(\out\times\out) \\
&&&\;=[(\id+\hat\tau)\o\ldist,\inr\o (\tuo\o (\id+\hat\tau))^\klstar\o [(D\inl)\tau,(D\inr)\tau] \o\ldist ]\o \dist\o(\out\times\out) \\
&&&\;=[\id+\hat\tau,\inr\o [(\tuo\o\inl)^\klstar \tau,(\tuo\o\inr\o\hat\tau)^\klstar\o\tau]]\o (\ldist+\ldist)\o\dist\o(\out\times\out) \\
&&&\;=[\id+\hat\tau,\inr [\eta^\klstar \tau,(\tuo\o\inr\o\hat\tau)^\klstar\o\tau] ]\o (\ldist+\ldist)\o\dist\o(\out\times\out) \\
&&&\;=[\id+\hat\tau,\inr [\tau, (\lat\hat\tau)^\klstar\o\tau] ]\o (\ldist+\ldist)\o\dist\o(\out\times\out) \\
&&&\;=[\id+\hat\tau, \inr[\tau, \lat\hat\tau^\klstar\o\tau] ]\o (\ldist+\ldist)\o\dist\o(\out\times\out).
\end{flalign*}
%
%By a symmetric argument, we obtain an analogous equation for $\hat\tau^\klstar\o\tau$. 
%In summary,
%%
%\begin{align*}
%\hat\tau^\klstar\o\tau =&\; \tuo[\id+\hat\tau, \inr[\tau, \lat\hat\tau^\klstar\o\tau]]\o w\\
%\tau^\klstar\o\hat\tau =&\; \tuo[\id+\hat\tau, \inr[\tau, \lat\tau^\klstar\o\hat\tau]]\o w
%\end{align*}
%Let $w = (\ldist+\ldist)\o\dist\o(\out\times\out)$. 
%
That is, $\hat\tau^\klstar\o\tau$ is a fixpoint of 
\begin{align}\label{eq:fix-tau-tau}
  f\mto \tuo\o[\id+\hat\tau, \inr\o[\tau,\, \lat f]]\o (\ldist+\ldist)\o\dist\o(\out\times\out).
\end{align}
on $\BC(DX\times DY, D(X\times Y))$.
By a symmetric argument, also $\tau^\klstar\o\hat\tau$ is a fixpoint of the same map.
Let us denote $(\ldist+\ldist)\o\dist\o(\out\times\out)$ by $w$. For every fixpoint $f$ 
of~\eqref{eq:fix-tau-tau}:
\begin{align*}
f =\;& \tuo\o [\id+\hat\tau, \inr[\tau, \lat f] ] \o  w\\
=\;& \tuo\o [[\inl, \out\o f],\inr [\eta,f]^\klstar]\o [\inl+(D\inl)\hat\tau, \inr[(D\inl)\tau, \lat\eta\o\inr] ]\o w\\
=\;& [\eta,f]^\klstar\o\tuo[\inl+(D\inl)\hat\tau, \inr[(D\inl)\tau, \eta\o\inr] ]\o w
\end{align*}
and, clearly, 
$\tuo[\inl+(D\inl)\hat\tau, \inr[(D\inl)\o\tau, \eta\o\inr]]\o w\c DX\times DY\to D(X\times Y + DX\times DY)$ 
is guarded. By \cref{prop:D-props}~(4) this implies $\tau^\klstar\o\hat\tau = \hat\tau^\klstar\o\tau$.
\qed

\subsection{Proof of \cref{prop:fold}}
Of course, the first two axioms are just corresponding instances of their 
guarded versions. We only need to show that \FOL is equivalent to 
\COM in presence of the other two axioms. Let us first equivalently 
reformulate the \FOL axiom. Note that by uniformity,
$[(\id + \inl)\o f\comma\inr\o h]^\pistar\o\inl = f^\pistar$ and also by \FIX,
$(f^\pistar+ h)^\pistar\o\inl = [\id\comma(f^\pistar+ h)^\pistar]\o(f^\pistar+ h)\inl = f^\pistar$.
Analogously, $(f^\pistar+ h)^\pistar\o\inr = [\id\comma(f^\pistar+ h)^\pistar]\o(f^\pistar+ h)\inr = (f^\pistar+ h)^\pistar\o h$,
and subsequently, by uniformity, $(f^\pistar+ h)^\pistar\o\inr = ((f^\pistar+\id)\o h)^\pistar$.
We have thus equivalently reduced \FOL to 
\begin{align*}
[(\id + \inl)\o f,\inr\o h]^\pistar\o\inr = ((f^\pistar+\id)\o h)^\pistar.
\end{align*}
Now, the task of checking that this is equivalent to \COM amounts 
to showing that
\begin{align*}
([(\id+\inl)\o f,\inr\o\inr]\o[\inl, h])^\pistar\o\inr = [(\id + \inl)\o f,\inr\o h]^\pistar\o\inr
\end{align*}
follows from \FIX and \UNI. Indeed, by \UNI,
\begin{align*}
([(\id+\inl)\o f,\inr\o\inr]\o[\inl, h])^\pistar
=\;& ((\id + [\inl, h])[(\id+\inl)\o f,\inr\o\inr])^\pistar\o[\inl, h]\\*
=\;& ([(\id+\inl)\o f,\inr\o h])^\pistar\o[\inl, h],
\intertext{hence}
([(\id+\inl)\o f,\inr\o\inr]\o[\inl, h])^\pistar\o\inr =\;& ([(\id+\inl)\o f,\inr\o h])^\pistar\o h.
\intertext{On the other hand, using \FIX,}
[(\id + \inl)\o f,\inr\o h]^\pistar\o\inr 
=\;&[\id, [(\id + \inl)\o f,\inr\o h]^\pistar]\inr\o h\\*
=\;&[(\id + \inl)\o f,\inr\o h]^\pistar\o h. 
\end{align*}
In summary, we obtain the desired identity.
\qed

\subsection{Proof of \cref{lem:b-mon}}
Existence of all free unguarded uniform-iteration algebras yields an adjunction between 
$\BC$ and the category of unguarded uniform-iteration algebras over $\BC$. Let us call 
the latter category $\BE$ and the corresponding adjunction $U\vdash F$. The only 
nonobvious condition of Beck's monadicity theorem is existence and preservation of $U$-split 
coequalizers. Consider a split coequalizer 
\begin{equation}\label{eq:becks}
\begin{tikzcd}
UA 
  \ar[r,shift right=.35ex,"U g"']
  \ar[r,shift left=.75ex, "U f"]
&[3em]
UB\rar["e"] 
& 
Z
\end{tikzcd}
\end{equation}
in $\BC$, i.e.\ for suitable $t\c UB\to UA$, $s\c Z\to UB$, $e\o s=\id$, $s\o t=g\o t$,
and $f\o t=\id$. Let~$C$ be uniform-iteration algebra whose carrier is $Z$ and 
whose iteration operator is defined as follows $(h\c X\to Z+X)^\iistar = e\o ((s+\id)\o h)^\pistar$.
Hence $e$ lifts to an uniform-iteration algebra morphism, which is moreover a coequalizer 
of $f$ and $g$ in $\BE$. The image of the resulting coequalizer in~$\BE$ is 
precisely~\eqref{eq:becks}, and thus we have shown that $U$ preserves $U$-split 
coequalizers. Therefore,~$ L$ extends to a monad whose category of algebras 
is precisely $\BE$.
\qed

\subsection{Proof of \cref{lem:stable}}
Assume stability of $\IA Y$.
It is easy to check 
that $\fst\c X\times A\to X$ is a uniform-iteration algebra in $\BC/X$, which implies 
that $\id\times\eta\c X\times Y\to X\times \IA Y$ is the unit morphism for $X\times \IA Y$.
Let us fix some $f\c X\times Y\to A$ and note that $\brks{\fst,f}\c X\times Y\to X\times A$ 
is a morphism in $\BC/X$. Using the universal property of $\fst\c X\times A\to X$ in $\BC/X$ we obtain a unique iteration preserving 
morphism $u\c X\times \IA Y\to X\times A$ in $\BC/X$ for which the diagram 
\begin{equation*}
\begin{tikzcd}[column sep=8em, row sep=normal]
X\times \IA Y\rar["u"]\drar["\fst", pos=.7] & X\times A\dar["\fst"]\\
X\times Y\uar["\id\times\eta"]\urar["\brks{\fst,f}", pos=.2]\rar["\fst"'] & X            
\end{tikzcd}
\end{equation*}
commutes. This entails that $u$ is of the form $\brks{\fst,f^\hash}$ where the 
requisite property of $f^\hash$ follows from the fact that $u$ is iteration 
preserving in $\BC/X$. Conversely, from a unique $f^\hash$, subject to the declared
properties we can render $u$ as $\brks{\fst,f^\hash}$. 
\qed

\subsection{Proof of \cref{lem:key}}
It suffices to show that 
\begin{align*}
(h+\id)\o\dist\o(\id\times (\snd+\id)\o f)\o\brks{((\fst+\id)\o f)^\pistar,\id}
= (\id+\brks{((\fst+\id)\o f)^\pistar,\id})\o(h + \id)\o f
\end{align*}
from which the claim follows by uniformity. After simplifications we obtain
\begin{align*}
(h+\id)\o\dist\o\brks{((\fst+\id)\o f)^\pistar,(\snd+\id)\o f}
= (h+\brks{((\fst+\id)\o f)^\pistar,\id})\o f,
\end{align*}
which is verified directly as follows:
\begin{flalign*}
\quad(h+\id)\o\dist&\o(((\fst+\id)\o f)^\pistar\o(z),(\snd+\id)\o(f(z)))\\*
=\;&\case{f(z)}{\inl (p,q)\mto\inl\o h(((\fst+\id)\o f)^\pistar\o(z),q)}{\\*&\qquad\inr z'\mto\inr (((\fst+\id)\o f)^\pistar\o(z),z')}\\
=\;&\case{f(z)}{\inl (p,q)\mto\inl\o h([\fst, ((\fst+\id)\o f)^\pistar] (f(z)),q)}{&\by{\FIX}\\&\qquad\inr z'\mto\inr ([\fst, ((\fst+\id)\o f)^\pistar] (f(z)),z')}&\by{\FIX}\\
=\;&\case{f(z)}{\inl (p,q)\mto\inl\o h([\fst, ((\fst+\id)\o f)^\pistar] \inl(p,q),q)}{\\*&\qquad\inr z'\mto\inr ([\fst, ((\fst+\id)\o f)^\pistar] \inr z',z')}\\
=\;&\case{f(z)}{\inl (p,q)\mto\inl (p,q)}{\\*&\qquad\inr z'\mto\inr (((\fst+\id)\o f)^\pistar\o(z'),z')}\\
=\;&(h+\brks{((\fst+\id)\o f)^\pistar,\id}) (f(z)).
\end{flalign*}
This completes the proof.
\qed

\subsection{Proof of \cref{thm:K-lift}}
\begin{lemma}\label{lem:k-comm}
$\IE$ is a commutative monad.
\end{lemma}
\begin{proof}
Let $\tau\c X\times \IA Y\to \IA(X\times Y)$ be strength of $\IE$ and let
$\hat\tau\c \IA X\times Y\to \IA(X\times Y)$ be its obvious transpose. We need to 
show that $\hat\tau^\klstar\o\tau = \tau^\klstar\o\hat\tau$. Note that 
$\IA X\times\IA Y$ is an uniform-iteration algebra by \cref{lem:prod_exp}~(1). It is easy to 
see that both diagrams 
\begin{equation*}
\begin{tikzcd}%[column sep=normal, row sep=normal]
\IA X\times \IA Y\rar["\tau^\klstar\o\hat\tau"] & \IA(X\times Y)\\
X\times Y\uar["\eta\times\eta"]\urar["\eta"'] &             
\end{tikzcd}\qquad
\begin{tikzcd}%[column sep=normal, row sep=normal]
\IA X\times \IA Y\rar["\hat\tau^\klstar\o\tau"] & \IA(X\times Y)\\
X\times Y\uar["\eta\times\eta"]\urar["\eta"'] &             
\end{tikzcd}
\end{equation*}
commute. It is therefore sufficient for obtaining the desired identity 
$\hat\tau^\klstar\o\tau = \tau^\klstar\o\hat\tau$ to show that both $\hat\tau^\klstar\o\tau$
and $\tau^\klstar\o\hat\tau$ are iteration preserving. We confine to the former case,
from which the second case is obtained by a symmetric argument.
Let $f\c Z\to \IA X\times\IA Y+Z$. We need to show that
\begin{align*}
\hat\tau^\klstar\o\tau\o\brks{((\fst+\id)\o f)^\pistar,
          ((\snd+\id)\o f)^\pistar} = ((\hat\tau^\klstar\o\tau + \id)\o f)^\pistar.
\end{align*}
This in fact essentially follows from \cref{lem:key}:
\begin{align*}
((\hat\tau^\klstar\o\tau + \id)\o f)^\pistar
=\;& \hat\tau^\klstar\o((\tau + \id)\o f)^\pistar\\*
=\;& \hat\tau^\klstar\o((\tau+\id)\o\dist\o(\id\times(\snd+\id)\o f))^\pistar\o\brks{((\fst+\id)\o f)^\pistar,\id}\\
=\;& \hat\tau^\klstar\o\tau\o(\id\times ((\snd+\id)\o f)^\pistar)\o\brks{((\fst+\id)\o f)^\pistar,\id}\\
=\;& \hat\tau^\klstar\o\tau\o\brks{((\fst+\id)\o f)^\pistar, ((\snd+\id)\o f)^\pistar}.
\end{align*}
This completes the proof.
\end{proof}
Let us continue the proof of the theorem. 

We have already shown that $\IE$ is strong and commutative. It remains to 
establish the law~\eqref{eq:L-eq}. 
Since $\IA\brks{\eta,\id}=(\eta\o\brks{\eta,\id})^\klstar$, using the definition 
of Kleisli star for $\IA $, it suffices to show that~$\tau\o\Delta$ is a unique 
iteration preserving morphism for which the diagram 
\begin{equation*}
\begin{tikzcd}[column sep=large, row sep=normal]
\IA X\rar["\tau\o\Delta"] & \IA(\IA X\times X)\\
X\uar["\eta"]\urar["\eta\o\brks{\eta,\,\id}"'] &             
\end{tikzcd}
\end{equation*}
commutes. Indeed, $\tau\o\Delta\o\eta = 
\tau\o(\id\times\eta)\o\brks{\eta,\id} = \eta\o\brks{\eta,\id}$, and
$\tau\o\Delta$ is iteration preserving by \cref{lem:tau-Delta}.
\qed

\subsection{Restriction Categories and Equational Lifting Monads}
Recall the axioms of restriction categories~\cite{CockettLack02} for further 
reference.
\begin{align}
f^\klstar\o(\dom f) &\;= f                                       \tag{\axname{RST{\scriptsize 1}}}\label{eq:rst1}\\*
(\dom f)^\klstar\o(\dom g) &\;= (\dom g)^\klstar\o(\dom f)    \tag{\axname{RST{\scriptsize 2}}}\label{eq:rst2}\\
\dom(g^\klstar\o(\dom f)) &\;= (\dom g)^\klstar\o(\dom f)     \tag{\axname{RST{\scriptsize 3}}}\label{eq:rst3}\\*
(\dom h)^\klstar\o f &\;= f^\klstar\o\dom(h^\klstar\o f)   \tag{\axname{RST{\scriptsize 4}}}\label{eq:rst4}
\end{align}
where $f\c X\to\IA Y$, $g\c X\to\IA Z$ and $h\c Y\to\IA Z$. 

For the rest of the section, let us fix an equational lifting monad $\BBT$, whose 
Kleisli category~$\BC_{\BBT}$ is thus a restriction category. We then collect miscellaneous 
facts about $\BBT$ for further reference.

\begin{lemma}\label{lem:dom-sum}
Given $f\c X\to TZ$, ${g\c Y\to TZ}$,
$\dom [f,g] = [(T\inl)\o\dom f,(T\inr)\o\dom g]$. 
\end{lemma}
\begin{proof}
By definition,
\begin{align*}
\dom [f,g] 
=&\; (T\fst)\o\tau\o\brks{\id,[f,g]} \\
=&\; (T\fst)\o\tau\o[\brks{\inl,f},\brks{\inr,g}]\\
=&\; [(T\fst)\o\tau\o\brks{\inl,f},(T\fst)\o\tau\o\brks{\inr,g}]\\
=&\; [(T\inl)\o(T\fst)\o\tau\o\brks{\id,f},(T\inr)\o(T\fst)\o\tau\o\brks{\id,g}] \\
=&\; [(T\inl)\o\dom f,(T\inr)\o\dom g]
\end{align*}
and we are done.
\end{proof}

\begin{lemma}\label{lem:dom-eta}
For any $f\c X\to TY$, $(T\eta)\o f = Tf\o(\dom f)$. 
\end{lemma}
\begin{proof}
Note that $Tf\o(\dom f) = (\eta\o f)^\klstar\o (\id^\klstar\o\eta\o f)$, and hence,
by~\ref{eq:rst4}, $Tf\o(\dom f) = (\dom\id)^\klstar\o\eta\o f = (\dom\id)\o f$.
Since $\dom\id = (T\fst)\o\tau\o\brks{\id,\id} = (T\fst)\o T\brks{\eta,\id} = T\eta$,
and we are done.
\end{proof}
\begin{lemma}
$\hat\tau^\klstar\o\tau\o\brks{T\fst,T\snd} = \id$.
\end{lemma}
\begin{proof}
Indeed, 
$\hat\tau^\klstar\o\tau\o\brks{T\fst,T\snd} = \hat\tau^\klstar\o T(T\fst\times\snd)\o\tau\o\Delta = 
\hat\tau^\klstar\o T(T\fst\times\snd)\o T\brks{\eta,\id}
=\hat\tau^\klstar\o  T\brks{\eta\o\fst,\snd}
=(\hat\tau\o (\eta\times\id))^\klstar = \eta^\klstar = \id. 
$
\end{proof}

\begin{lemma}\label{lem:dag-dom}
Suppose that $\BBT$ is equipped with 
an operator $(\argument)^\istar\c\BC(X,T(Y+X))\to\BC(X,TY)$ that satisfies 
\FIX, and \UNI. Then for any $f\c X\to T(Y+X)$, $f^\istar = (f^\klstar\o(\dom f^\istar))^\istar$.
\end{lemma}
\begin{proof}
Using~\UNI, $f^\istar = (T(\id + \eta)\o f)^\istar\o\eta$.
Next, using~\cref{lem:dom-eta},
\begin{align*}
f^\istar 
=&\; (f^\istar)^\klstar\o(\dom f^\istar)\\* 
=&\; ((T(\id + \eta)\o f^\klstar)^\istar\o\eta)^\klstar\o(\dom f^\istar)\\
=&\; ((T(\id + \eta)\o f^\klstar)^\istar)^\klstar\o(T\eta)\o(\dom f^\istar)\\ 
=&\; ((T(\id + \eta)\o f^\klstar)^\istar)^\klstar\o T(\dom f^\istar)\o\dom(\dom f^\istar)\\ 
=&\; ((T(\id + \eta)\o f^\klstar)^\istar\o (\dom f^\istar))^\klstar\o(\dom f^\istar).
\end{align*}
We then show that
\begin{align}\label{eq:dag-dom}
(T(\id + \eta)\o f^\klstar)^\istar\o (\dom f^\istar) = (f^\klstar\o(\dom f^\istar))^\istar.
\end{align}
This will entail the goal using \FIX as follows:
\begin{align*}
f^\istar 
=&\; ((T(\id + \eta)\o f^\klstar)^\istar\o (\dom f^\istar))^\klstar\o(\dom f^\istar)\\
=&\; ((f^\klstar\o(\dom f^\istar))^\istar)^\klstar\o(\dom f^\istar)\\
=&\; [\eta,(f^\klstar\o(\dom f^\istar))^\istar]^\klstar\o f^\klstar\o(\dom f^\istar)^\klstar\o(\dom f^\istar)\\
=&\; [\eta,(f^\klstar\o(\dom f^\istar))^\istar]^\klstar\o f^\klstar\o(\dom f^\istar)^\klstar\\
=&\; (f^\klstar\o(\dom f^\istar))^\istar.
\end{align*}
We show~\eqref{eq:dag-dom} by~\UNI after establishing another auxiliary property:
\begin{align*}
f^\klstar\o (\dom f^\istar) = [\eta\inl, (T\inr)\o (\dom f^\istar)]^\klstar\o f %^\klstar\o (\dom f^\istar).
\end{align*}
Indeed, by~\FIX and~\ref{eq:rst4},
\begin{flalign*}
&&f^\klstar\o (\dom f^\istar) 
=&\; f^\klstar\o \dom([\eta,f^\istar]^\klstar\o(\dom f))\\
&&=&\; (\dom[\eta,f^\istar])^\klstar\o f\\
&&=&\; (T\fst)\o (\tau\o\brks{\id,[\eta,f^\istar]})^\klstar\o f\\
&&=&\; (T\fst)\o [\tau\o\brks{\inl,\eta},\tau\o\brks{\inr,f^\istar}]^\klstar\o f\\
&&=&\; [(T\fst)\o\eta\o\brks{\inl,\id},(T\fst)\o\tau\o\brks{\inr,f^\istar}]^\klstar\o f\\
&&=&\; [\eta\o\inl,(T\inr)\o(\dom f^\istar)]^\klstar\o f.
\intertext{
Finally,
}
&&T(\id + \eta)\o f^\klstar&\o (\dom f^\istar) \\
&&=&\; T(\id + \eta)\o [\eta\inl, (T\inr)\o (\dom f^\istar)]^\klstar\o f\\
&&=&\; [\eta\inl, (T\inr)\o (T\eta)\o (\dom f^\istar)]^\klstar\o f\\
&&=&\; [\eta\inl, (T\inr)\o T(\dom f^\istar)\o (\dom f^\istar)]^\klstar\o f&&&\by{\cref{lem:dom-eta}}\\
&&=&\; T(\id + \dom f^\istar)\o f^\klstar\o (\dom f^\istar),
\end{flalign*}
which entails~\eqref{eq:dag-dom} by~\UNI.
\end{proof}

\subsection{Proof of \cref{prop:K-enriched}}
\begin{lemma}\label{lem:strict}
For any $f\c Z\to Y + Z$, 
\begin{align*}
\bigl(X\times Z\xto{(\eta\o\fst+\id)\o\dist\o(\id\times f)}\IA X + X\times Z\bigr)^\pistar~~\appr~~\bigl(X\times Z\xto{\eta\o\fst} \IA X\bigr).
\end{align*}
\end{lemma}
\begin{proof}
The goal is equivalent to 
\begin{align*}
((\eta\o\fst+\id)\o\dist\o(\id\times f))^\pistar
=(\IA\fst)\o\dom(((\eta\o\fst+\id)\o\dist\o(\id\times f))^\pistar).
\end{align*}
By transforming the right-hand expression as follows
\begin{align*}
(\IA\fst)&\o\dom\bigl(((\eta\o\fst+\id)\o\dist\o(\id\times f))^\pistar\o\bigr)\\*
&\;= (\IA\fst)\o (\IA\fst)\o\tau\o\brks{\id,((\eta\o\fst+\id)\o\dist\o(\id\times f))^\pistar}\\
&\;= \IA(\fst\fst)\o\bigl((\tau+\id)\o\dist\o(\id\times(\eta\o\fst+\id)\o\dist\o(\id\times f))\bigr)^\pistar\o\Delta\\
&\;= \bigl((\IA(\fst\fst)\o\tau\o(\id\times\eta\o\fst)+\id)\o\dist\o(\id\times\dist\o(\id\times f))\bigr)^\pistar\o\Delta\\
&\;= \bigl((\IA(\fst\fst)\o\eta(\id\times\fst)+\id)\o\dist\o(\id\times\dist\o(\id\times f))\bigr)^\pistar\o\Delta\\
&\;= \bigl((\eta\o\fst\o\fst+\id)\o\dist\o(\id\times\dist\o(\id\times f))\bigr)^\pistar\o\Delta
\intertext{Next,}
(\id+\fst&\times\snd)\o(\eta\o\fst\fst+\id)\o\dist\o(\id\times\dist\o(\id\times f))\\
&\;= (\eta\o\fst\o(\fst\times\snd)+\fst\times\snd)\o\dist\o(\id\times\dist\o(\id\times f))\\
&\;= (\eta\o\fst+\id)\o\dist\o(\fst\times(\snd+\snd)\o\dist\o(\id\times f))\\
&\;= (\eta\o\fst+\id)\o\dist\o(\fst\times f\snd)\\
&\;= (\eta\o\fst+\id)\o\dist\o(\id\times f)\o(\fst\times\snd),
\intertext{Hence, by uniformity,}
((\eta\o\fst+&\id)\o\dist\o(\id\times f))^\pistar\\
&\;=((\eta\o\fst+\id)\o\dist\o(\id\times f))^\pistar\o(\fst\times\snd)\o\Delta\\
&\;=((\eta\o\fst\fst+\id)\o\dist\o(\id\times\dist\o(\id\times f)))^\pistar\o\Delta\\
&\;=(\IA\fst)\o\dom ((\eta\o\fst+\id)\o\dist\o(\id\times f))^\pistar,
\end{align*}
and we are done.
\end{proof}

We are left to show that the order and the bottom elements are respected by the 
Kleisli composition.

\emph{Right monotonicity of composition:} Let $f\appr g$, i.e.\ $f = g\o(\dom f)$. Then 
$h\o g\o\dom(h\o f) = h\o g\o\dom(h\o g\o(\dom f))= 
h\o g\o\dom(h\o g)\o(\dom f) = h\o g\o(\dom f) = h\o f$, 
hence $h\o f\appr h\o g$. 

\emph{Left monotonicity of composition:} Analogously, if $f = g\o(\dom f)$ then
$g\o u\o\dom(f\o u) = g\o u\o\dom(g\o(\dom f)\o u) = g\o\dom(g\o(\dom f))\o u
=g\o(\dom g)\o(\dom f)\o u = g\o(\dom f)\o u = f\o u$, hence 
$f\o u\appr g\o u$.

\emph{Monotonicity of strength:} Suppose that $f\appr g$, i.e.\ $f=g^\klstar\o(\dom f)$. Then
\begin{align*}
\qquad(\tau\o\brks{\id, g}&)^\klstar\o\dom(\tau\o\brks{\id, f})\\*
=\;&(\tau\o\brks{\id, g})^\klstar\o(\IA\fst)\o\tau\o\brks{\id,\tau\o\brks{\id, f}}\\
=\;&(\tau\o\brks{\id, g})^\klstar\o(\IA\fst\fst)\o\tau\o\brks{\brks{\id,\id}, f}\\
=\;&(\tau\o\brks{\id, g})^\klstar\o(\IA\fst)\o\IA(\fst\times\id)\o\tau\o\brks{\brks{\id,\id}, f}\\
=\;&(\tau\o\brks{\id, g})^\klstar\o(\IA\fst)\o\tau\o\brks{\id, f}\\
=\;&(\tau\o(\id\times g)\fst)^\klstar\o\tau\o\brks{\brks{\id,\id}, f}\\
=\;&(\tau\o(\id\times g\o\fst))^\klstar\o\tau\o\brks{\id,\tau\o\brks{\id, f}}\\
=\;&(\tau\o(\id\times g))^\klstar\o\tau\o\brks{\id,(\IA\fst)\o\tau\o\brks{\id, f}}\\
=\;&(\tau\o(\id\times g))^\klstar\o\tau\o\brks{\id,\dom f}\\
=\;&\tau\o(\id\times g^\klstar)\o\brks{\id,\dom f}\\
=\;&\tau\o\brks{\id, g^\klstar\o(\dom f)}.
\intertext{\hspace{\parindent}\emph{Strictness of strength:}} 
\qquad\tau\o(\id\times\bot) 
=\;& \tau\o(\id\times\inr^\pistar)\\ 
=\;& ((\tau+\id)\o\dist\o(\id\times\inr))^\pistar\\ 
=\;& ((\tau+\id)\o\inr)^\pistar\\ 
=\;& \inr^\pistar\\
=\;& \bot.
\end{align*}

\emph{Right strictness of composition:} The equation $f^\klstar\o\bot=\bot$ follows from the fact that $f^\klstar$ preserves 
iteration. 

\emph{Left strictness of composition:} The equation $\bot^\klstar\o f = \bot$ is much more subtle. First, observe 
that $\bot\o f = \bot$, which easily follows by uniformity. We are left to show that 
${\bot^\klstar = \bot}$. Since $\bot\o\eta=\bot^\klstar\o\eta=\bot$, by definition of the 
lifting $\bot^\klstar$, it suffices to show that $\bot$ is iteration preserving,
i.e.\ for any $f\c X\to\IA Y + X$, $\bot\o f^\pistar = ((\bot + \id)\o f)^\pistar$,
equivalently, that ${((\bot + \id)\o f)^\pistar=\bot}$. Note that 
$((\bot + \id)\o f)^\pistar = ((\bot+\id)\o\dist\o(\id\times f))^\pistar\o\Delta$.
Indeed, 
\begin{align*}
&(\bot+\snd)\o\dist\o(\id\times f) = (\bot + \id)\o f\o\snd,
\intertext{hence, by uniformity,} 
&((\bot + \id)\o f)^\pistar=((\bot + \id)\o f)^\pistar\o\snd\Delta= 
((\bot+\id)\o\dist\o(\id\times f))^\pistar\o\Delta.
\end{align*}
Now, 
\begin{flalign*}
&&((\bot + \id)\o f)^\pistar
=\;& ((\bot+\id)\o\dist\o(\id\times f))^\pistar\o\Delta\\
&&=\;& ((\bot^\klstar\o\eta\o\fst+\id)\o\dist\o(\id\times f))^\pistar\o\Delta\\
&&=\;& \bot^\klstar\o((\eta\o\fst+\id)\o\dist\o(\id\times f))^\pistar\o\Delta\\
&&\appr\;&\bot^\klstar\o\eta\o\fst\o\Delta&\by{\cref{lem:strict}, monotonicity}\\
&&=\;& \bot.
\end{flalign*}
Therefore, indeed, $((\bot + \id)\o f)^\pistar = \bot$.
\qed

\subsection{Proof of \cref{prop:copy}}
For the identity, we have
\begin{align*}
\hat\tau^\klstar\o\tau\o\Delta  
=&\;\hat\tau^\klstar\o\IA\brks{\eta,\id}\\
=&\;(\hat\tau\o\brks{\eta,\id})^\klstar\\
=&\;(\eta\o\Delta)^\klstar\\
=&\;\IA\Delta.
\intertext{For the inequality, note that}
(\IA\fst)\o\hat\tau^\klstar&\o\tau\o\brks{f,g}\\* 
=&\; \fst^\klstar\o\tau\o\brks{f,g}\\ 
=&\; (\fst\o(f\times\id))^\klstar\o\tau\o\brks{\id,g}\\
=&\; f^\klstar\o(\IA\fst)\o\tau\o\brks{\id,g}\\
=&\;f^\klstar\o(\dom g).
\end{align*}
Hence, $(\IA\fst)\o\hat\tau^\klstar\o\tau\o\brks{f,g} = f^\klstar\o(\dom g) 
\appr f^\klstar\o\eta = f$.
\qed

\subsection{Proof of \cref{thm:klee}}
For the first clause we need to show that $f^{\bistar} = f^\pistar\o\fst\rest f^{\bistar}$. Since
\begin{align*}
(f^\pistar\o\fst\rest f^{\bistar})\o(x,\zero)
=&\; f^\pistar x\rest\bot\\
=&\; \bot,\\[1ex]
(f^\pistar\o\fst\rest f^{\bistar})\o(x,\suc n) 
=&\; \case{f(x)}{\inl r\mto r\rest f^{\bistar}(x,\zero) }{\inl y\mto f^\pistar\o(x) \rest f^{\bistar}(x,\suc n) }\\
=&\; \case{f(x)}{\inl r\mto r}{\inl y\mto f^\pistar\o(y) \rest f^{\bistar}(y, n)},
\end{align*}
$f^\pistar\o\fst\rest f^{\bistar}$ satisfies the definition for $f^{\bistar}$, and 
therefore we obtain the identity in question.

We proceed with the second clause. Suppose that $f^{\bistar}\appr g\o\fst$, 
%i.e.\ $f^\bistar\o\fst = g\rest f^\bistar\o\fst$ 
and show $f^\pistar\appr g$. 
%, i.e.\ 
%$f^\pistar = g\rest f^\pistar$. 
%
The idea is to introduce such $h\c X\times\nat\to\IA\nat + X\times\nat$ that 
$h^\pistar\o(x,\zero)$ runs like $f^\pistar\o(x)$, but instead of the final result of 
$f^\pistar\o(x)$ delivers the number of steps needed to reach the result. We then 
show that $f^\pistar = (f^\bistar)^\klstar\o\tau\o\brks{\id,h^\pistar\o\brks{\id,\zero\bang}}$,
which entails the desired property:
\begin{align*}
f^\pistar 
=&\; (f^\bistar)^\klstar\o\tau\o\brks{\id,h^\pistar\o\brks{\id,\zero\bang}}\\
\appr&\; (g\o\fst)^\klstar\o\tau\o\brks{\id,h^\pistar\o\brks{\id,\zero\bang}}\\
=&\; \fst^\klstar\o\tau\o\brks{g,h^\pistar\o\brks{\id,\zero\bang}}\\
=&\; g\rest h^\pistar\o\brks{\id,\zero\bang}\\
\appr&\; g.
\end{align*}
Thus we are left to produce the requisite $h$. Let 
$h = (\eta\o\snd + \id)\o\ldist\o(f\times\suc)$. We then have
\begin{flalign*}
&&(f^\bistar)^\klstar&\o\tau\o\brks{\id,h^\pistar\o\brks{\id,\zero\bang}}\\*
&&=&\;(f^\bistar)^\klstar\o\tau\o(\id\times h^\pistar)\o\brks{\id,\brks{\id,\zero\bang}}\\
&&=&\;(f^\bistar)^\klstar\o((\tau+\id)\o\dist\o(\id\times h))^\pistar\o\brks{\id,\brks{\id,\zero\bang}}&\by{\cref{prop:tau-prop}}\\
&&=&\;(((f^\bistar)^\klstar\o\tau+\id)\o\dist\o(\id\times h))^\pistar\o\brks{\id,\brks{\id,\zero\bang}}
\end{flalign*}
We are left to show that the latter is equal to $f^\pistar$. We strengthen the
goal slightly and show 
\begin{align}\label{eq:klee3}
(((f^\bistar)^\klstar\o\tau+\id)\o\dist\o(\id\times h))^\pistar\o\brks{\fst,\brks{w,\snd}} = f^\pistar\o w 
\end{align}
instead, where $w\c X\times\nat\to X$ is defined by primitive recursion as follows: 
\begin{align*}
w(x,\zero) = x,&& w(x,\suc n) = \case{f(x)}{\inl r\mto x}{\inr y\mto w(y,n)}.
\end{align*}
We will need the following facts:
\begin{align}
w(x,\suc n) =&\; \case{f(w(x,n))}{\inl r\mto w(x,n)}{\inr y\mto y},\label{eq:klee1}\\
f^\bistar(x,\suc n) =&\; \case{f(w(x,n))}{\inl r\mto r}{\inr y\mto y},\label{eq:klee2}
\end{align}
which both follow by induction. The former one follows from
\begin{align*}
\case{f(w(&x,\zero))}{\inl r\mto w(x,\zero)}{\inr y\mto y}
=w(x,\suc\zero),\\[1ex]
\case{f(w(&x,\suc n))}{\inl r\mto w(x,\suc n)}{\inr y\mto y}\\*
=&\;\case{f(x)}{\inl r\mto w(x,\suc n)}{\inr y\mto\\*&\qquad \case{f(w(y,n))}{\inl r\mto w(x,\suc n)}{\inr z\mto z}}\\
=&\;\case{f(x)}{\inl r\mto x}{\inr y\mto\\*&\qquad \case{f(w(y,n))}{\inl r\mto w(y, n)}{\inr z\mto z}}.
\end{align*}
To show~\eqref{eq:klee2}, observe that
\begin{align*}
f(w(x,\suc n)) =\;& \case{f(x)}{\inl r\mto f(x)}{\inr y\mto f(w(y,n))}\\*
 =\;& \case{f(x)}{\inl r\mto\inl r}{\inr y\mto f(w(y,n))}.
\end{align*}
Hence, 
\begin{align*}
\case{f(w(&x,\zero))}{\inl r\mto r}{\inr y\mto \bot}
=f(x,\suc\zero),\\[1ex]
\case{f(w(&x,\suc n))}{\inl r\mto r}{\inr y\mto y}\\
 =\;& \case{f(x)}{\inl r\mto r}{\inr y\mto\\&\qquad\case{f(w(y,n))}{\inl r\mto r}{\inr y\mto \bot}}.
\end{align*}
Now,
\begin{flalign*}
&&(((f^\bistar)^\klstar&\tau+\id)\o\dist\o(\id\times h))\o(x,(w(x,n),n))\\
&&=&\;\case{f(w(x,n))}{\inl r\mto\inl\o f^\bistar\o(x,\suc n)}{\inr y\mto\inr (x, (y, \suc n))}\\
&&=&\;\case{f(w(x,n))}{\inl r\mto\inl\o f^\bistar\o(x,\suc n)}{\inr y\mto\inr (x, (w(x,\suc n), \suc n))}&\by{\eqref{eq:klee1}}\\
&&=&\;\case{f(w(x,n))}{\inl r\mto\inl r}{\inr z\mto\inr (x, (w(x,\suc n), \suc n))}.&\by{\eqref{eq:klee2}}\\
&&=&\;((\snd+\fst)\o\dist\o\brks{\id\times\suc,f\o w})\o(x,n).
\end{flalign*}
Therefore, by~\UNI,
\begin{align}\label{eq:klee4}
(((f^\bistar)^\klstar\o\tau+\id)\o\dist\o(\id\times h))^\pistar\o\brks{\fst,\brks{w,\snd}} =
((\snd+\fst)\o\dist\o\brks{\id\times\suc,f\o w})^\pistar
\end{align}
Next,
\begin{flalign*}
\qquad((\id+w)\o((&\snd+\fst)\o\dist\o\brks{\id\times\suc,f\o w}))\o(x,n)\\
=&\;\case{f(w(x,n))}{\inl r\mto\inl r}{\inr y\mto\inr w(x,\suc n)}\\
=&\;\case{f(w(x,n))}{\inl r\mto\inl r}{\inr y\mto\inr y}&\by{\eqref{eq:klee1}}\\
=&\;f(w(x,n)),
\end{flalign*}
from which, again by~\UNI, we obtain
\begin{align}\label{eq:klee5}
((\snd+\fst)\o\dist\o\brks{\id\times\suc,f\o w})^\pistar = f^\pistar\o w.
\end{align}
By combining~\eqref{eq:klee4} with~\eqref{eq:klee5}, we obtain~\eqref{eq:klee3},
which completes the proof.
\qed
\begin{corollary}\label{cor:K-pre-Elgot}
Each $\IA X$ is a free Elgot algebra on $X$.
\end{corollary}
\begin{proof}
By \cref{prop:fold}, it remains to show~\FOL. Let 
$h\c Y \to X+Y$ and $f\c X\to \IA Z+X$. It easily follows by uniformity
that $[(\id + \inl)\o f\comma\inr\o h]^\pistar\o\inl = f^\pistar$. Then
\begin{align*}
[(\id + \inl)\o f\comma\inr\o h]^\pistar
=\;& [\id,[(\id + \inl)\o f\comma\inr\o h]^\pistar]\o[(\id + \inl)\o f\comma\inr\o h]\\
=\;& [f^\pistar,[(\id + \inl)\o f\comma\inr\o h]^\pistar\o h]\\
=\;& [\id,[(\id + \inl)\o f\comma\inr\o h]^\pistar]\o(f^\pistar+h),\\[1ex]
(f^\pistar+ h)^\pistar
=\;& [\id,(f^\pistar+ h)^\pistar]\o(f^\pistar+ h)\\
=\;& [[\id, f^\pistar]\o f\comma (f^\pistar+ h)^\pistar\o h]\\
=\;& [[\id,(f^\pistar+ h)^\pistar\o\inl]\o f\comma (f^\pistar+ h)^\pistar\o h]\\
=\;& [\id,(f^\pistar+ h)^\pistar]\o[(\id + \inl)\o f\comma\inr\o h].
\end{align*}
That is, both $[(\id + \inl)\o f\comma\inr\o h]^\pistar$ and 
$(f^\pistar+ h)^\pistar$ mutually satisfy the fixpoint identities of each other.
Hence, by \cref{thm:klee} they are mutually smaller under $\appr$, and hence
equal.
\end{proof}

\subsection{Proof of \cref{cor:klee}}
Suppose that $[\id, g]\o f\appr g$ for some $g\c X\to\IA Y$, i.e.\ $[\id, g]\o f = g^\klstar\o\dom ([\id, g]\o f)$.
This yields
\begin{align*}
\case{f(x)}{&\inl r\mto r}{\inr y\mto g(y)}=
\case{f(x)}{\inl r\mto g(x)\rest r}{\inr y\mto g(x)\rest g(y)} 
\end{align*}
This entails
\begin{align*}
g(x)\rest f^{\bistar}(x,\zero)  =&\; \bot, \\ 
g(x)\rest f^{\bistar}(x,\suc n) =&\; \case{f(x)}{\inl r\mto  r}{\inr y\mto g(x)\rest f^{\bistar}(y,n)}.
\end{align*}
By induction, $f^{\bistar} = g\o\fst\rest f^{\bistar}$, i.e.\ $f^{\bistar}\appr g\o\fst$, 
hence, by \cref{thm:klee}, $f^{\pistar}\appr g$.
\qed

\subsection{Proof of \cref{thm:K-pre-Elgot}}
We build on \cref{cor:K-pre-Elgot}. To show that $\IE$ is pre-Elgot, 
we are left to check the remaining property: $h^\klstar\o f^\pistar = 
((h^\klstar+\id)\o f)^\pistar$ where $f\c Z\to\IA X+Z$ and $h\c X\to\IA Y$.
This is, in fact a consequence of \COM:
\begin{align*}
((h^\klstar+\id)\o f)^\pistar
&\,=(((\inl\o h^\klstar)^\pistar+\id)\o f)^\pistar\\* 
&\,= ([(\id+\inl)\o\inl\o h^\klstar\comma\inr\o\inr]\o{[\inl\comma f]})^\pistar\inr\\
&\,= ((h^\klstar+\inr)\o{[\inl\comma f]})^\pistar\inr\\
&\,= ((h^\klstar+\id)\o f)^\pistar.
\end{align*}
Now, given any pre-Elgot monad $\BBT$, for every $X$ we define $\alpha_X\c\IA X\to TX$
as the unique Elgot algebra morphism such that $\alpha_X\o\eta_X=\eta_X$ by
\cref{cor:K-pre-Elgot}. Naturality of $\alpha_X$ in $X$ follows from 
the diagram:
\begin{equation*}
\begin{tikzcd}[column sep = 12ex,row sep = 1ex]
\IA X
  \rar["\IA f"]
  \ar[rd,"\alpha"'] & \IA Y\rar["\alpha"] & TY\\[1ex]
 & 
TX
  \ar[ur, "Tf"']\\
X
  \ar[uu,"\eta"] 
  \ar[ur,"\eta"']            
\end{tikzcd}
\end{equation*}
where we make use of the fact that both $\IA  f = (\eta\o f)^\klstar$ and 
$T f = (\eta\o f)^\klstar$ are iteration preserving since both $\IA$ and $T$
are pre-Elgot, hence both $\alpha\o (\IA f)$ and $(Tf)\o\alpha$ are 
iteration preserving. Analogously, since for every $f\c X\to\IA Y$, $f^\klstar$
is iteration preserving and $\alpha\o f^\klstar\o\eta=\alpha\o f = (\alpha\o f)^\klstar\o\alpha\o\eta$,
$\alpha\o f^\klstar = (\alpha\o f)^\klstar\o\alpha$, $\alpha$ respects 
Kleisli lifting.

Finally, to show that $\IE$ is initial strong pre-Elgot, we are left to show that 
for strong pre-Elgot $\BBT$, the induced $\alpha\c\IE\to\BBT$ respects strength,
i.e.\ $\alpha\o\tau = \tau\o(\id\times\alpha)$. Since $\alpha\o\tau\o(\id\times\eta) = 
\alpha\o\eta=\eta=\tau\o(\id\times\eta)=\tau\o(\id\times\alpha)\o(\id\times\eta)$,
we are done by stability of the $\IA X$.
\qed

\subsection{Proof of \cref{prop:uniform-iteration-eq}}
Let us show the first clause. First, let us check that given $a\c DA\to A$, the induced operator $(\argument)^\pistar$ indeed satisfies the axioms of uniform-iteration algebras.
\begin{itemize}
  \item \FIX Let every $f\c X\to A+X$. Then
  \begin{align*}
   f^\pistar 
   =&\; a\o(\coit f)\\
   =&\; a\o[\now,\lat]\o\out\o(\coit f)\\
   =&\; [a\o\now,a\o\lat]\o\out\o(\coit f)\\
   =&\; [\id,a]\o(\id+\coit f)\o f\\
   =&\; [\id, a\o(\coit f)]\o f\\
   =&\; [\id, f^\pistar]\o f,
  \end{align*}
  \item \UNI Let $f\c X\to A+X$, $g\c Y \to A+Y$ and $h\c X\to Y$, and 
  assume that $(\id+h)\o f = g\o h$. We have to show that $f^\pistar = g^\pistar\o h$.
  Indeed,
  \begin{align*}
  f^\pistar
  =&\; a\o(\coit f)\\ 
  =&\; a\o(\coit g)\o h\\
  =&\; g^\pistar\o h.
  \end{align*}
\end{itemize}  
Conversely, given a uniform-iteration algebra structure $(A,(\argument)^\pistar)$,
observe that $\out^\pistar\o\now = [\id,\out^\pistar]\o\out\o\now=
[\id,\out^\pistar]\o\inl=\id$ and $\out^\pistar\o\lat = 
[\id,\out^\pistar]\o\out\o\lat = [\id,\out^\pistar]\o\inr = \out^\pistar$.

Let us show that the given passages are mutually inverse. On the one hand,
${a\o(\coit\out) = a}$, and on the other hand, $\out^\pistar\o(\coit f)= f^\pistar$
by uniformity, since $\out\o(\coit f) = (\id+\coit f)\o f$. 

Finally, given two uniform-iteration algebras $(A,(\argument)^\pistar)$ and 
$(B,(\argument)^{\pistar'})$ and a morphism $h\c A\to B$, show that iteration 
preservation by $h$ is equivalent to being a $D$-algebra morphism. If $h$ is 
iteration preserving then $\out\o (Dh) = (h+Dh)\o\out$, which by uniformity 
entails $\out^\pistar\o (Dh) = ((h+\id)\o\out)^\pistar$. Therefore,
\begin{align*}
h\o\out^\pistar = ((h+\id)\o\out)^{\pistar'} = \out^\pistar\o (Dh), 
\end{align*}
i.e.\ $h$ is a $D$-coalgebra morphism. Conversely, if $h$ is a $D$-algebra morphism,
then for any $f\c X\to A+X$,
\begin{align*}
h\o f^\pistar 
=&\; h\o\out^\pistar\o(\coit f)\\*
=&\; ((h+\id)\o\out)^{\pistar'}\o(\coit f)\\
=&\; ((h+\id)\o f)^{\pistar'},
\end{align*}
where the last step is by uniformity, for 
$(h+\id)\o\out\o(\coit f) = (h+\id)\o (\id+\coit f)\o f = (\id+\coit f)\o (h+\id)\o f$.

Let us proceed with the second clause. %
In the guarded case, the requisite equivalence between Elgot algebras and 
$\BBD$-algebras is shown previously~\cite[Theorem 5.7]{GoncharovMiliusEtAl16}. We are left 
to check that the ``unguardedness condition'' $a = [\id,a]\o\out$ of 
search-algebras $(A, a\c DA\to A)$ corresponds to the requirement that $\hat a=\id$ on 
the respective $\Id$-guarded Elgot algebras $(A,\hat a\c A\to A,(\argument)^\pistar)$.
The involved connection between $a$ and $\hat a$ is precisely: $a=\out^\pistar$.
Now, if $\hat a = \id$ then $\out^\pistar=[\id,\out^\pistar]\o\out$ and if 
$\out^\pistar = [\id,\out^\pistar]\o\out$ then $\hat a = 
[\id,[\hat a,\hat a\o\hat a\o\out^\pistar]\o\out]\o\out\o\lat\o\now =
[\id,\hat a\o\out^\pistar]\o\out\o\lat\o\now=\out^\pistar\o\lat\o\now=
[\id,[\id,\out^\pistar]\o\out]\o\out\o\lat\o\now=\id$.
%
%Recall that a pair $(A,a\c DA\to A)$ is a $\BBD$-algebra iff it satisfies 
%$a\o\now = \id$ and $a^\klstar = a\o (Da)$. 
%%
%We are left to show that for 
%$a=\out^\pistar$ these conditions are equivalent to \FOL. The first condition is 
%automatic: $\out^\pistar\o\now = [\id,\out^\pistar]\o\out\o\now = [\id,\out^\pistar]\o\inl = \id$, 
%and we thus drop it from now on.
%%
%
%$\Rightarrow$: Suppose that $\out^\pistar\o\mu = \out^\pistar\o(D\out^\pistar)$.
%Let $h\c Y \to X+Y$, $f\c X\to A+X$, and show that $[(\id + \inl)\o f\comma\inr\o h]^\pistar=(f^\pistar+ h)^\pistar$.
%%
%\begin{align*}
%[(\id + \inl)\o f\comma\inr\o h]^\pistar
%=&\,([\id+\inl,\inr]\o (f+h))^\pistar\\
%=&\,\out^\pistar\o\mu\o\coit (\coit f +h)\\
%=&\,\out^\pistar\o(D\out^\pistar)\o\coit (\coit f +h)\\
%=&\,\out^\pistar\o\coit (\out^\pistar\o(\coit f) +h)\\
%=&\,(f^\pistar+ h)^\pistar
%\end{align*}
%
%
%$\Leftarrow$: Conversely:
%\begin{flalign*}
%\qquad\out^\pistar\o(D\out^\pistar)
%=&\,((\out^\pistar+\id)\o\out)^\pistar&\by{uniformity}\\
%=&\,(\out^\pistar+\out)^\pistar\o\out&\by{uniformity}\\
%=&\,[(\id + \inl)\o\out\comma\inr\o\out]^\pistar\o\out&\by{assumption}\\
%=&\,((\id+\out)\o[(\id + \tuo\o\inl)\o\out\comma\inr])^\pistar\o\out\\
%=&\, ([(\id+\now)\o\out,\inr]\o\out)^\pistar\\
%=&\,\out^\pistar\o\mu&\by{uniformity}
%\end{flalign*}
%%
%
\qed

\subsection{\texorpdfstring{$\protect\wave D$ is a strong functor}{Quotient of $D$ is a strong functor}} %{}
$\wave D$ is a strong functor and $\rho$ is a strong natural transformation.
\begin{proof}
Recall that $D$ is strong with a strength $\tau\c X\times DY\to D(X\times Y)$.
By the axioms of strength, $\tau\o(\id\times\iota^\klstar) = 
(\tau\o(\id\times\iota))^\klstar\o\tau$, and it is easy to obtain by coinduction 
that $\tau\o (\id\times\iota) = \iota\o\assoc^\mone$. In summary, 
$\tau\o(\id\times\iota^\klstar) = (\iota\o\assoc^\mone)^\klstar\o\tau = 
\iota^\klstar\o (D\assoc^\mone)\o\tau$.
Using the fact that~\eqref{eq:D-quote} is preserved by products, we introduce strength for~$\wave D$
as the universal map in
\begin{equation*}
\begin{tikzcd}
X\times D(Y\times\nat) 
  \ar[r,shift right=.35ex,"\id\times D\fst"']
  \ar[r,shift left=.75ex,"\id\times\iota^\klstar"]
  \dar["(D\assoc^\mone)\o\tau_{X,Y\times\nat}"']
&[3em]
X\times DY
  \rar["\id\times\rho_Y"] 
  \dar["\tau_{X,Y}"']
&[1em]
X\times \wave DY
  \dar[dotted]
\\[1ex]
D((X\times Y)\times\nat) 
  \ar[r,shift right=.35ex,"D\fst"']
  \ar[r,shift left=.75ex,"\iota^\klstar"]
&
D(X\times Y)\rar["\rho_{X\times Y}"] & \wave D(X\times Y)
\end{tikzcd}
\end{equation*}
The axioms of strength then follow automatically, as well as the fact that 
$\rho$ is a strong natural transformation.
\end{proof}

\begin{lemma}\label{lem:D-proj}
The morphism $\brks{D\fst,D\snd}\c D(X\times Y)\to DX\times DY$ is a section.
\end{lemma}
\begin{proof}
We define the requisite retraction $DX\times DY\to D(X\times Y)$ as the composition 
$w^\pistar\o\hat\tau^\klstar\o\tau\c DX\times DY\to D(X\times Y)$ where
\begin{align*}
w= \bigl(D(X\times Y) \xto{[\now\inl,\,\lat\now\out]\o\out} D(X\times Y + D(X\times Y)) \bigr).
\end{align*}
It is easy to check that $\out w$ factors through $(\inl+\id)$, hence the application
of the iteration operator is legit. We next show that $u=\hat\tau^\klstar\o\tau\o\brks{D\fst,D\snd}$
satisfies the following equation
\begin{displaymath}
  \out\o u = (\id+\lat u)\o\out. 
\end{displaymath}
From~\eqref{lem:d-comm}, recall that 
\begin{displaymath}
    \out\o \hat\tau^\klstar\o\tau = [\id+\hat\tau, \inr[\tau, \lat \hat\tau^\klstar\o\tau] ]\o (\ldist+\ldist)\o\dist\o(\out\times\out).
\end{displaymath}
Also, using the fact that $\dist$ and $\ldist$ are isomorphisms, it is easy to
show that $(\ldist+\ldist)\o\dist\o\Delta = \inl\Delta + \inr\o\Delta$. Therefore
\begin{flalign*}
&& \out u
&  \;= [\id+\hat\tau, \inr[\tau, \lat \hat\tau^\klstar\o\tau]]\\
&&    &\qquad\quad (\ldist+\ldist)\o\dist\o\brks{\out (D\fst), \out (D\snd)} &\\
&&&\;= [\id+\hat\tau, \inr[\tau, \lat \hat\tau^\klstar\o\tau]]\\
&&    &\qquad\quad (\ldist+\ldist)\o\dist\o\brks{\fst+D\fst, \snd +D\snd}\out \\
&&&\;= [\id+\hat\tau, \inr[\tau, \lat \hat\tau^\klstar\o\tau]]\\
&&    &\qquad\quad \bigl((\fst\times\snd+D\fst\times\snd) + (\fst\times D\snd+D\fst\times D\snd)\bigr)\\
&&    &\qquad\quad (\ldist+\ldist)\o\dist\o \Delta\out \\
&&&\;= [\id+\hat\tau, \inr[\tau, \lat \hat\tau^\klstar\o\tau]]\\
&&    &\qquad\quad \bigl((\fst\times\snd+D\fst\times\snd) + (\fst\times D\snd+D\fst\times D\snd)\bigr)\\
&&    &\qquad\quad (\inl+\inr)\out \\
&&&\;= [\id+\hat\tau, \inr[\tau, \lat \hat\tau^\klstar\o\tau]]\\
&&    &\qquad\quad \bigl(\inl\Delta\o (\fst\times\snd) + \inr\o\Delta\o(D\fst\times D\snd)\bigr)\out \\
&&&\;= [\inl, \inr\o\lat \hat\tau^\klstar\o\tau\o\brks{D\fst, D\snd}]\o\out \\
&&&\;= (\id+ \lat u)\out.
\intertext{
Now
}
&& w u
&  \;= [\now\inl,\,\lat\now\out]\o\out u  &\\
&&&\;= [\now\inl,\,\lat\now\out] (\id+\lat u)\o\out  \\
&&&\;= [\now\inl,\,\lat\now\out\o\lat u] \out  \\
&&&\;= [\now\inl,\,\lat\now\inr u] \out  \\
&&&\;= D(\id+ u)\o [\now\inl,\,\lat\now\inr] \out.
\end{flalign*}
By uniformity, $w^\pistar u = ([\now\inl,\,\lat\now\inr]\o\out)^\pistar$. But $\id$ 
satisfies the fixpoint equation for~$w^\pistar u$: 
\begin{align*}
[\now,\id]^\klstar\o [\now\inl,\,\lat\now\inr]\o\out = %&\; 
[\now, \lat [\now,\id]^\klstar \now\inr]\o\out %\\
= [\now, \lat]\o\out %\\
=\id.
\end{align*}
Therefore $w^\pistar u=\id$.
\end{proof}

\begin{lemma}\label{lem:tau-retr}
Let $(A,a)$ be a search-algebra. Then for any $X$, $\tau_{A,X}\c A\times DX\to D(A\times X)$
is a retraction, whose section is $\brks{a\o(D\fst), D\snd}\c D(A\times X)\to A\times DX$.
\end{lemma}
\begin{proof}
Using \cref{lem:D-proj}, we proceed to show that 
\begin{align*}
\brks{D\fst, D\snd} = \brks{D\fst, D\snd}\o\tau_{A,X}\o\brks{a\o(D\fst), D\snd},
\end{align*}
which is equivalent to $D\fst=(Da)\o(D\fst)\o\tau\o\brks{D\fst,D\snd}$
and $D\snd=D\snd$. To obtain the former equation, we show commutativity of the 
diagram

\begin{equation*}
\begin{tikzcd}[column sep=6ex, row sep=normal]
D(A\times X)
  \dar["\brks{D\fst,D\snd}"']
  \ar[rr,"\out"] &[2ex]&
A \times X + D(A\times X)
  \dar["\now\times\id +\brks{\lat\o (D\fst),D\snd}"]\\
DA\times DX 
  \dar["\tau"']
  \rar["\id\times\out"] & DA\times (X + DX)\rar["\ldist"] & 
DA\times X + DA\times DX
  \dar["\id+\tau"] \\
D(DA\times X)  
  \dar["D\fst"']
  \ar[rr,"\out"] && 
DA\times X + D(DA\times X)
  \dar["\fst + D\fst"]\\
DDA 
  \dar["Da"']
  \ar[rr,"\out"]&& 
  \dar["a + Da"] DA+DDA\\
DA 
  \ar[rr,"\out"] && 
A + DA
\end{tikzcd}
\end{equation*}
Except for the top cell, the remaining ones commute by definition. Assume 
commutativity of the top cell for the time being. Then the composition on the 
right vertical edge is easily seen to be equal to 
$\fst + (Da)\o(D\fst)\o\tau\o\brks{D\fst,D\snd}$. The resulting diagram 
then witnesses the fact that $(Da)\o(D\fst)\o\tau\o\brks{D\fst,D\snd}$
is a coalgebra morphism, which must be equal to~$D\fst$ by finality of the 
coalgebra $(DA,\out)$. 

Finally, let us show commutativity of the top cell. Using the fact that $\out$
is an isomorphism, this amounts to showing
\begin{align*}
\ldist\o(\id\times\out)\o\brks{D\fst,D\snd}\o\tuo\o\inl &\,=\inl\o(\now\times\id),\\
\ldist\o(\id\times\out)\o\brks{D\fst,D\snd}\o\tuo\o\inr &\,=\inr\o\brks{\lat\o (D\fst),D\snd}.
\end{align*}
Indeed,
\begin{align*}
\ldist\o(\id&\times\out)\o\brks{D\fst,D\snd}\o\tuo\o\inl\\
 &\;= \ldist\o (\id\times\out)\o\brks{D\fst,D\snd}\o\now\\
 &\;= \ldist\o (\id\times\out)\o\brks{\now\o\fst,\now\o\snd}\\
 &\;= \ldist\o \brks{\now\o\fst,\inl\o\snd}\\
 &\;= \inl\o \brks{\now\o\fst,\snd}\\
 &\;= \inl\o (\now\times\snd),\\[1ex]
\ldist\o(\id&\times\out)\o\brks{D\fst,D\snd}\o\tuo\o\inr\\
 &\;= \ldist\o(\id\times\out)\o\brks{D\fst,D\snd}\o\lat\\
 &\;= \ldist\o(\id\times\out)\o\brks{\lat\o (D\fst), \lat\o (D\snd)}\\
 &\;= \ldist\o\brks{\lat\o (D\fst), \inr\o (D\snd)}\\
 &\;= \inr\o\brks{\lat\o (D\fst),\o D\snd},  
\end{align*}
and we are done.
\end{proof}

\subsection{Proof of \cref{lem:elgot-iota}}
Indeed, consider the commutative diagram
\begin{equation*}
\begin{tikzcd}[column sep = 8ex,row sep = 4ex]
D(A\times\nat)
  \rar["D\iota"]
  \drar["D\fst"'] & 
DDA
  \dar["Da"]
  \rar["\mu"] & 
X\times \nat 
  \dar["a"]\\
&Y
  \rar["a"] & Y
\end{tikzcd}
\end{equation*}
where the square commutes, since $(DA,a)$ is a $\BBD$-algebra (\cref{prop:uniform-iteration-eq}) and the triangle commutes, because $a\o\iota = \fst$, which easily follows 
by induction.
\qed

\subsection{Proof of \cref{lem:bis}}
First, note that 
\begin{flalign*}
&& \rho\o {[\eta,\,[\eta\o\fst,\,\iota\o (\id\times\suc)]]}^\klstar 
&\;=   \rho\o \iota^\klstar\o {[\eta\o\brks{\id,\zero\o\bang},\,[\eta\o\brks{\id,\zero\o\bang}\o\fst,\,\id\times\suc]]}^\klstar&\\
&&&\;= \rho\o D\fst\o {[\eta\o\brks{\id,\zero\o\bang},\,[\eta\o\brks{\id,\zero\o\bang}\o\fst,\,\id\times\suc]]}^\klstar\\
&&&\;= \rho\o{[\eta,\,[\eta\o\fst,\,\eta\o\fst]]}^\klstar.
\end{flalign*}
By combining it with a symmetric argument, we obtain that 
$\rho\o {[\eta,\,[\eta\o\fst,\,\iota\o(\id\times\suc)]]}^\klstar =
\rho\o{[\eta,\,[\eta\o\fst,\,\eta\o\fst]]}^\klstar = 
\rho\o{[\eta,\,[\iota\o(\id\times\suc),\,\eta\o\fst]]}^\klstar$. 

Now, given a morphism $f\c DX\to A$ that equalizes $\rho\o {[\eta,\,[\eta\o\fst,\,\iota\o(\id\times\suc)]]}^\klstar$
and $\rho\o{[\eta,\,[\iota\o(\id\times\suc),\,\eta\o\fst]]}^\klstar$, we obtain that
$f\o\iota^\klstar = f\o {[\eta,\,[\eta\o\fst,\,\iota\o(\id\times\suc)]]}^\klstar\o D\inr\inr = f\o{[\eta,\,[\iota\o(\id\times\suc),\,\eta\o\fst]]}^\klstar\o D\inr\inr = f\o D\fst$,
by assumption, there is a unique factorization of $f$ through $\rho$.
\qed

\subsection{Proof of \cref{thm:Drho}}
We show the following implications: 
{\bfseries\sffamily\ref{it:q2}.$\impl$\ref{it:q1}.$\impl$\ref{it:q4}.$\impl$\ref{it:q2}.}
and 
{\bfseries\sffamily\ref{it:q4}.$\impl$\ref{it:q3}.$\impl$\ref{it:q1}.}

  {\bfseries\sffamily\ref{it:q2}.$\impl$\ref{it:q1}.} 
Using the assumption, we define $\alpha_X\c D\wave DX\to\wave D X$ by the 
  universal property:
\begin{equation*}
\begin{tikzcd}[column sep = 10ex,row sep = 4ex]
DD(X\times\nat) 
  \ar[r,shift right=.35ex,"DD\fst"']
  \ar[r,shift left=.75ex,"D\iota^\klstar"]
&[3ex]
DDX
  \rar["D\rho_{X}"] 
  \ar[d, "\mu_X"']
&
D\wave DX
  \dar[dotted, "\alpha_X"]
\\
&
DX\rar["\rho_X"]
&
\wave DX
\end{tikzcd}
\end{equation*}
where we call on the calculation:
\begin{align*}
\rho\o\mu\o D\iota^\klstar =
\rho\o\mu\o D\mu\o DD\iota =  
\rho\o\mu\o\mu\o DD\iota =  
\rho\o\mu\o D\iota\o\mu =  
\rho\o\iota^\klstar\o\mu =  
\rho\o D\fst\o\mu =  
\rho\o\mu\o DD\fst.  
\end{align*}
By definition, $(\wave DX,\alpha_X)$ is a $D$-algebra and $\rho_X$ is a $D$-algebra 
morphism. Let us show that $(\wave DX,\alpha_X)$ is a search-algebra, i.e.\ that 
$\alpha_X\o\now = \id$ and $\alpha_X\o\lat = \alpha_X$. For the first equation,
note that $\alpha_X\o\now\o\rho = \alpha_X\o D\rho\o\now = \rho\o\mu\o\now=\rho$,
which entails $\alpha_X\o\now = \id$ using the fact that $\rho$ is an epi.
Analogously, $\alpha_X\o\lat\o D\rho = \alpha_X\o D\rho\o\lat = \rho\o\mu\o\lat =
\rho\o\lat\o\mu = \rho\o\mu = \alpha_X\o D\rho$ using Lemmas~\ref{lem:later} 
and~\ref{lem:rho-eq-lat}, and again, we are done by discarding $D\rho$, which 
is epic by assumption.

  {\bfseries\sffamily\ref{it:q1}.$\impl$\ref{it:q4}.} 
Let $(\wave DX,\alpha_X\c D\wave DX\to\wave D X)$ be a search-algebra structure, 
which exists by assumption. To show that $\wave DX$ is an Elgot algebra, by  \cref{prop:uniform-iteration-eq}, it suffices to show that 
  it is a $\BBD$-algebra, i.e.\ $\alpha_X\o\now = \id$, which is by assumption,
  and $\alpha_X\o\mu = \alpha_X\o D\alpha_X$, which we can prove, using the 
  assumption that $\rho$ is a $D$-algebra morphism and assuming for the 
  time being that $DD\rho$ is epic, as follows: 
  $\alpha_X\o\mu\o DD\rho=\alpha_X\o D\rho\o\mu= 
  \rho\o\mu\o\mu = 
  \rho\o\mu\o D\mu =
  \alpha_X\o D\rho\o D\mu =
  \alpha_X\o D\alpha_X\o DD\rho$. The proof that $DD\rho$ is epic is entailed 
  by the following commutative diagram:
\begin{equation*}
\begin{tikzcd}[column sep=16ex, row sep=4ex]
DX\times D\enat
  \rar["\rho\times\id"]
  \dar["\tau"'] & 
\wave DX\times D\enat
  \dar["\tau"]  & \\
D(DX\times\enat)
  \rar["D(\rho\times\id)"]
  \dar["D\tau"'] & 
D(\wave DX\times\enat)
  \dar["D\tau"]  & \\
DD(DX\times 1)
  \rar["DD(\rho\times\id)"] &
DD(\wave DX\times 1)
\end{tikzcd}
\end{equation*}
where, up to the obvious isomorphisms, our morphism of interest is the horizontal bottom one. 
To show that it is epic, it suffices to show  that any path from the left top corner to the right bottom corner is epic, specifically, we consider the composition $D\tau\o\tau\o (\rho\times\id)$. 
This is epic, because $\rho\times\id$ is a coequalizer and the involved $\tau$ are 
retractions by \cref{lem:tau-retr}.

To prove that $\rho = ((\rho\o\now+\id)\o\out)^\pistar$, note that,
by definition, $((\rho\o\now+\id)\o\out)^\pistar=\alpha_X\o\coit ((\rho\o\now+\id)\o\out)$.
It is easy to see by the universal property of $\coit$ that 
$\coit ((\rho\o\now+\id)\o\out) = D(\rho\o\now)$. Hence $((\rho\o\now+\id)\o\out)^\pistar
 = \alpha_X\o D(\rho\o\now)$, which is equal to $\rho\o\mu\o (D\now) =\rho$,
 since, by assumption, $\rho$ is a $D$-algebra morphism.
 
Finally, let us show freeness. Given an Elgot algebra $A$ and $f\c X\times Y\to A$, we provide a 
  unique right iteration preserving $f^{\hash}\c X\times\wave DY\to A$ such that $f=f^\hash\o(\id\times\rho\o\now)$.
  Using \cref{prop:uniform-iteration-eq}, again, we assume a search-algebra $(A,a)$ such that ${a\c DA\to A}$ is a $\BBD$-algebra. We define $f^{\hash}$ by a universal property 
  from the diagram:
\begin{equation*}
\begin{tikzcd}[column sep=6ex, row sep=3ex]
X\times D(Y\times\nat) 
  \ar[r,shift right=.35ex,"\id\times D\fst"']
  \ar[r,shift left=.75ex,"\id\times\iota^\klstar"]
&[3em]
X\times DY
  \rar["\id\times\rho_Y"] 
  \dar["(Df)\o\tau"']
&[1em]
X\times \wave DY
  \dar[dotted, "f^\hash"]
\\[1ex]
&
DA
  \rar["a"] & 
A
\end{tikzcd}
\end{equation*}
which is justified by the following calculation:
\begin{flalign*}
&&a\o (Df)\o\tau\o(\id\times\iota^\klstar)
\,&= a\o\iota^\klstar\o D(f\times\id)\o D\assoc^\mone\o\tau\\
&&\,&=a\o D\fst\o D(f\times\id)\o D\assoc^\mone\o\tau&\by{\cref{lem:elgot-iota}}\\
&&\,&=a\o Df\o D\fst\o D\assoc^\mone\o\tau\\
&&\,&=a\o Df\o D(\id\times\fst)\o\tau\\
&&\,&=a\o (Df)\o D(\id\times \fst)\o\tau\\
&&\,&=a\o (Df)\o\tau\o (\id\times D\fst)
\intertext{
We then immediately have 
}
&&f^\hash\o(\id\times\rho\o\now)
\,&= a\o (Df)\o\tau\o (\id\times\now)\\
&&\,&= a\o (Df)\o\now\\
&&\,&= a\o\now\o f\\
&&\,&= f.
\end{flalign*}
Let us show that $f^\hash$ is right iteration preserving, i.e.\ given 
$g\c Z\to\wave DY+Z$, 
\begin{align*}
f^\hash\o(\id\times \alpha_Y\o(\coit g)) = 
a\o\coit((f^\hash+\id)\o\dist\o(\id\times g)).
\end{align*}
First, we show that 
\begin{align*}
f^\hash\o(\id\times\alpha_Y) = a\o (Df^\hash)\o\tau. 
\end{align*}
To that end we compose both sides with $\id\times D\rho$, and make use the fact that 
it is an epi\sgnote{prove}.
\begin{align*}
f^\hash\o(\id\times\alpha_Y)\o (\id\times D\rho)
\,&= f^\hash\o(\id\times\rho\o\mu)\\
\,&= a\o (Df)\o\tau\o(\id\times\mu)\\
\,&= a\o (Df)\o\mu\o (D\tau)\o\tau\\
\,&= a\o \mu\o(DDf)\o (D\tau)\o\tau\\
\,&= a\o (Da)\o(DDf)\o (D\tau)\o\tau\\
\,&= a\o (Df^\hash)\o D(\id\times\rho)\o\tau\\
\,&= a\o (Df^\hash)\o\tau\o(\id\times D\rho).
\end{align*}
This reduces the goal to
\begin{align*}
(Df^\hash)\o\tau\o(\id\times\coit g) = \coit((f^\hash+\id)\o\dist\o(\id\times g)),
\end{align*}
and the latter follows from the fact that the left hand side satisfies the characteristic
equation for the right hand side:
\begin{align*}
\out\o(D&f^\hash)\o\tau\o(\id\times\coit g)\\
\,&=(f^\hash+Df^\hash)\o\out\o\tau\o(\id\times\coit g)\\
\,&=(f^\hash+Df^\hash)\o(\id+\tau)\o\dist\o(\id\times\out)\o(\id\times\coit g)\\
\,&=(f^\hash+Df^\hash)\o(\id+\tau)\o\dist\o(\id\times (\id+\coit g)\o g)\\
\,&=(f^\hash+Df^\hash)\o(\id+\tau\o (\id\times\coit g))\o\dist\o(\id\times g)\\
\,&=(\id+(Df^\hash)\o\tau\o (\id\times\coit g))\o (f^\hash+\id)\o\dist\o(\id\times g).
\end{align*}
Finally, let $g\c X\times\wave DY\to A$ be right iteration preserving, such 
that $f=g\o(\id\times\rho\o\now)$ and show that $g=f^\hash$. By definition of 
$f^\hash$, we need to show that $g\o (\id\times\rho) = a\o (Df)\o\tau$. Using 
the equation $\rho = ((\rho\o\now+\id)\o\out)^\pistar$, we proved above,
we derive the goal as follows: 
\begin{align*}
g\o (\id\times\rho)
\,&= g\o (\id\times((\rho\o\now+\id)\o\out)^\pistar)\\
\,&= ((g+\id)\o\dist\o(\id\times(\rho\o\now+\id)\o\out))^\pistar\\
\,&= ((g\o(\id\times\rho\o\now)+\id)\o\dist\o(\id\times\out))^\pistar\\
\,&= ((f+\id)\o\dist\o(\id\times\out))^\pistar\\
\,&= a\o\coit((f+\id)\o\dist\o(\id\times\out))\\
\,&= a\o(Df)\o\coit(\dist\o(\id\times\out))\\
\,&= a\o(Df)\o\tau.
\end{align*}

  {\bfseries\sffamily\ref{it:q4}.$\impl$\ref{it:q2}.} 
  Let $(\wave DX\comma\alpha_X\c D\wave DX\to\wave DX)$ be the relevant Elgot
  algebra structure, which exists by definition.
Let $w\c DX\times\enat\to DX+DX\times\enat$ be as follows:
\begin{align*}
w(p,\now\star)  = \inl p&&
w(p,\lat n) = \inr(\ear p, n).  
\end{align*}
Analogously, let $u\c D(X\times\nat)\times\enat\to D(X\times\nat)+D(X\times\nat)\times\enat$
be as follows:
\begin{align*}
u(p,\now\star)   =&\; \inl p,&
u(\lat p,\lat n) =&\; \inr(p, n),&\\[1ex]
u(\now(x,\zero),\lat n) =&\; \inr(\now(x,\zero),n),&
u(\now(x,\suc(k)),\lat n) =&\; \inr(\now(x,k),n).  
\end{align*}
We thus obtain two morphisms: $\coit w\c DX\times\enat\to DDX$ and 
$\coit u\c D(X\times\nat)\times\enat\to DD(X\times\nat)$. To build intuition, 
let us replace $X$ with $1$. Then $\coit w\c\enat\times\enat\to D\enat$ essentially
computes truncated difference: it subtracts a possibly infinite second number 
$m$ from a possibly infinite first number $n$ and produces a process $D\enat$, 
which runs $m$ time units, and in case of termination returns the truncated 
difference $n\dotdiv m$. Since subtraction is inverse to summation, this explains 
why $\coit w$ is a section, whose retraction is $\brks{\mu,D\bang}\c D\enat\to\enat\times\enat$,
which remains true for arbitrary $X$. The morphism $\coit u\c D\nat\times\enat\to DD\nat$
refines $\coit w$ in the following sense. The first argument can be regarded
as a sum of a possibly infinite~$n$ with a finite $k$ (so that $n+k=n$ if $n$ is 
infinite) and then $\coit u$ again computes the truncated difference $n+k\dotdiv m$
in the form of a formal sum $(n\dotdiv m) + k$ if $n$ is greater than~$m$ and $(n+k)\dotdiv m$
otherwise.

Now, consider the following diagram, which summarizes the argument.
\begin{equation*}
\begin{tikzcd}[column sep=12ex, row sep=normal]
& DDX
  \rar["D\rho_X"]
  \dar["\brks{\mu,D\bang}"'] & 
D\wave DX
  \dar["\brks{\alpha\comma D\bang}"] \\
D(X\times\nat)\times\enat
  \dar["\coit u"'] 
  \ar[r,shift right=.35ex,"D\fst\times\id"']
  \ar[r,shift left=.75ex, "\iota^\klstar\times\id"]
&[3em]
DX\times\enat\rar["\rho_X\times\id"]
\dar["\coit w"'] 
& 
\wave DX\times\enat
  \dar["(D\fst)\o\tau"]
  \ar[ddr,bend left=35, start anchor=east, "c"]\\
DD(X\times\nat) 
  \ar[r,shift right=.35ex,"DD\fst"']
  \ar[r,shift left=.75ex, "D\iota^\klstar"]
&[3em]
DDX
  \rar["D\rho_X"] 
  \ar[drr,bend right=20, "a"]
& 
D\wave DX
  \drar[dotted, "b"]\\[-2ex]
& & &[-8ex] Y
\end{tikzcd}
\end{equation*}
We would like to show that $D\rho$ is a coequalizer of the bottom parallel 
pair of morphisms. To that end, we fix $a\c D\wave DX\to Y$, such that $a\o\iota^\klstar = a\o (D\fst)$
and construct such ${b\c D\wave DX\to Y}$ that $a = b\o (D\rho)$. Assuming,
for the moment, that 
all rectangular cells (with coherently chosen edges of the corresponding parallel 
pairs) commute, we obtain that $a\o (\coit w)$ coequalizes~$\iota^\klstar,D\fst$,
which produces a suitable $c$, by a coequalizer property. Let $b = c\o\brks{\alpha,D\bang}$
and using the fact that the two vertical morphisms from $DDX$ to $DDX$ and from 
$D\wave DX$ to $D\wave DX$ are identities, obtain the desired equation $b\o D\rho = a$.
The fact that the vertical morphism from $D\wave DX$ to $D\wave DX$ is the identity 
is by \cref{lem:tau-retr}.
The fact that the vertical morphism from $DDX$ to $DDX$ is the identity, we
show directly. We will show that
\begin{align*}
\out\o (\coit w)\o\brks{\mu,D\bang} = (\id+ (\coit w)\o\brks{\mu,D\bang})\o\out
\end{align*}
This identifies $(\coit w)\o\brks{\mu,D\bang}$ as a unique final coalgebra morphism,
which thus must be equal to~$\id$. Since
\begin{align*}
\out\o (\coit w)\o\brks{\mu,D\bang} 
= (\id+\coit w)\o w\o\brks{\mu,D\bang},
\end{align*}
we reduce the previous equation to
\begin{align*}
w\o\brks{\mu,D\bang} = (\id+ \brks{\mu,D\bang})\o\out
\end{align*}
By composing both sides with $\now$ and $\lat$ correspondingly,
we reduce to
\begin{align*}
w\o\brks{\id,\now\o\bang} = \inl,
&&
w\o\brks{\mu\o\lat,(D\bang)\o\lat} = \inr\o \brks{\mu,D\bang}.
\end{align*}
The first equation directly follows by definition of $w$. For the second equation,
$w\o\brks{\mu\o\lat,(D\bang)\o\lat}=w\o\brks{\lat\o\mu,\lat\o(D\bang)}
=\inr\o\brks{\mu,D\bang}$, using \cref{lem:later}~(1).

We proceed to show non-trivial commutativity conditions for the square cells 
of our diagram. These are the following:
\begin{align}
\label{eq:Drho1}
(\coit w)\o (\iota^\klstar\times\id) =&\; D\iota^\klstar\o (\coit u)\\
\label{eq:Drho0}
(\coit w)\o (D\fst\times\id) =&\; (DD\fst)\o (\coit u)\\
\label{eq:Drho2}
(D\rho)\o(\coit w) =&\; (D\fst)\o\tau\o (\rho\times\id)
\end{align}
To obtain~\eqref{eq:Drho1}, we show that both sides are equal to 
$\coit ((\iota^\klstar+\id)\o u)$, which in turn amounts to proving that the left and 
the right hand are both universal coalgebra maps, i.e.\
\begin{align*}
\out\o (D\iota^\klstar)\o (\coit u) =&\; (\id+(D\iota^\klstar)\o\coit u)\o(\iota^\klstar+\id)\o u\\
\out\o(\coit w)\o (\iota^\klstar\times\id) =&\;
 (\id+(\coit w)\o (\iota^\klstar\times\id))\o(\iota^\klstar+\id)\o u
\end{align*}
The first one is obvious, and we proceed with the second one. Since
\begin{align*}
\out\o(\coit w)\o (\iota^\klstar\times\id) =&\;
 (\id+\coit w)\o w\o (\iota^\klstar\times\id),
\end{align*}
we are left to show that
\begin{align*}
(\iota^\klstar+\iota^\klstar\times\id)\o u = w\o (\iota^\klstar\times\id).
\end{align*}
By case distinction:
\begin{align*}
(\iota^\klstar+\iota^\klstar\times\id)\o(u(p,\now\star)) 
=&\; (\iota^\klstar+\iota^\klstar\times\id)\o(\inl p)\\
=&\; \inl(\iota^\klstar(p))\\ 
=&\;w (\iota^\klstar(p),\now\star),\\[1ex]
(\iota^\klstar+\iota^\klstar\times\id)\o(u(\lat p,\lat n)) =&\; 
(\iota^\klstar+\iota^\klstar\times\id)\o(\inr(p,n))\\
=&\; \inr(\iota^\klstar(p),n)\\
=&\;\inr(\ear(\iota^\klstar (\lat p)), n)\\
=&\;w\o (\iota^\klstar (\lat p),\lat n),\\[1ex]
(\iota^\klstar+\iota^\klstar\times\id)\o(u(\now(x,\zero),\lat n)) =&\; 
(\iota^\klstar+\iota^\klstar\times\id)\o(\inr(\now(x,\zero),n))\\
=&\;\inr(\iota^\klstar(\now(x,\zero)),n)\\
=&\;\inr(\iota(x,\zero),n)\\
=&\;\inr(\now x,n)\\
=&\;\inr(\ear(\now x),n)\\
=&\;\inr(\ear(\iota^\klstar (\now(x,\zero))), n)\\
=&\;w\o (\iota^\klstar (\now(x,\zero)),\lat n),\\[1ex]
(\iota^\klstar+\iota^\klstar\times\id)\o(u(\now(x,\suc(k)),\lat n)) =&\; 
(\iota^\klstar+\iota^\klstar\times\id)\o(\inr(\now(x,k),n))\\
=&\; \inr(\iota^\klstar(\now(x,k)),n)\\
=&\; \inr(\iota(x,k),n)\\
=&\; \inr(\ear(\iota(x,\suc(k))),n)\\
=&\;\inr(\ear(\iota^\klstar (\now(x,\suc(k)))), n)\\
=&\;w\o (\iota^\klstar (\now(x,\suc(k))),\lat n).
\end{align*}
The proof of~\eqref{eq:Drho0} runs analogously. To obtain~\eqref{eq:Drho2}, again, 
we show that both sides are equal to $\coit ((\rho+\id)\o w)$, by establishing equations
\begin{align*}
\out\o (D\rho)\o (\coit w) =&\; (\id+(D\rho)\o\coit w)\o(\rho+\id)\o w,\\*
\out\o(D\fst)\o\tau\o (\rho\times\id) =&\;
 (\id+(D\fst)\o\tau\o (\rho\times\id))\o(\rho+\id)\o w.
\end{align*}
The first equation is again easy to see. Since 
$\out\o(D\fst)\o\tau\o (\rho\times\id)
=(\fst+(D\fst)\o\tau)\o\dist\o(\rho\times\out)
=(\rho\o\fst+(D\fst)\o\tau\o(\rho\times\id))\o\dist\o(\id\times\out)$,
we reduce~to
\begin{align*}
(\fst+\rho\times\id)\o\dist\o(\id\times\out) = (\id+\rho\times\id)\o w.
\end{align*}
This is however easy to see by definition of $w$.

Finally, we have to show that $b$ is the unique morphism for which $a=b\o (D\rho)$.
We obtain this by proving that $D\rho$ is an epi. To this end, consider the 
following diagram
\begin{equation*}
\begin{tikzcd}[column sep=12ex, row sep=normal]
DX\times\enat
  \rar["\rho\times\id"]
  \dar["\tau"] & 
\wave DX\times\enat\dar["\tau"]  & \\
D(DX\times 1)
  \rar["D(\rho\times\id)"] &
D(\wave DX\times 1)
  \ar[r,shift right=.35ex,"f"']
  \ar[r,shift left=.75ex,"g"] &
Y                
\end{tikzcd}
\end{equation*}
where we assume that $f\o D(\rho\times\id) = g\o D(\rho\times\id)$.
Then, also $f\o\tau\o (\rho\times\id) = g\o\tau\o(\rho\times\id)$.
Since $\rho\times\id$ is an epi, and $\tau$ is a retraction by 
\cref{lem:tau-retr}, their composition is epic, hence $f=g$. We have thus shown
that $D(\rho\times\id)$ is epic, and since $DY\times 1\iso DY$, so 
is~$D\rho$.

  {\bfseries\sffamily\ref{it:q4}.$\impl$\ref{it:q3}.}
The proof that $\wave D$ extends to a monad is analogous to that of 
\cref{lem:b-mon}. Strength is defined and characterized in the same way 
as in \cref{prop:tau-prop}. We are left to show that~$\rho$ is a strong 
monad morphism. By definition, $\rho$ respects monad unit. Let us show that it 
respects multiplication, i.e.\ $\rho\o\mu = \mu\o\rho\o D\rho$. By assumption,
every $\rho_X$ is a $D$-algebra morphism, i.e.\ $\rho_X\o\mu_X = \out^\pistar\o D\rho_X$
(using the definition of the $D$-algebra structure for $\wave DX$ from 
\cref{prop:uniform-iteration-eq}). Thus, we are left to show that 
$\mu\o\rho = \out^\pistar$. By assumption, $\rho = 
((\rho\o\now +\id)\o\out)^\pistar$, and $\mu\c\wave D\wave DX\to\wave DX$ is 
iteration preserving, hence 
\begin{align*}
\mu\o\rho
=\mu\o ((\rho\o\now +\id)\o\out)^\pistar = ((\mu\o\rho\o\now +\id)\o\out)^\pistar = \out^\pistar.
\end{align*}
Finally, let us show that $\rho$ preserves strength, i.e.\ $\rho\o\tau = \tau\o(\id\times\rho)$.
Using assumption $\rho = ((\rho\o\now +\id)\o\out)^\pistar$, we have 
\begin{flalign*}
&&\tau\o(\id\times\rho) 
=&\; \tau\o(\id\times((\rho\o\now +\id)\o\out)^\pistar)&\\
&&=&\; ((\tau+\id)\o \dist\o (\id\times (\rho\o\now +\id)\o\out))^\pistar\\
&&=&\; ((\tau\o(\id\times\rho\o\now)+\id)\o \dist\o (\id\times\out))^\pistar\\
&&=&\; ((\rho\o\now+\id)\o \dist\o (\id\times\out))^\pistar\\
&&=&\; ((\rho\o\now +\id)\o\out)^\pistar\o\tau\\
&&=&\; \rho\o\tau.  
\end{flalign*}
where the second to last step is the characterization of $\tau$ from 
\cref{prop:D-props} and~\UNI.  

  {\bfseries\sffamily\ref{it:q3}.$\impl$\ref{it:q1}.}
Suppose that $\wave D$ extends to a strong monad $\wave\BBD$ and $\rho$ to a strong monad 
morphism. Let us define the search-algebra structure on $\wave DX$ as 
$\alpha_X = \mu_X\o\rho_{\wave DX}\c D\wave DX\to\wave DX$. The axioms of search-algebras follow 
by definition and by \cref{lem:rho-eq-lat}. By assumption, $\rho$ is a monad
morphism, in particular, $\rho\o\mu = \mu\o\rho\o D\rho$, hence 
$\rho_X\o\mu_X = \mu_X\o\rho_{DX}\o D\rho_X = \alpha_X\o D\rho_X$, i.e.\ $\rho_X$
is a $D$-algebra morphism.
\qed

\subsection{Proof of \cref{thm:elg}}
The difficult clause is {\bfseries\sffamily\ref{it:elg2}.} In order to show it,
some preparatory work is needed.
\begin{lemma}\label{lem:sigma-lat}
Suppose that the equivalent conditions of~\cref{thm:Drho} hold. Then $\Sigma=\wave D1$
is an internal distributive lattice with $\bot\c 1\to\wave D1$ as the bottom, 
$\top=\eta\c 1\to\wave D1$ as the top and $\land=\wave D\bang\o\hat\tau^\klstar\o
\tau\c\wave D1\times\wave D1\to\wave D1$.
% and the join $\lor\c\wave D1\times
%\wave D1\to\wave D1$, calculated from the universal property:
\end{lemma}
\begin{proof}
$T1$ is in fact a meet-semilattice for any equational lifting monad $\BBT$ with 
a point $\bot\c 1\to T\iobj$, which is easy to check. We proceed to define binary 
joins $\lor\c\Sigma\times\Sigma\to\Sigma$ in two steps. First, we define an 
auxiliary map $j\c\enat\times\Sigma\to\enat$ as a universal arrow from the diagram:
\begin{equation*}
\begin{tikzcd}
\enat\times D\nat 
  \ar[r,shift right=.35ex,"\id\times D\bang"']
  \ar[r,shift left=.75ex,"\id\times\hat\iota^\klstar"]
  \dar["{D[\zero\bang,t]}\o\zeta_{1,\nat}"']
&[8em]
\enat\times\enat
  \rar["\id\times\rho"] 
  \dar["{D\bang }\o\zeta_{1,1}"']
&[1em]
\enat\times\Sigma
  \dar[dotted,"j"]
\\[.8ex]
D\nat 
  \ar[r,shift right=.35ex,"D\bang"']
  \ar[r,shift left=.75ex,"\hat\iota^\klstar"]
&
\enat
  \rar["\rho"] & 
\Sigma
\end{tikzcd}
\end{equation*}
where $\zeta_{X,Y}\c DX\times DY\to D(X\times DY+DX\times Y)$ runs the arguments 
in parallel until one of the computations terminates and $t\c\enat\times\nat\to\nat$
returns the minimum of the first and the second arguments, concretely, by primitive 
recursion: 
\begin{align*}
t(p,\zero) = \zero,\qquad t(p,\suc(n)) =\case{(\out p)}{\inl\star\mto\zero}{\inr q\mto\suc (t(q,n))}.
\end{align*}
Then we define $\lor$ from the diagram
\begin{equation*}
\begin{tikzcd}
D\nat\times\enat 
  \ar[r,shift right=.35ex,"D\bang\times\id"']
  \ar[r,shift left=.75ex,"\hat\iota^\klstar\times\id"]
  \dar["\id\times\rho"']
&[8em]
\enat\times\enat
  \rar["\rho\times\id"]
  \dar["\id\times\rho"'] 
&[1em]
\Sigma\times\enat
  \dar["\id\times\rho"]
\\[.8ex]
D\nat\times\Sigma 
  \ar[r,shift right=.35ex,"D\bang\times\id"']
  \ar[r,shift left=.75ex,"\hat\iota^\klstar\times\id"]
&[8em]
\enat\times\Sigma
  \rar["\rho\times\id"]
  \ar[dr,"j"'] 
&[1em]
\Sigma\times\Sigma
  \dar[dotted,"\lor"]
\\[.8ex]
&
& 
\Sigma
\end{tikzcd}
\end{equation*}
where the fact that $j$ equalizes $D\bang\times\id$ and $\hat\iota^\klstar\times\id$
follows from the next calculation and the fact that $\rho\times\id$ is epic:
\begin{flalign*}
&& j\o(D\bang\times\id)\o(\id\times\rho) &\;= j\o(\id\times\rho)\o(D\bang\times\id)&\hspace{3cm}\\*
&&  &\;= j\o\rho\o(D\bang)\o\zeta_{1,1}\o(D\bang\times\id) \\
&&  &\;= j\o\rho\o(D\bang)\o\zeta_{1,1}\o\swap\o(D\bang\times\id) \\
&&  &\;= j\o\rho\o(D\bang)\o\zeta_{1,1}\o(\id\times D\bang)\o\swap \\
&&  &\;= j\o\rho\o(D\bang)\o\zeta_{1,1}\o(\id\times \hat\iota^\klstar)\o\swap \\
&&  &\;= j\o\rho\o(D\bang)\o\zeta_{1,1}\o\swap\o(\hat\iota^\klstar\times\id) \\
&&  &\;= j\o\rho\o(D\bang)\o\zeta_{1,1}\o(\hat\iota^\klstar\times\id) \\
&&  &\;= j\o(\id\times\rho)\o(\hat\iota^\klstar\times\id) \\
&&  &\;= j\o(\hat\iota^\klstar\times\id)\o(\id\times\rho)
\end{flalign*}
where we used the obvious fact that $\zeta$ is commutative. By definition, 
\begin{align*}
\lor\o(\rho\times\rho) = \rho\o (D\bang)\o\zeta_{1,1}
\end{align*}
which immediately entails that $\lor$ is commutative using the fact that $\rho\times\rho$
is epic. In a similar way, we can transfer the distributivity law 
$a\land (b\lor c) = (a\land b)\lor (a\land c)$. To that end, we use the following 
analogue of that law for $\BBD$:
\begin{equation*}
\begin{tikzcd}[column sep=20ex, row sep=normal]
\enat\times(\enat\times\enat)
    \ar[dd,"\id\times(D\bang)\o\zeta_{1,1}"']
    \rar["\brks{\id\times\fst,~\id\times\snd}"]&
(\enat\times\enat)\times(\enat\times\enat)
  \dar["((D\bang)\o\hat\tau^\klstar\o\tau)\times((D\bang)\o\hat\tau^\klstar\o\tau)"]\\
&\enat\times\enat
  \dar["(D\bang)\o\zeta_{1,1}"]\\
\enat\times\enat
  \rar["(D\bang)\o\hat\tau^\klstar\o\tau"]&
\enat 
\end{tikzcd}
\end{equation*}
This law essentially states that the operation of minimum on $\enat$ distributes 
over the operation of summation. By postcomposing this law with $\rho$ and using 
the property that $\rho\times(\rho\times\rho)$ is epi again, we obtain the 
desired distributivity law for $\Sigma$. In a similar way we obtain the absorption
law $a\lor (a\land b) = a$. 

The laws that we obtained are sufficient to show that $\lor$ is a join. On the one hand,
by the absorption law $a\land (a\lor b) = (a\land a)\lor (a\land b) = a\lor (a\land b) = a$, 
which is equivalent to $a\leq a\lor b$, and analogously, $b\leq a\lor b$. On the 
other hand, if $a\leq c$ and $b\leq c$, then $a\land c = a$, $b\land c= b$ and hence
$(a\lor b)\land c = (a\land b)\lor (a\land c) = a\lor b$, i.e.\ $a\lor b\leq c$. 
\end{proof}

\begin{lemma}\label{lem:frame}
Suppose that the equivalent conditions of~\cref{thm:Drho} hold, and additionally 
the coequalizer~\eqref{eq:D-quote} with $X=1$ is preserved by 
$(\argument)^\nat$. Then~$\Sigma$ is an internal $\omega$-frame, i.e.\ there is 
an $\omega$-join $\bigor\c\Sigma^\nat\to\Sigma$, which satisfying the frame 
distributive law $a\land\bigor_i b_i = \bigor_i (a\land b_i)$.
 
%Moreover, there 
%is a choice function $\xi\c\Sigma^\nat\to\wave D\nat$, such that 
%$\bigor\c\ev^\klstar\tau\o\brks{\id,\xi}\c \Sigma^\nat\to\Sigma$.
\end{lemma}
\begin{proof}
The approach is similar to that of Lemma~\ref{lem:sigma-lat}: we construct 
$\bigor\c\Sigma^\nat\to\Sigma$ by the following universal property
\begin{equation*}
\begin{tikzcd}[column sep=6ex, row sep=2ex]
(D\nat)^\nat 
  \ar[r,shift right=.35ex,"(D\bang)^\nat"']
  \ar[r,shift left=.75ex,"(\hat\iota^\klstar)^\nat"]
&[3em]
\enat^\nat
  \rar["\rho^\nat"] 
  \dar
&[1em]
\Sigma^\nat
  \dar[dotted, "\bigor"]
\\[1ex]
&
\enat
  \rar["\rho"] & 
\Sigma
\end{tikzcd}
\end{equation*}
where the morphism $\enat^\nat\to\enat$ converts a countable collections of 
infinite streams into a single stream, which is possible using the standard 
diagonalization argument.
\end{proof}
Recall the bounded iteration operator $(\argument)^\bistar$ from~\cref{def:biter}.
For a monad $\BBT$ with a constant $\bot\c 1\to TX$, we can instantiate this definition
in the Kleisli category $\BC_{\BBT}$ of $\BBT$, which yields $(\argument)^\bbistar\c
\BC(X, T(Y+X))\to\BC(X\times\nat, TY)$ for any $Y$ (because every $Y$ has a point 
$\bot\c 1\to TY$ in $\BC_{\BBT}$).

\begin{lemma}\label{lem:eql-elgot}
Let $\BBT$ be an equational lifting monad, and suppose that it is equipped with 
an operator $(\argument)^\istar\c\BC(X,T(Y+X))\to\BC(X,TY)$ that satisfies 
\FIX, \NAT, \UNI and~\STR. This induces a divergence constant $\bot = (\eta\o\inr)^\istar\c 1\to TX$.

Given $f\c X\to T(Y + X)$, and $g\c X\to TY$, (i) $f^{\bbistar}\appr f^\istar\o\fst$, 
and (ii) $f^{\bbistar}\appr g\o\fst$ implies~$f^\istar\appr g$.
\end{lemma}
\begin{proof}
The clauses~(i) and~(ii) are analogous to those of~\cref{thm:klee}.
The first clause is shown in exactly the same way. For the second clause, analogously,
we define $h\c X\times\nat\to T(\nat+X\times\nat)$ with the property that
\begin{align}\label{eq:bbstar}
f^\istar = (f^\bbistar)^\klstar\o\tau\o\brks{\id,h^\istar\o\brks{\id,\zero\o\bang}}
\end{align}
which will entail~(ii) using the same argument as in~\cref{thm:klee}, using~\NAT 
and~\STR. Here, we take 
\begin{align*}
h = T(\snd+\id)\o (T\ldist)\o\hat\tau\o(f\times\suc)
\end{align*}
This function runs just like $f$, but produces the number of computation steps 
as a final result. For the sake of brevity, let us denote $h^\istar\o\brks{\id,\zero\o\bang}\c X\to T\nat$
by $c$. First of all, we show the following auxiliary property:
\begin{align}\label{eq:bbstar2}
\dom f^\istar = \dom c.
\end{align}
It easily follows by~\NAT and~\UNI that $(T\bang)\o h^\istar = (T(\bang+\id)\o f)^\istar\o\fst$.
Then
$\dom (h^\istar\o\brks{\id,\zero\o\bang})
= \dom ((T\bang)\o h^\istar\o\brks{\id,\zero\o\bang}) 
= \dom ((T(\bang + \id)\o f)^\istar\o\fst\o\brks{\id,\zero\o\bang})
= \dom f^\istar$.

The main step of the proof is showing commutativity of the following diagram:
\begin{equation*}
\begin{tikzcd}[column sep=30ex, row sep=normal]
X
  \rar["f^\klstar\o(\dom f^\istar)"]
  \dar["\tau\o\brks{\id,c}"']& 
T(Y+X)
  \dar["T(\id+\tau\o\brks{\id,c})"]\\
X\times T\nat
  \rar["T(\fst+\eta)\o (T\ldist)\o\hat\tau^\klstar\o(\tau\o(f\times d))^\klstar"] & T(Y+X\times T\nat)               
\end{tikzcd}
\end{equation*}
where $d\c\nat\to T\nat$ is as follows:
\begin{align*}
d(\zero)  = \bot, &&
d(\suc n) = \eta(n).
\end{align*}
By \UNI, this entails the equation
\begin{align}\label{eq:bbstar1}
 (f^\klstar\o(\dom f^\istar))^\istar
 = (T(\fst+\eta)\o (T\ldist)\o\hat\tau^\klstar\o(\tau\o(f\times d))^\klstar)^\istar\o\tau\o\brks{\id,c}
\end{align}
Roughly, this says that iterating $f$ is equivalent to iterating $f$ and additionally checking 
that the counter that is initialized by $c$ and decreased at each iteration  
remains non-zero. The term $\dom f^\istar$ is needed to balance the effect of recalculating 
$c$ at each iteration out.
We have
\begin{flalign*}
&&T(\fst+\eta)\o& (T\ldist)\o\hat\tau^\klstar\o(\tau\o(f\times d))^\klstar\o\tau\o\brks{\id,c}\\
&&=&\;T(\fst+\eta)\o (T\ldist)\o\hat\tau^\klstar\o\tau\o(f\times d^\klstar)\o\brks{\id,c}\\
&&=&\;T(\fst+\eta)\o (T\ldist)\o\hat\tau^\klstar\o\tau\o\brks{f,d^\klstar\o c}\\
&&=&\;T(\fst+\eta)\o (T\ldist)\o\hat\tau^\klstar\o\tau\o\brks{f,d^\klstar\o [\eta,h^\istar]^\klstar\o h\o\brks{\id,\zero\o\bang}}\\
&&=&\;T(\fst+\eta)\o (T\ldist)\o\hat\tau^\klstar\o\tau\o\brks{f,d^\klstar\o [\eta\o\suc\o\zero\bang,h^\istar\o\brks{\id,\suc\o\zero\bang}]^\klstar\o f}\\
&&=&\;T(\fst+\eta)\o (T\ldist)\o\hat\tau^\klstar\o\tau\o\brks{f,[\eta\o\zero\bang,c ]^\klstar\o f}\\
&&=&\;T(\fst+\eta)\o (T\ldist)\o\hat\tau^\klstar\o\tau\o (\id\times[\eta\o\zero\bang,c ]^\klstar)\o \brks{f, f}\\
&&=&\;T(\fst+\eta)\o (T\ldist)\o\tau^\klstar\o\hat\tau\o (\id\times[\eta\o\zero\bang,c ]^\klstar)\o \brks{f, f}\\
&&=&\;T(\fst+\eta)\o (T\ldist)\o\tau^\klstar\o T(\id\times[\eta\o\zero\bang,c ]^\klstar)\o\hat\tau\o \brks{f, f}\\
&&=&\;T(\fst+\eta)\o (T\ldist)\o\tau^\klstar\o T(\id\times[\eta\o\zero\bang,c ]^\klstar)\o T\brks{\id,\eta}\o f\\
&&=&\;T(\fst+\eta)\o (T\ldist)\o\tau^\klstar\o T\brks{\id, [\eta\o\zero\bang,c ]}\o f\\
&&=&\;T(\fst+\eta)\o [(T\inl)\o\tau\o\brks{\id,\eta\o\zero\bang},(T\inr)\o\tau\o\brks{\id,c }]^\klstar\o f\\
&&=&\;[\eta\o\inl, (T\inr)\o (T\eta)\o \tau\brks{\id,c}]^\klstar\o f\\
&&=&\;[\eta\o\inl, (T\inr)\o T(\tau\o\brks{\id,c })\o\dom c ]^\klstar\o f\\
&&=&\;T(\id+\tau\o\brks{\id,c })\o [\eta\o\inl, (T\inr)\o\dom c ]^\klstar\o f\\
&&=&\;T(\id+\tau\o\brks{\id,c })\o [\eta\o\inl, (T\inr)\o\dom f^\istar]^\klstar\o f\\
&&=&\;T(\id+\tau\o\brks{\id,c })\o (\dom [\eta, f^\istar])^\klstar\o f\\
&&=&\;T(\id+\tau\o\brks{\id,c })\o f^\klstar\o \dom([\eta, f^\istar]^\klstar\o f)\\
&&=&\;T(\id+\tau\o\brks{\id,c })\o f^\klstar\o (\dom f^\istar).
\intertext{
The proof of~\eqref{eq:bbstar} now runs as follows:
}
&&f^\istar 
=&\; (f^\istar)^\klstar\o (\dom f^\istar)&\by{\ref{eq:rst1}}\\
&&=&\; (f^\istar\o \dom f^\istar)^\klstar\o (\dom f^\istar)&\by{\cref{lem:dag-dom}}\\
&&=&\; (f^\istar\o \dom f^\istar)^\klstar\o (\dom c) &\by{\eqref{eq:bbstar2}}\\
&&=&\; ((T(\fst+\eta)\o(T\ldist)\o\hat\tau^\klstar\o(\tau\o(f\times d))^\klstar)^\istar\o\tau\o\brks{\id,c})^\klstar\o (\dom c)&\by{\eqref{eq:bbstar1}}\\
&&=&\; ((T(\fst+\eta)\o(T\ldist)\o\hat\tau^\klstar\o(\tau\o(f\times d))^\klstar)^\istar)^\klstar\\
&&&\hspace{4.5cm} T(\tau\o\brks{\id,c})\o\dom(\tau\o\brks{\id,c}) \\
&&=&\; ((T(\fst+\eta)\o(T\ldist)\o\hat\tau^\klstar\o(\tau\o(f\times d))^\klstar)^\istar)^\klstar\o (T\eta)\o\tau\o\brks{\id,c} &\by{\cref{lem:dom-eta}}\\
&&=&\; ((T(\fst+\eta)\o(T\ldist)\o\hat\tau^\klstar\o(\tau\o(f\times d))^\klstar)^\istar\o\eta)^\klstar\o\tau\o\brks{\id,c} \\
&&=&\; ((T(\fst+\id)\o(T\ldist)\o\hat\tau^\klstar\o\tau\o(f\times d))^\istar)^\klstar\o\tau\o\brks{\id,c}&\by{\UNI}\\
&&=&\; (f^\bbistar)^\klstar\o\tau\o\brks{\id,c}.
\end{flalign*}
Here, the last step requires checking that $f^\bbistar=(T(\fst+\id)\o(T\ldist)\o\hat\tau^\klstar\o\tau\o(f\times d))^\istar$,
which is easy, since $f^\bbistar$ is defined by primitive recursion.
\end{proof}
Let us call a pair $(M\in |\BC|, M_{\appr}\subseteq M\times M)$ a \emph{(internal) 
regular poset} if $M_{\appr}$ is a regular monic, i.e.\ an equalizer of some pair $l,r\c M\times M\to E$
and the relation $\appr$ on $\Hom(X,M)$, defined by putting $f\appr g$ iff 
$l\o\brks{f,g} = r\o\brks{f,g}$, is a partial order for every $X$. By definition,
$f\appr g$ iff $\brks{f,g}$ factors through $M_{\appr}\subseteq M\times M$. Because 
of the equalizer property, this factorization is unique. Our notion of a regular 
poset is thus an instance of Barr's notion of an (internal) poset~\cite{Barr90} to 
which we added the regularity requirement. We will follow Barr's work further 
on in internalizing further relevant notions of order theory. First, observe the 
following. 
\begin{lemma}
For every $X\in |\BC|$, the equalizer of the pair $\fst, (\snd^\klstar)\o\hat\tau\c TX\times TX\to TX$
defines a regular poset.
\end{lemma}
\begin{proof}
Obvious by definition of the enrichment of $\BC_\BBT$.
\end{proof}
Given a regular poset $(M\in |\BC|, M_{\appr}\subseteq M\times M)$, we define
the object of monotone sequences~$M^\omega$ from the pullback:
\begin{equation*}
\begin{tikzcd}[column sep=normal, row sep=normal]
M^\omega
  \pbk
  \dar
  \ar[rr,hook] && 
M^\nat
  \dar["\brks{\id,M^{\suc}}"]\\
(M_\appr)^\nat
  \rar[hook] & 
(M\times M)^\nat
  \rar[iso] & M^\nat\times M^\nat             
\end{tikzcd}
\end{equation*}
Now, a regular poset $(M\in |\BC|, M_{\appr}\subseteq M\times M)$ is an internal $\omega$-cpo
if it is equipped with an operation $\ij\c M^\omega\to M$, such that for every 
$f\c X\to M^\omega$, $\ij\o f\c X\to M$ is the least upper bound for 
$\hat f = \uncurry(X\xto{f} M^\omega\ito M^\nat)$ in $\BC(X,M)$, meaning that 
(i) $\hat f\appr(\ij\o f)\o\fst$ and (ii) $\hat f\appr g\o\fst$ implies 
$\ij\o f\appr g$ for any $g\c X\to M$.
\begin{proposition}\label{prop:elg-form-join}
An equational lifting monad $\BBT$ is a strong Elgot monad if\/ $T\nat$ is an internal 
$\omega$-cpo.
\end{proposition}
\begin{proof}
First, given $f\c X\to T(\nat + X)$, we define $f^\istar\c X\to T\nat$ as the 
least fixpoint of the map $[\eta,\argument]^\klstar\o f\c\BC(X,T\nat)\to\BC(X,T\nat)$.
Then we extend $(\argument)^\istar$ to all $f\c X\to T(Y+X)$ using the 
formula~\eqref{eq:bbstar} from~\cref{lem:eql-elgot}. \FIX, \NAT, \UNI and~\STR
follow from this definition, which enables the characterization of~\cref{lem:eql-elgot},
and then~\COM follows analogously to an existing argument~\cite[Theorem 5.8]{GoncharovSchroderEtAl18}.  
\end{proof}
Finally, we proceed with the proof of~\cref{thm:elg}.

{\bfseries\sffamily\ref{it:elg1}.} 
In extensive categories coequalizers are preserved by $(\argument+1)$, hence, by 
assumption,~\eqref{eq:D-quote} is preserved by $(\argument+1)^\nat$. Recall 
that $DX$ is a retract of $(X+1)^\nat$ naturally in $X$. It is easy to verify 
that coequalizers are preserved by retractions, hence~\eqref{eq:D-quote} is 
preserved by $D$. 

{\bfseries\sffamily\ref{it:elg2}.} 
The claim follows from~\cref{prop:elg-form-join} and~\cref{lem:eql-elgot}. We 
only need to extend $\wave D\nat$ to an $\omega$-cpo. The effectiveness assumption
implies that $\wave D\nat$ is isomorphic to $\Sigma^\omega$. By~\cref{lem:frame},
we have internal joins on $\Sigma$, which extend to $\Sigma^\omega$ pointwise.
%
%\begin{equation*}
%\begin{tikzcd}[column sep=6ex, row sep=2ex]
%\wave D\nat
%  \rar["\curry({[\eta,\bot]^\klstar}\o \oname{eq}_{\nat}^\klstar\o\hat\tau)"]&[15ex]
%\Sigma^\nat 
%  \ar[r,shift right=.35ex,"(D\bang)^\nat"']
%  \ar[r,shift left=.75ex,"(\hat\iota^\klstar)^\nat"]
%&[3em]
%\Sigma^\nat
%\end{tikzcd}
%\end{equation*}
%
%
%
%{\bfseries\sffamily\ref{it:elg3}.} 
%%
%This clause largely amounts to showing that monads whose Kleisli categories are
%enriched $\omega$-cpos and strict continuous morphisms are Elgot. This goes by a 
%routine internalization of the existing argument~\cite[Theorem 5.8]{GoncharovSchroderEtAl18}.
%Since every Elgot monad $\BBT$ is pre-Elgot, by~\cref{thm:K-pre-Elgot}, we obtain
%a universal map from $\wave\BBD$ to $\BBT$.
%
\qed
\end{document}